\documentclass[a4paper,11pt]{article}
\usepackage{amssymb}
\usepackage{amsfonts}
\usepackage{amsmath}
\usepackage[ignoreall,tmargin = 20mm, lmargin =20mm, rmargin = 20mm, bmargin = 20mm, noheadfoot]{geometry}
\usepackage[onehalfspacing]{setspace}
\usepackage[round]{natbib}
\usepackage[pdftex]{graphicx}
\usepackage{epstopdf}
\usepackage{rotating,booktabs}
\usepackage{subfigure}
\usepackage{color}
\usepackage{algorithm}
\usepackage{algorithmic}
\usepackage{multirow}
\newtheorem{theorem}{Theorem}
\newtheorem{assumption}{Assumption}

\newtheorem{condition}{Condition}

\newtheorem{definition}{Definition}

\newtheorem{lemma}{Lemma}

\newtheorem{step}{Step}
\newenvironment{proof}[1][Proof]{\textbf{#1.} }{\ \rule{0.5em}{0.5em}}

\begin{document}
\doublespacing
\title{\textbf{Testing Forecast Accuracy of Expectiles and Quantiles with the Extremal Consistent Loss Functions}\thanks{We thank seminar participants in 2016 macroeconometric modelling workshop (Academia Sinica), the 1st International Conference on Econometrics and Statistics (HKUST), 2017 European meeting of the Econometric society (University of Lisbon), 10th Annual Meeting of Taiwan Econometric Society, 2018 IAAE Annual Conference (UQAM), CRETA workshop (National Taiwan University) and National Chengchi University for helpful comments. 
 }
}

\author{Yu-Min Yen\thanks{Associate professor, Department of International Business, National Chengchi University, 64, Section 2, Zhi-nan Road, Wenshan, Taipei 116, Taiwan. Email: \texttt{yyu\_min@nccu.edu.tw}} and Tso-Jung Yen\thanks{Assistant research fellow, Institute of Statistical Science, Academia Sinica. Address: 128 Academia Road, Section 2, Nankang, Taipei 11529, Taiwan. E-mail: \texttt{tjyen@stat.sinica.edu.tw}.} 
}
\maketitle
\begin{abstract}Forecast evaluations aim to choose an accurate forecast for making decisions by using loss functions. However, different loss functions often generate different ranking results for forecasts, which complicates the task of comparisons. In this paper, we develop statistical tests for comparing performances of forecasting expectiles and quantiles of a random variable under consistent loss functions. The test statistics are constructed with the extremal consistent loss functions of \citet{EGJK_2016}. The null hypothesis of the tests is that a benchmark forecast at least performs equally well as a competing one under all extremal consistent loss functions. It can be shown that if such a null holds, the benchmark will also perform at least equally well as the competitor under all consistent loss functions. Thus under the null, when different consistent loss functions are used, the result that the competitor does not outperform the benchmark will not be altered. We establish asymptotic properties of the proposed test statistics and propose to use the re-centered bootstrap to construct their empirical distributions. Through simulations, we show the proposed test statistics perform reasonably well. We then apply the proposed method on (1) re-examining abilities of some often-used predictors on forecasting risk premium of the S\&P500 index; (2) comparing performances of experts' forecasts on annual growth of U.S. real gross domestic product; (3) evaluating performances of estimated daily value at risk of the S\&P500 index.  
\clearpage
\end{abstract}
\section{Introduction}

When evaluating performances of a benchmark and a competing forecasts for a target functional of a random variable (e.g., conditional expectation), typically we can compare expected values of a loss function  (e.g., the squared error loss) evaluated with the two forecasts and the random variable. We say that the competitor outperforms the benchmark under a loss function if the expected value of the loss function for the former is lower than that for the latter. There are many loss functions can be chosen for comparing forecast performances. Such choices may reflect forecast users' concerns on cost of wrong forecasts in the future \citep{Granger_1969, GN_1986}. 
For example, when controlling downside risk of purchasing an asset, one may focus on negative forecast errors\footnote{We follow the convention to define a forecast error as realization of the random variable minus the forecast.} of the asset's conditional expected return rather than their positive counterparts. In this situation, it would be suitable to choose a loss function that penalizes more on the negative forecast errors.

An important guideline for choosing a loss function for evaluating forecasts is that the loss function should be consistent \citep{Gneiting_2011,Patton_2015}. If the target functional
can be obtained by minimizing expectation of a certain loss function, then we say the loss function is a consistent loss function for the target functional. If a target functional is the only one minimizer of the expectation of a consistent loss function, then this target functional is called an elicitable target functional and the loss function is called strictly consistent (for the elicitable target functional). 

The criterion of consistency reduces the set of loss functions for comparing forecast performances. However, for an elicitable target functional, there may still exist infinitely many corresponding consistent loss functions. \citet{Patton_2015} shows that using different consistent loss functions may yield different ranking results for two forecasts, unless (1) they are issued by using correctly specified models, and (2) the information used for generating one forecast is a subset of that used for generating the other. However, conditions (1), (2) or both often do not hold in practice. If either condition (1) or (2) is violated, 
or estimated forecast models have estimation errors, 
then using different consistent loss functions may yield different ranking results, which complicates the task of evaluating forecast performances. 



In this paper we develop statistical tests for comparing performances of forecasting
expectiles and quantiles of a random variable under consistent loss
functions. The proposed tests can alleviate the aforementioned
difficulty when different consistent loss functions are used on
evaluating forecast performances. The test statistics are constructed
by using the extremal consistent loss functions of \citet{EGJK_2016}. The null hypothesis of the tests is that a benchmark forecast at least performs equally well as a competing one under all extremal consistent loss functions. It can be shown that if such a null holds, the benchmark will also at least performs equally well as the competitor under all consistent loss functions, regardless whether the aforementioned
conditions (1) or (2) holds or not. Thus under the null hypothesis, using different
consistent loss functions will not alter the result that the competitor does not outperform the benchmark. On contrary, if this null hypothesis is rejected, we may see that the competitor outperforms the benchmark under certain consistent loss functions. 

The proposed tests may be suitable as a first-step check when the consistent loss function used to generate the competing forecast is unknown, such as that from a survey. In this situation, sometimes it is hard to fairly judge whether one forecast outperforms the other under a chosen consistent loss function. 
With the proposed test, the forecasts will have a fair chance to demonstrate their ability regardless which consistent loss function is used, since the proposed test verifies whether one forecast outperforms the other over all possible consistent loss functions. 

\citet{EGJK_2016} use the extremal consistent loss functions to graphically compare performances of two forecasts for the expectiles and quantiles. They term such a graph as a Murphy diagram. While the Murphy diagram is a useful tool, it only provides graphical evidence of the performance differences but gives no formal statistical justification. Our proposed tests can be viewed as formal statistical tests for testing such performance differences uniformly. In addition, our proposed tests are not like traditional forecast accuracy tests, such as the Diebold-Marino test \citep{DM_1995}, which use only one consistent loss function at a time. Rather our proposed tests seek to detect the performance differences between two forecasts over infinitely many possible consistent loss functions, which may be particularly important when the loss function used to generate the competing forecast is unknown.

We establish theoretical properties of the proposed test statistics under some mild conditions. \citet{West_1996} shows that if a loss function has some regular properties, it can be consistently estimated and the estimate is asymptotically normally distributed. However, the extremal consistent loss functions do not possess all the regular properties mentioned in \citet{West_1996}. In addition, the proposed test statistics have a form of Kolmogorov-Smirnov type. Thus analyzing theoretical properties of our proposed test statistics relies on using non-traditional techniques. 
We show that the test statistics have a non-degenerate asymptotic distribution related to a mean zero Gaussian process. To efficiently conduct the tests, we propose to use the re-centered bootstrap to construct empirical distributions of the test statistics. We then show validity of the bootstrap scheme by proving empirical distributions of the re-centered bootstrap test statistics converge to distributions of the re-centered sample test statistics.

We next conduct intensive simulations to understand how the proposed test statistics perform with finite samples. In the first simulation, 
we design a situation in which two forecasts for a conditional expectation perform equally well under the square error loss but differently under the exponential Bregman loss. In this situation, if we use the Diebold Marino test statistic with the squared error loss, we have a low probability to reject the null and it is unlikely to identify which forecast performs better than the other under the exponential Bregman loss. However, our proposed test statistic has a high probability to correctly detect such performance differences in this case. We further show that the proposed test statistics with the re-centered bootstrap work well in more realistic situations. 

We apply the proposed tests on three empirical studies. We first re-examine abilities of some often-used predictors on forecasting risk premium of the S\&P500 index. 
We find that evidence for these predictors outperforming historical average of excess returns is weak. 
We also compare performances of experts' forecasts on annual growth of U.S. real gross domestic product (RGDP) and find that the mean forecast of experts performs better than or at least equally well as an individual forecast. Finally, we evaluate different models' performances of forecasting daily value at risk (VaR) of the S\&P500 index and find that the CAViaR type models \citep{EM_2004} performs better than or at least equally well as the other two simple methods. 
All these empirical results are robust to choices of different consistent loss functions.  

Loss functions can be functions of forecast errors and other parameters. 
Such loss functions, together with some mild restrictions, are called the generalized loss functions \citep{Granger_1969,Granger_1999} and some relevant important results were derived, see \citet{EKT_2004}, \citet{DS_2015} and \citet{JCS_2016}. 
The class of the generalized loss functions nests some (but not all) consistent loss functions of forecasting the expectiles and quantiles as special cases, for example, the squared error loss and lin-lin (tick) loss. But some loss functions belonging to the class are not consistent loss functions for the expectiles and quantiles forecasts, for example, linex loss function of \citet{Varian_1975}
and double exponential loss function of \citet{Granger_1999}. Thus our proposed tests may be a complementary to forecast accuracy tests based on such a class of loss functions. 


Recently \citet{EK_2017} also propose tests to compare forecasts on the expectiles and quantiles based on the extremal consistent loss functions of \citet{EGJK_2016}. Our proposed method has several differences from theirs. First, empirical p-values of their test statistics are constructed by sign randomization and consequently have different theoretical and empirical properties than those of ours. 
More importantly, they test hypotheses of conditional performances of the forecasts, but our hypotheses focus on the unconditional performances.

The rest of the paper is organized as follows. In Section 2 we review
concepts of consistent loss functions and the extremal consistent
loss functions of \citet{EGJK_2016}. In Section 3 we introduce the proposed tests and establish their theoretical properties, and illustrate how to use the re-centered bootstrap to construct their empirical distributions for statistical inferences. In Section 4 we conduct simulation studies for examining performances of the test statistics in various situations. In Section 5 we use the proposed tests on the three empirical applications. 
Section 6 is for conclusions.

\section{Consistent loss functions for point forecasts}

Let $L\left(x,y\right)$ denote a loss function for evaluating a forecast for a target functional of a random variable. Following convention, we let the first argument of $L(x,y)$ be the forecast and the second argument be the random variable. For all pairs $\left(x,y\right)$, assume $L\left(x,y\right)\geq0$ and if $x=y$, $L\left(x,y\right)=0$. 
Let $\mathcal{F}$ denote a class of probability functions on a closed subset $D\subset\mathbb{R}$ and $F$ be an element in $\mathcal{F}$. Let $\lambda:\mathcal{F}\mapsto \mathbb{R}$ denote a statistical functional which maps $F\in \mathcal{F}$ to $\mathbb{R}$. The loss function $L(x,y)$ is consistent for a statistical functional $\lambda(F)$ if $ E_{F}\left[L\left(\lambda\left(F\right),Y\right)\right]\leq E_{F}\left[L\left(x,Y\right)\right] $
for all $F\in\mathcal{F}$, $x\in\mathbb{R}$ and a random variable $Y\in D$ and $Y\sim F$. 
The loss function $L(x,y)$ is \textit{strictly} consistent for the functional $\lambda(F)$ if 
\begin{equation}
\lambda\left(F\right)=\arg\min_{x}E_{F}\left[L\left(x,Y\right)\right]
\label{consistent}
\end{equation}
and $E_{F}\left[L\left(\lambda\left(F\right),Y\right)\right]=E_{F}\left[L\left(x,Y\right)\right]$ implies $x=\lambda\left(F\right)$. If $L(x,y)$ is a strictly consistent loss function and $\lambda(F)$ satisfies (\ref{consistent}), then $\lambda(F)$ is called elicitable. 

\subsection{Consistent loss functions for expectiles and quantiles}
The functionals $\lambda(F)$ we are interested in this paper are conditional expectiles and conditional quantiles.\footnote{We use the term ``conditional'' here since in forecast, the amount of information we can use is only up to current period and is not unlimited. Thus $F$ is a distribution conditioning on a limited amount of information and $\lambda(F)$ is a conditional statistical functional.}
The expectile of a random variable $Y\sim F$ at level $\alpha\in\left(0,1\right)$, called the $\alpha-$expectile of $Y$, can be obtained by solving $t$ in the following equation
\[
\frac{E_{F}\left[\left(t-Y\right)_{+}\right]}{E_{F}\left[\left(Y-t\right)_{+}\right]} = \frac{\alpha}{1-\alpha}.
\]
When $\alpha=0.5$, it is easy to see that $t$ is expectation of $Y$ under the distribution function $F$, $E_{F}\left[Y\right]$. \citet{Savage_1971} shows that a consistent loss function for an expectation of a random variable, denoted by $L^{E}\left(x,y\right)$, can be expressed as the following Bregman type function
\begin{equation}
L^{E}\left(x,y\right)=\phi\left(y\right)-\phi\left(x\right)-\phi^{\prime}\left(x\right)\left(y-x\right),\label{scoring_expectation}
\end{equation}
where $\phi(.)$ is a convex function and $\phi^{\prime}(.)$ is its subgradient. The consistent loss function $L^{E}(x,y)$ in (\ref{scoring_expectation}) nests some frequently used loss functions as special cases. With different specifications of $\phi(.)$ in (\ref{scoring_expectation}), we list examples of $L^{E}(x,y)$ in Table \ref{table1}, which include the squared error loss and the QLIKE loss \citep{Patton_2011}. Another interesting case in Table \ref{table1} is when $\phi\left(x\right)=x\log x+\left(1-x\right)\log\left(1-x\right)$
for $x\in\left[0,1\right]$, 
and this kind of consistent loss function 
is associated with the negative log likelihood for the logistic regression estimation.

For the $\alpha-$expectile of a random variable, \citet{Gneiting_2011} shows that the corresponding consistent loss function, denoted by $L_{\alpha}^{E}\left(x,y\right)$, can be expressed as
\begin{eqnarray}
L_{\alpha}^{E}\left(x,y\right)&=&\left|1\left\{y<x\right\}-\alpha\right|\times L^{E}(x,y)\nonumber\\
&=& \left|1\left\{y<x\right\}-\alpha\right|\times\left[\phi\left(y\right)-\phi\left(x\right)-\phi^{\prime}\left(x\right)\left(y-x\right)\right].
\label{scoring_expectile}
\end{eqnarray}
Combining with different forms of $L^{E}$ in Table \ref{table1}, we can obtain various loss functions for the $\alpha-$expectile forecasts. For example, if we set $\phi\left(t\right)=t^{2}$, $L_{\alpha}^{E}\left(x,y\right)$
becomes 
the asymmetric squared error loss for estimating the $\alpha-$expectile regression of \citet{NP_1987}. The $\alpha-$expectile regression can be applied to forecast the expectile-based Value at Risk (EVaR), which measures the relative cost of the expected margin shortfall. \citet{KYH_2009} show that the EVaR is a useful alternative risk measurement for extreme loss to the quantile based VaR.


The $\alpha-$quantile of a random variable $Y\sim F$, denoted by $q\left(\alpha\right)$, is defined as
\begin{equation}
q\left(\alpha\right):=\inf\left\{\tau:P\left(Y\leq \tau\right)\geq\alpha\right\},
\label{quantile}
\end{equation}
where $P(.)$ is the probability of $Y$. If the distribution function $F(y)$ is strictly monotonically
increasing and continuous, then $q\left(\alpha\right)=F^{-1}\left(\alpha\right)$. Quantile forecasts are important in risk managements. For example, the value at risk (VaR) are often constructed by using conditional quantile forecasts of an asset's return. 

Let $L^{Q}\left(x,y\right)=\zeta\left(x\right)-\zeta\left(y\right)$, where $\zeta(.)$ is a nondecreasing function. \citet{Thomson_1979} and \citet{Saerens_2000} show that a consistent loss function for the $\alpha-$quantile of a random variable, denoted by $L^{Q}_{\alpha}(x,y)$, can be expressed as 
\begin{eqnarray}
L^{Q}_{\alpha}(x,y)&=& (1\{y<x\}-\alpha)\times L^{Q}\left(x,y\right) \nonumber\\
&=& (1\{y<x\}-\alpha)\times\left[\zeta(x)-\zeta(y)\right].
\label{scoring_quantile}
\end{eqnarray}
The right hand side of (\ref{scoring_quantile}) is the generalized piecewise linear (GPL) function of order $\alpha$. Several examples of $L^{Q}(x,y)$ are listed in Table \ref{table2}. When $\zeta(t)=t$, $L^{Q}_{\alpha}(x,y)=(1\{y<x\}-\alpha)\left(x-y\right)$ is the lin-lin or asymmetric piecewise linear loss function, which can be used to estimate the $\alpha-$quantile regression \citep{KB_1978}. Another interesting case of $L^{Q}_{\alpha}(x,y)$ is the scaled lin-lin loss 
by setting $\zeta(t)=t/\alpha$ \citep{HE_2014}. When $Y$ is a continuous random variable, \citet{HE_2014} show that under distribution $F$, the expected scaled lin-lin loss with $x=q(\alpha)$ is   
\begin{equation}
E_{F}\left[\left(1\left\{ Y<q\left(\alpha\right)\right\} -\alpha\right)\left(\frac{q\left(\alpha\right)}{\alpha}-\frac{Y}{\alpha}\right)\right] =  E_{F}\left[Y\right]-\frac{1}{\alpha}E_{F}\left[1\left\{ Y<q\left(\alpha\right)\right\} Y\right].
\label{scaled_LinLin}
\end{equation}
The second term of right hand side of (\ref{scaled_LinLin}) is the expected shortfall of $Y$. Thus equation (\ref{scaled_LinLin}) provides a way to estimate the expected shortfall by subtracting the minimized expected scaled lin-lin loss from the expectation of $Y$. 
\subsection{Extremal consistent loss functions}
In this subsection we introduce the extremal consistent loss functions of \citet{EGJK_2016} for the $\alpha-$expectile and $\alpha-$quantile of a random variable. Let $\mathcal{L}_{\alpha}^{E}$ denote the class of consistent loss functions for the $\alpha-$expectile which admits the form of (\ref{scoring_expectile}). \citet{EGJK_2016} show that every consistent loss function $L_{\alpha}^{E}(x,y)\in\mathcal{L}_{\alpha}^{E}$
can be represented as 
\begin{equation}
L_{\alpha}^{E}\left(x,y\right)=\int_{-\infty}^{\infty}L_{\theta,\alpha}^{E}\left(x,y\right)dH\left(\theta\right),\label{scoring_expectile1}
\end{equation}
where $L_{\theta,\alpha}^{E}\left(x,y\right)$ is the extremal consistent loss function for the $\alpha-$expectile, which is given by 
\begin{equation}
L_{\theta,\alpha}^{E}\left(x,y\right) =  \left|1\left\{y<x\right\}-\alpha\right|\left[ \left(y-\theta\right)_{+}-\left(x-\theta\right)_{+}-1\left\{\theta<x\right\}\left(y-x\right)\right]. 
\end{equation}
It can be shown that $0\leq L_{\theta,\alpha}^{E}(x,y)\leq \max(\alpha,1-\alpha)\times|y-x|$. It is also easy to see that $L_{\theta,\alpha}^{E}(x,y)\in\mathcal{L}_{\alpha}^{E}$
if we set $\phi\left(t\right)=\left(t-\theta\right)_{+}$ in (\ref{scoring_expectile}). 
The representation of (\ref{scoring_expectile1}) states that every
consistent loss function for the $\alpha-$expectile is
a weighted sum of the extremal consistent loss function $L^{E}_{\theta,\alpha}(x,y)$. The representation of (\ref{scoring_expectile1})
is a Choquet-type mixture representation in functional analysis \citep{EGJK_2016}, in which $H(.)$ is a unique non-negative mixing measure which satisfies $dH\left(\theta\right)=d\phi^{\prime}\left(\theta\right)$
for $\theta\in\Theta\subseteq\mathbb{R}$, where $\phi^{\prime}(.)$ is the left-hand
derivative of the convex function $\phi(.)$ in (\ref{scoring_expectile}) and $\Theta$ is a bounded subset of $\mathbb{R}$. Also $\left(1-\alpha\right)[H\left(x\right)-H\left(y\right)]=\partial L_{\alpha}^{E}\left(x,y\right)/\partial y$ for $x>y$, where $\partial L_{\alpha}^{E}\left(x,y\right)/\partial y$ denotes the left-hand derivative with respect to $y$.

For the $\alpha-$quantile, let $\mathcal{L}_{\alpha}^{Q}$ denote the class of consistent loss functions for the $\alpha-$quantile which admits the form of (\ref{scoring_quantile}). Like the case of $L^{E}_{\alpha}$, \citet{EGJK_2016} show that every consistent loss function $L_{\alpha}^{Q}(x,y)\in\mathcal{L}_{\alpha}^{Q}$ also has a Choquet-type mixture representation 
\begin{equation}
L_{\alpha}^{Q}\left(x,y\right)=\int_{-\infty}^{\infty}L_{\theta,\alpha}^{Q}\left(x,y\right)dG\left(\theta\right),\label{scoring_quantile1}
\end{equation} 
where $L^{Q}_{\theta,\alpha}(x,y)$ is the extremal consistent loss function for the $\alpha-$quantile, which is given by
\begin{equation}
	L_{\theta,\alpha}^{Q}\left(x,y\right) =  \left(1\left\{ y<x\right\} -\alpha\right)\left(1\left\{ \theta<x\right\} -1\left\{ \theta<y\right\} \right).
\end{equation}  
It can be shown that $0\leq L_{\theta,\alpha}^{Q}(x,y)\leq \max(\alpha,1-\alpha)$. It also easy to see that  $L_{\theta,\alpha}^{Q}(x,y)\in\mathcal{L}_{\alpha}^{E}$ since it is the consistent loss function when $\zeta\left(t\right)=1\{\theta<t\}$ in (\ref{scoring_quantile}). 
In (\ref{scoring_quantile1}), $G(.)$ is a unique non-negative mixing measure which satisfies $dG\left(\theta\right)=d\zeta\left(\theta\right)$
for $\theta\in\Theta\subseteq\mathbb{R}$, where $\zeta(.)$ is the nondecreasing function in (\ref{scoring_quantile}) and $\Theta$ is a bounded subset of $\mathbb{R}$. Also $\left(1-\alpha\right)[G\left(x\right)-G\left(y\right)]= L_{\alpha}^{Q}\left(x,y\right)$ for $x>y$.
\subsection{Accuracy of the representations}
The representations (\ref{scoring_expectile1}) and (\ref{scoring_quantile1}) can be used to numerically approximate the consistent loss functions for the $\alpha-$expectile and $\alpha-$quantile forecasts. An accurate approximation from the representation is crucial for constructing the proposed test statistic. In this subsection we compare numerical values of several consistent loss functions with those obtained from using the representations of (\ref{scoring_expectile1}) and (\ref{scoring_quantile1}). For the $\alpha-$expectile, 
we choose the exponential (non-homogeneous) Bregman loss 
and the homogeneous Bregman loss 
for the comparisons. 
For the former, $dH\left(\theta\right)=\exp\left(a\theta\right)d\theta$
and for the latter, $
dH\left(\theta\right)=\left(b\left(b-1\right)\left|\theta\right|^{b-2}+b\delta\left(\theta\right)\left|x\right|^{b-1}\right)d\theta$, where $\delta\left(\theta\right)$ is the Dirac function. 
For the $\alpha-$quantile, 
we choose the lin-lin loss and the homogeneous (power) loss with order $c=2$ for the comparisons. 
For the former, $dG(\theta)=1$ and for the latter, $dG(\theta)=2\theta$. 

Let $N(\mu,\sigma^{2})$ denote the normal distribution with mean $\mu$ and variance $\sigma^{2}$ and $\chi(\kappa)$ denote the chi-square distribution with degree of freedom $\kappa$. 
For the $\alpha-$expectile, the simulated data for each comparison are 1000 pairs of $X\sim N(0,1)$ and 
$Y\sim N(0,1)$. For the $\alpha-$quantile, in the case of the lin-lin loss, the simulated data for each comparison are 1000 pairs of $X\sim N(0,1)$ and $Y\sim N(0,1)$. In the case of the homogeneous loss with order $c=2$, the data for each comparison are 1000 pairs of $X\sim \chi^{2}\left(1\right)$ and $Y\sim \chi^{2}\left(1\right)$. 

With pairs $(X,Y)$, we numerically evaluate integrals of (\ref{scoring_expectile1}) and (\ref{scoring_quantile1}) with the Trapezoid method. We then compare the numerical integrals with the corresponding consistent loss functions directly calculated with pairs $(X,Y)$. In Figure \ref{figure1}, left panel shows comparison results for the exponential Bregman loss with $\alpha=0.5$, $a=-1$, 0.3 and 1. Right panel shows those for the homogeneous Bregman loss with $\alpha=0.5$, $b=1.5$, 2 and 3. 
In Figure \ref{figure2}, left panel shows the comparison results for the lin-lin loss and right panel shows those for the homogeneous loss with $\alpha = 0.01$, 0.05 and 0.5. The solid line in each plot is a 45 degree line. From each figure, it can be seen that all pairs of value of the consistent loss function and that obtained from using the representation of (\ref{scoring_expectile1}) (or (\ref{scoring_quantile1})) almost lie on the 45 degree line, which suggests that the two are virtually identical and the representation of (\ref{scoring_expectile1}) (or (\ref{scoring_quantile1})) works well on approximating the corresponding consistent loss function.
\section{Forecast accuracy tests with the extremal consistent loss functions}

In this section we introduce the proposed tests and test statistics for comparing forecast
accuracy of the $\alpha-$expectile or $\alpha-$quantile under all consistent loss functions. 
Let $X_{1}$ be a benchmark and $X_{2}$ be a competing forecasts for the $\alpha-$expectile or the $\alpha-$quantile of a random variable $Y$. For forecasting the $\alpha-$expetile, under a consistent loss function $L_{\alpha}^{E}\in\mathcal{L}_{\alpha}^{E}$, we say that $X_{1}$ at least performs equally well as $X_{2}$ if 
\begin{equation}
E\left[L_{\alpha}^{E}\left(X_{1},Y\right)\right]\leq E\left[L_{\alpha}^{E}\left(X_{2},Y\right)\right].\label{outperformance}
\end{equation}
With the representation of (\ref{scoring_expectile1}), (\ref{outperformance}) can be expressed as 
\begin{equation}
\int_{-\infty}^{\infty}E\left[L_{\theta,\alpha}^{E}\left(X_{1},Y\right)\right]dH\left(\theta\right)\leq\int_{-\infty}^{\infty}E\left[L_{\theta,\alpha}^{E}\left(X_{2},Y\right)\right]dH\left(\theta\right).\label{outperformance1}
\end{equation}
Since for every $H\left(.\right)$, $dH\left(\theta\right)=d\phi^{\prime}\left(\theta\right)$ is
nonnegative for all $\theta\in\Theta$ and the functional form of the extremal consistent loss $L_{\theta,\alpha}^{E}(x,y)$ is independent of $H\left(.\right)$, a sufficient condition for $X_{1}$
at least performing equally well as $X_{2}$ as the $\alpha-$expectile
forecast under all $L_{\alpha}^{E}\in\mathcal{L}_{\alpha}^{E}$ is that $E\left[L_{\theta,\alpha}^{E}\left(X_{1},Y\right)\right]\leq E\left[L_{\theta,\alpha}^{E}\left(X_{2},Y\right)\right]$
holds for all $\theta$. 
Thus given $\alpha$, to see whether such a sufficient condition holds, 
we may test the following null hypothesis 
\begin{equation}
H_{0}:E\left[L_{\theta,\alpha}^{E}\left(X_{1},Y\right)\right]\leq E\left[L_{\theta,\alpha}^{E}\left(X_{2},Y\right)\right]\text{ for all }\theta.\label{null}
\end{equation}
If the null of (\ref{null}) is rejected, it indicates that for forecasting the $\alpha-$expectile, there is evidence that $X_{2}$ is not outperformed by $X_{1}$ under all $L_{\alpha}^{E}\in\mathcal{L}_{\alpha}^{E}$, or $X_{2}$ may outperform $X_{1}$ at
least when a certain $L_{\alpha}^{E}\in\mathcal{L}_{\alpha}^{E}$ is used in the forecast
evaluation.\footnote{To see this, let $\Theta_{H_{1}}^{E}=\left\{ \theta:E\left[L_{\theta,\alpha}^{E}\left(X_{1},Y\right)\right]-E\left[L_{\theta,\alpha}^{E}\left(X_{1},Y\right)\right]>0\right\} $.
If $\Theta_{H_{1}}^{E}\neq\emptyset$, the null of (\ref{null}) is
violated. In this case, $X_{2}$ outperforms $X_{1}$ under the extremal
consistent loss $L_{\theta^{*},\alpha}^{E}\left(x,y\right)$
where $\theta^{*}\in\Theta_{H_{1}}^{E}$. Note that $L_{\theta^{*},\alpha}^{E}\left(x,y\right)$
itself is also a consistent loss function for forecasting the $\alpha-$expectile.
The same argument can be applied to the case of evaluating the $\alpha-$quantile
forecasts. } On contrary, if the null is not rejected, there is evidence that for forecasting
the $\alpha-$expectile, $X_{1}$ performs equally well as or better
than $X_{2}$ under all $L_{\alpha}^{E}\in\mathcal{L}_{\alpha}^{E}$. 

Similarly, for comparing forecasts for the $\alpha-$quantile under all consistent loss functions, by
using the representation of (\ref{scoring_quantile1}) and the arguments that $dG(\theta)=d\zeta(\theta)$
is nonnegative for all $\theta\in\Theta$ and the functional form of the extremal consistent loss $L_{\theta,\alpha}^{Q}(x,y)$ is independent of $G\left(.\right)$, we may formulate the following null
hypothesis 
\begin{equation}
H_{0}:E\left[L_{\theta,\alpha}^{Q}\left(X_{1},Y\right)\right]\leq E\left[L_{\theta,\alpha}^{Q}\left(X_{2},Y\right)\right]\text{ for all }\theta.\label{null2}
\end{equation}
If the null of (\ref{null2}) is rejected, there is evidence that
$X_{2}$ may outperform $X_{1}$ for forecasting the $\alpha-$quantile,
at least when a certain $L_{\alpha}^{Q}\in\mathcal{L}_{\alpha}^{Q}$ is used in the forecast
evaluation. If the null is not rejected, there is evidence that for
forecasting the $\alpha-$quantile, $X_{1}$ at least can perform
no worse than $X_{2}$ over a class of consistent loss functions belonging
to $\mathcal{L}_{\alpha}^{Q}$. 

\subsection{The test statistics}

In the following we introduce procedures for testing the nulls of (\ref{null})
and (\ref{null2}). 
We consider $h$-period ahead out-of sample (OoS) forecasts of the $\alpha-$expectile or $\alpha-$quantile of a random variable $Y_{t+h}$ at each period $t$. Suppose total length of samples available for the forecast evaluation is $T$. Let $T_{R}$ denote the length of samples used to generate
the forecasts (such as length of samples used in estimating a
model). Let $T_{P}$ denote the number of generated forecasts and
so $T_{P}=T-h-T_{R}+1$. Let $f_{1,t+h|t}$ and $f_{2,t+h|t}$ denote the benchmark and competing forecasts for the $\alpha-$expectile or the $\alpha-$quantile of $Y_{t+h}$ at period $t$, $t=T_{R},\ldots,T-h$. To ease the notations, we let
$X_{1t}:=f_{1,t+h|t}$ and $X_{2t}:=f_{2,t+h|t}$. Let $
D_{\alpha}^{i}\left(\theta\right) =  E\left[L_{\theta,\alpha}^{i}\left(X_{1t},Y_{t+h}\right)\right]-E\left[L_{\theta,\alpha}^{i}\left(X_{2t},Y_{t+h}\right)\right]$, where $i\in \{E,Q\}$. The null hypotheses of (\ref{null}) or (\ref{null2}) is equivalent to 
\begin{equation}
H_{0}:D_{\alpha}^{i}\left(\theta\right)\leq0\text{ for all }\theta,\label{null1}
\end{equation}
if we replace $\left(X_{1},X_{2},Y\right)$ with $\left(X_{1t},X_{2t},Y_{t+h}\right)$.
Let $
\hat{d}_{t}^{i}\left(\theta\right)=L_{\theta,\alpha}^{i}\left(X_{1t},Y_{t+h}\right)-L_{\theta,\alpha}^{i}\left(X_{2t},Y_{t+h}\right)$. 
We can calculate a sample analogue of
$D_{\alpha}^{i}\left(\theta\right)$ as 
\begin{equation}
\hat{D}_{T_{P},\alpha}^{i}\left(\theta\right)=\frac{1}{T_{P}}\sum_{t=T_{R}}^{T-h}\hat{d}_{t}^{i}\left(\theta\right).\label{D_hat}
\end{equation}
If with some assumptions, $\sup_{\theta\in\Theta}\left|\hat{D}_{T_{P},\alpha}^{i}\left(\theta\right)-E\left[\hat{D}_{T_{P},\alpha}^{i}\left(\theta\right)\right]\right|\stackrel{p.}{\rightarrow}0$, then we may use the following test statistic 
\begin{equation}
\hat{S}_{T_{P},\alpha}^{i}=\sup_{\theta\in\Theta}\sqrt{T_{P}}\hat{D}_{T_{P},\alpha}^{i}\left(\theta\right)\label{test_stat}
\end{equation}
to test the null of (\ref{null1}). Here $\Theta\subseteq\mathbb{R}$
is the union of supports of $X_{1t}$, $X_{2t}$ and $Y_{t+h}$. 
To find the suprema in $\sqrt{T_{P}}\hat{D}_{T_{P},\alpha}^{i}\left(\theta\right)$, 
we may take the maxima over a grid of points in the joint supports of $X_{1t}$, $X_{2t}$ and $Y_{t+h}$, for example, all sample points of $X_{1t}$, $X_{2t}$ and $Y_{t+h}$. In practice, to save time
of computations, we may calculate approximations to the suprema based
on a smaller subset of the points. As the evaluation points increase in the joint supports, the theoretical properties for the test statistics will not be affected by using such approximations \citep{LMW_2005}.

\subsection{Properties of the test statistics}

In the following, we provide asymptotic results for the proposed test statistics of (\ref{test_stat}). 
We consider a more general version of the null of (\ref{null1}) in which $\left(X_{1t},X_{2t},Y_{t+h}\right)$
is replaced by $\left(X_{kt},X_{lt},Y_{t+h}\right)$, $k\neq l$, $k,l=1,\ldots,K$.
In the more generalized situation, we have $K$ generated forecasts
and the $k$th forecast is the benchmark and the other $K-1$ forecasts
are the competitors. Let 
\begin{eqnarray}
\hat{d}_{kl,t}^{i}\left(\theta\right) & = & L_{\theta,\alpha}^{i}\left(X_{kt},Y_{t+h}\right)-L_{\theta,\alpha}^{i}\left(X_{lt},Y_{t+h}\right)\nonumber\\
D_{kl,\alpha}^{i}\left(\theta\right) & = & E\left[L_{\theta,\alpha}^{i}\left(X_{kt},Y_{t+h}\right)\right]-E\left[L_{\theta,\alpha}^{i}\left(X_{lt},Y_{t+h}\right)\right]=E\left[\hat{d}_{kl,t}^{i}\left(\theta\right)\right],\nonumber\\
\hat{D}_{kl,\alpha}^{i}\left(\theta\right) & = & \frac{1}{T_{P}}\sum_{t=T_{R}}^{T-h}\left[L_{\theta,\alpha}^{i}\left(X_{kt},Y_{t+h}\right)-L_{\theta,\alpha}^{i}\left(X_{lt},Y_{t+h}\right)\right]=\frac{1}{T_{P}}\sum_{t=T_{R}}^{T-h}\hat{d}_{kl,t}^{i}\left(\theta\right),\nonumber\\
S_{\alpha}^{i} & = & \max_{k\neq l,k,l=1,\ldots,K}\sup_{\theta\in\Theta}D_{kl,\alpha}^{i}\left(\theta\right),\nonumber\\
\hat{S}_{T_{P},\alpha}^{i} & = & \max_{k\neq l,k,l=1,\ldots,K}\sup_{\theta\in\Theta}\sqrt{T_{P}}\hat{D}_{kl,\alpha}^{i}\left(\theta\right),\label{test_stat_sample}
\end{eqnarray}
where $i\in\left\{ E,Q\right\} $ is for the expectile and quantile
forecasts and $\Theta\subseteq\mathbb{R}$ is non-empty. By assuming
that $\left(X_{kt},X_{lt},Y_{t+h}\right)$ is strictly stationary,
it can be shown that 
\begin{eqnarray*}
\sup_{\theta\in\Theta}\sqrt{T_{P}}\hat{D}_{kl,\alpha}^{i}\left(\theta\right) & = & \sup_{\theta\in\Theta}\frac{1}{\sqrt{T_{P}}}\sum_{t=T_{R}}^{T-h}\left(\hat{d}_{kl,t}^{i}\left(\theta\right)-E\left[\hat{d}_{kl,t}^{i}\left(\theta\right)\right]+E\left[\hat{d}_{kl,t}^{i}\left(\theta\right)\right]\right)\\
 & = & \sup_{\theta\in\Theta}\left(v_{k,T_{P}}^{i}\left(\theta\right)-v_{l,T_{P}}^{i}\left(\theta\right)+\sqrt{T_{P}}D_{kl,\alpha}^{i}\left(\theta\right)\right),
\end{eqnarray*}
where 
\begin{equation}
v_{j,T_{P}}^{i}\left(\theta\right)=\sqrt{T_{P}}\left(\frac{1}{T_{P}}\sum_{t=T_{R}}^{T-h}\left(L_{\alpha,\theta}^{i}\left(X_{jt},Y_{t+h}\right)-E\left[L_{\alpha,\theta}^{i}\left(X_{jt},Y_{t+h}\right)\right]\right)\right),\label{v_j}
\end{equation}
for $i=\left\{ E,Q\right\} $, and $j=k,l$. With these notations,
we may rewrite a more generalized version of the nulls of (\ref{null1}) as 
\begin{equation}
H_{0}^{i}:S_{\alpha}^{i}\leq0,\label{null_master}
\end{equation}
for $i\in\left\{ E,Q\right\}$. 

If the null of (\ref{null_master})
is not true, the term $\sqrt{T_{P}}D_{kl,\alpha}^{i}\left(\theta\right)\rightarrow\infty$
as $T_{P}\rightarrow\infty$ for some $\theta$. If the null of (\ref{null_master}) is true, there exists at least a pair $\left(k,l\right)$ such that $D_{kl,\alpha}^{i}\left(\theta\right)\leq0$ for all $\theta\in \Theta$. Now suppose that under the null of (\ref{null_master}), with the pair $\left(k,l\right)$, $D_{kl,\alpha}^{i}\left(\theta\right)\leq0$
for all $\theta\in\Theta$ but $D_{kl,\alpha}^{i}\left(\theta\right)=0$
for some $\theta\in\mathcal{A}_{kl}^{i}\subseteq\Theta$. This implies
that $\sup_{\theta\in\Theta}D_{kl,\alpha}^{i}\left(\theta\right)=0$.
Let $\tilde{D}_{kl,\alpha}^{i}\left(\theta\right)=\hat{D}_{kl,\alpha}^{i}\left(\theta\right)-D_{kl,\alpha}^{i}\left(\theta\right)$.
Under some suitable conditions, with the central limit theorem of
an empirical process, it can be shown that the centered process $\sqrt{T_{P}}\tilde{D}_{kl,\alpha}^{i}\left(\theta\right)$
will converge weakly to a mean zero Gaussian process indexed by $\theta$,
say $\tilde{g}_{kl}^{i}\left(\theta\right)$. Since for $\theta\in\mathcal{A}_{kl}^{i}$,
$\sqrt{T_{P}}D_{kl,\alpha}^{i}\left(\theta\right)=0$ but for $\theta\notin\mathcal{A}_{kl}^{i}$,
$\sqrt{T_{P}}D_{kl,\alpha}^{i}\left(\theta\right)\rightarrow-\infty$
as $T_{P}\rightarrow\infty$ and $\sup_{\theta\in\Theta}\left(-\sqrt{T_{P}}\hat{D}_{kl,\alpha}^{i}\left(\theta\right)\right)\rightarrow\infty$
as $T_{P}\rightarrow\infty$. But $\sup_{\theta\in\Theta}\sqrt{T_{P}}\hat{D}_{kl,\alpha}^{i}\left(\theta\right)$
will approximately equal to $\sup_{\theta\in\Theta}\sqrt{T_{P}}\tilde{D}_{kl,\alpha}^{i}\left(\theta\right)$.
Thus the asymptotic distribution of $\sup_{\theta\in\Theta}\sqrt{T_{P}}\hat{D}_{T_{P},\alpha}^{i}\left(\theta\right)$
is determined by $\sup_{\theta\in\Theta}\sqrt{T_{P}}\tilde{D}_{T_{P},\alpha}^{i}\left(\theta\right)$,
which will weakly converge to $\sup_{\theta\in\Theta}\tilde{g}_{kl}^{i}\left(\theta\right)$
under some suitable conditions. On contrary, if with the pair $\left(k,l\right)$,
$D_{kl,\alpha}^{i}\left(\theta\right)<0$ for all $\theta\in\Theta$,
which implies that $\mathcal{A}_{kl}^{i}$ is empty, then 
\[
\sup_{\theta\in\Theta}\sqrt{T_{P}}\hat{D}_{kl,\alpha}^{i}\left(\theta\right)=\sup_{\theta\in\Theta}\sqrt{T_{P}}\left[\tilde{D}_{kl,\alpha}^{i}\left(\theta\right)+D_{kl,\alpha}^{i}\left(\theta\right)\right]\rightarrow-\infty
\]
as $T_{P}\rightarrow\infty$. 

We now state relevant assumptions and a formal theorem for the properties of the test statistic $\hat{S}_{T_{P},\alpha}^{i}$ as follows. Let $x\vee y=\max(x,y)$ and $x\wedge y=\min(x,y)$ and $\Rightarrow$ denote weak convergence of stochastic processes.
\begin{assumption}
	For $k=1,\ldots,K$, $\left\{ \left(Y_{t+h},X_{kt}\right):t=1,\ldots,T-h\right\} $ is strictly stationary and satisfies strong mixing condition. The mixing coefficients $\alpha\left(n\right)$ satisfy $\sum_{n=1}^{\infty}\left[\alpha\left(n\right)\right]^{A}<\infty$, 
	where $A<1/[\left(r-1\right)\left(r+1\right)]\wedge\left(\varrho/\left(2+\varrho\right)\wedge\left(s-r\right)/rs\right)$, $2\leq r < s$, $2+\varrho\leq s$ and $0<\varrho$ are some constants. 
\end{assumption}
\begin{assumption} The forecast error $\varepsilon_{k,t+h}=Y_{t+h}-X_{kt}$ should satisfy \[\left\Vert \varepsilon_{k,t+h}\right\Vert _{s}:=E\left[|\varepsilon_{k,t+h}|^{s}\right]^{\frac{1}{s}}<\infty,\]where $s$ is the constant satisfying the conditions in Assumption 1.
\end{assumption}
\begin{assumption} For $k=1,\ldots,K$ and $t=1,\ldots,T-h$, the marginal
density functions of $X_{kt}$ and $Y_{t+h}$, denoted by $f_{X_{kt}}(x)$
and $f_{Y_{t+h}}(y)$, are bounded with respect to Lebesgue measure
a.s.
\end{assumption}
Assumption 1 requires that the generated forecasts and random variable should satisfy a mixing condition. This kind of requirement for time series data is commonly seen in proving consistency results which rely on using property of stochastic equicontinuity of an empirical process (e.g., \cite{Hansen_1996a}, \cite{JCS_2016}, \cite{LMW_2005}, \cite{LWY_2016}). 
Assumption 2 requires the forecast error should satisfy a certain moment condition and Assumption 3 states density functions of the generated forecasts and random variable should be bounded from above. There is a trade-off between the moment condition of Assumption 2 and restriction on the constant $A$ in Assumption 1. In our case, we need all the three assumptions to construct the stochastic equicontinuity of the empirical process for $v_{j,T_{p}}^{i}\left(\theta\right)$ in (\ref{v_j}), which is indexed by the parameter $\theta$. With the results of the stochastic equicontinuity, some other useful statistical convergence results can be established. Please see Lemma 1 to 3 and their proofs in Appendix 7.1. 

\begin{theorem}

Suppose Assumptions 1 to 3 hold. Then under the null of (\ref{null_master}), the test statistic 
\[
\hat{S}_{T_{P},\alpha}^{i}\Rightarrow\begin{cases}
\max_{\left(k,l\right)\in\mathcal{K}}\sup_{\theta\in\mathcal{A}_{kl}^{i}}\tilde{g}_{kl}^{i}\left(\theta\right) & \text{ if }S_{\alpha}^{i}=0\\
-\infty & \text{ if }S_{\alpha}^{i}<0,
\end{cases}
\]
for $i\in\left\{ E,Q\right\} $, where $\tilde{g}_{kl}^{i}\left(\theta\right)$
is a mean zero Gaussian process with covariance $var_{kl}^{i}\left(\theta_{1},\theta_{2}\right)$
defined in Lemma 3, and $\mathcal{K}=\left\{ \left(k,l\right):k\neq l,k,l=1,\ldots,K,\sup_{\theta\in\Theta}D_{kl,\alpha}^{i}\left(\theta\right)=0\right\} $
and $\mathcal{A}_{kl}^{i}=\left\{\theta: \theta\in\Theta,D_{kl,\alpha}^{i}\left(\theta\right)=0\right\} .$

\end{theorem}
A detailed proof of Theorem 1 can be found in Appendix 7.1. The theorem
says that the sample test statistic $\hat{S}_{T_{P},\alpha}^{i}$ of (\ref{test_stat_sample})
has a non-degenerate asymptotic distribution associated with $\tilde{g}_{kl}^{i}\left(\theta\right)$,
which can be used to construct empirical $p$-values. In next subsection
we will introduce the method for empirically constructing the distribution
of the sample test statistic $\hat{S}_{T_{P},\alpha}^{i}$.

\subsection{Constructing empirical distributions of the test statistics}

We use the re-centered bootstrap \citep{LMW_2005} to construct the
empirical distribution of the sample test statistic $\hat{S}_{T_{P},\alpha}^{i}$,
where $i\in\left\{ E,Q\right\}$ is for the $\alpha-$expectile or $\alpha-$quantile forecast. 
In the following we briefly describe procedures for implementing the re-centered bootstrap. We focus on the case of comparing two forecasts $X_{1t}$ and $X_{2t}$. 
Let 
\[
\hat{d}_{t}^{i*}\left(\theta\right):=\hat{d}_{12,t}^{i*}\left(\theta\right)=L_{\theta,\alpha}^{i}\left(X_{1t}^{*},Y_{t+h}^{*}\right)-L_{\theta,\alpha}^{i}\left(X_{2t}^{*},Y_{t+h}^{*}\right),
\]
where $i\in \{E,Q\}$ and $\left(X_{1t}^{*},X_{2t}^{*},Y_{t+h}^{*}\right)$ is the bootstrap
sample randomly drawn with replacement from the empirical (joint)
distribution of $\left(X_{1t},X_{2t},Y_{t+h}\right)$ by using a bootstrap
re-sampling scheme, e.g., the stationary bootstrap of \citet{PR_1994}. 
Let 
$
\hat{D}_{T_{P},\alpha}^{i*}\left(\theta\right)=1/T_{P}\sum_{t=T_{R}}^{T-h}\hat{d}_{t}^{i*}\left(\theta\right),
$ which is an analogue of
$\hat{D}_{\alpha}^{i}\left(\theta\right)$ in (\ref{D_hat}) calculated
with the bootstrap sample. Let $\hat{D}_{c,T_{P},\alpha}^{i*}\left(\theta\right)=\hat{D}_{T_{P},\alpha}^{i*}\left(\theta\right)-E^{*}\left[\hat{D}_{T_{P},\alpha}^{i}\left(\theta\right)\right]$. 
Here $E^{*}[.]$ denotes the expectation relative
to the distribution of bootstrap sample $\left(X_{1t}^{*},X_{2t}^{*},Y_{t+h}^{*}\right)$
conditional on the original sample $\left(X_{1t},X_{2t},Y_{t+h}\right).$ 
Practically, we may replace $E^{*}\left[\hat{D}_{T_{P},\alpha}^{i}\left(\theta\right)\right]$
with $\hat{D}_{T_{P},\alpha}^{i}\left(\theta\right)$, the test statistic
calculated with the full sample. 
Let $
\hat{S}_{c,T_{P},\alpha}^{i*}=\sup_{\theta\in\Theta}\sqrt{T_{P}}\hat{D}_{c,T_{P},\alpha}^{i*}\left(\theta\right)
$
denote the re-centered bootstrap sample test statistic. We then compute
the bootstrap distribution of $\hat{S}_{c,T_{P},\alpha}^{i*}$ as
$
\hat{H}_{M}^{i}\left(\omega\right)=1/M\sum_{i=1}^{M}1\left\{\hat{S}_{c,T_{P},\alpha}^{i*}\leq\omega\right\}
$
and use it to construct the critical value and empirical p-value for
the test. Here $M$ is the size of the bootstrap sample. Let $\hat{h}_{M}^{i}\left(1-\gamma\right)$
denote ($1-\gamma$)th sample quantile of $\hat{H}_{M}^{i}\left(\omega\right)$: $
\hat{h}_{M}^{i}\left(1-\gamma\right)=\inf\left\{ \omega:\hat{H}_{M}^{i}\left(\omega\right)\geq1-\gamma\right\},
$ which is the re-centered bootstrap
critical value of significance level $\gamma$. We reject the null
hypothesis at the significance level $\gamma$ if $\hat{S}_{T_{P},\alpha}^{i}\geq\hat{h}_{M}^{i}\left(1-\gamma\right)$, $i\in \{E,Q\}$.

Let $W_{t}=\left(X_{1t},X_{2t},\ldots,X_{kt},Y_{t+h}\right)$, $t=1,\ldots,K$.
Let $p_{T_{P}}$ be the reciprocal of mean block length for the stationary bootstrap of \citet{PR_1994}, which is a function of $T_{P}$. With the notations used in Subsection 3.2, the theoretical result for validation of using the re-centered bootstrap method with the stationary bootstrap scheme are stated as follows.

\begin{theorem} Suppose Assumptions 1 and 2 hold and $p_{T_{P}}\rightarrow0$
	and $T_{P}\times p_{T_{P}}\rightarrow\infty$ as $T_{P}\rightarrow\infty$.
	Then for $i\in\{E,Q\}$, we have 
	\begin{eqnarray*}
		\sup_{\omega\in\mathbb{R}}\left|P\left(\max_{k\neq l,k,l=1,\ldots,K}\sup_{\theta\in\Theta}\sqrt{T_{P}}\left(\hat{D}_{kl,\alpha}^{i*}\left(\theta\right)-\hat{D}_{kl,\alpha}^{i}\left(\theta\right)\right)\leq\omega|W_{T_{R}},\ldots,W_{T-h}\right)\right.\\
		\left.-P\left(\max_{k\neq l,k,l=1,\ldots,K}\sup_{\theta\in\Theta}\sqrt{T_{P}}\left(\hat{D}_{kl,\alpha}^{i}\left(\theta\right)-D_{kl,\alpha}^{i}\left(\theta\right)\right)\leq\omega\right)\right| & \overset{p.}{\rightarrow} & 0
	\end{eqnarray*}
	as $T_{P}\rightarrow\infty$. Furthermore, as $T_{P}$ and $M\rightarrow\infty$, 
	\begin{itemize}
		\item[1.] if 
		\begin{equation}
		E\left[L_{\theta,\alpha}^{i}\left(X_{1t},Y_{t+h}\right)\right]=E\left[L_{\theta,\alpha}^{i}\left(X_{2t},Y_{t+h}\right)\right]=\ldots=E\left[L_{\theta,\alpha}^{i}\left(X_{kt},Y_{t+h}\right)\right]\text{ for all }\theta\in\Theta\label{implicit_constraint}
		\end{equation}
		holds, we have $S_{\alpha}^{i}=0$ and  $P\left(\hat{S}_{T_{p},\alpha}\geq\hat{h}_{M}^{i}\left(1-\gamma\right)\right)\rightarrow\gamma$. 
		\item[2.] if $S_{\alpha}^{i}>0$, we have $P\left(\hat{S}_{T_{p},\alpha}\geq\hat{h}_{M}^{i}\left(1-\gamma\right)\right)\rightarrow1$. 
	\end{itemize}
\end{theorem}

As pointed out by \citet{LMW_2005}, to suitably approximate the distribution
of the test statistic under the null, using the re-centered bootstrap method (or other re-centered re-sampling methods) requires (\ref{implicit_constraint})
holds. The implicit constraint of (\ref{implicit_constraint}) is
a least favorable configuration for the test, which is a special case
of $S_{\alpha}^{i}=0$ and the null $H_{0}^{i}:S_{\alpha}^{i}\leq0$. But note that $S_{\alpha}^{i}=0$ does not imply the favorable configuration. When (\ref{implicit_constraint}) holds, using the re-centered bootstrap method
would yield an exact asymptotic size of the test statistic. But when
it fails to hold, in general the exact asymptotic size of the test
statistic would not be obtained by using the re-centered bootstrap method.
To sum, the re-centered bootstrap sample test statistic is not asymptotically
similar on the boundary of the null. When an alternative is too close
to the null, in general, a non-asymptotic similar test statistic may
be less powerful for it than an asymptotic similar test statistic.
However, previous studies show that the re-centered bootstrap method performs at least equally well as other re-sampling methods, either in
simulations or empirical applications, see \citet{LMW_2005} and
\citet{JCS_2016}. This is the main reason why we suggest to use the
re-centered bootstrap method to conduct the proposed tests.\footnote{In an early work, we also used subsampling method suggested by \citet{LMW_2005} to conduct the proposed tests but found in most situations it performs worse than the re-centered bootstrap method. The relevant results of using the subsampling scheme can be requested.} We will use the re-centered bootstrap method in the following simulations and empirical analyses. 

\section{Simulations}

In this section, we conduct simulations to understand how the proposed test statistics perform. 
In the first simulation in Section 4.1.1, we investigate how the proposed test statistic works when different consistent loss functions provide different ranking results for two forecasts on the conditional expectation. In the rest simulations, models E1 to E3 are for the conditional expectile forecasts and models Q1 and Q2 are for the conditional quantile forecasts. We use these models to examine how the proposed test statistics perform under different data generating processes. 

For each simulation, we set the number of generated forecasts $T_{P}=100$, 300 and 1000, and the number of bootstrap $M=400$. Each scenario is simulated 1000 times. For the simulation in Section 4.1.1 and model E1 and Q1, the forecasts are not generated from any estimated model. 
For models E2, E3 and Q2, the forecasts are generated by using rolling window scheme with window length $l=100$, and for each model, length of a generated sample path $T=T_{R}+T_{P}$, where $T_{R}=l=100$ is the sample size for initial estimations of the model parameters. 
In the main context, for each simulation, we show rejection frequencies of the proposed test statistics used for the simulations from the 1000 iterations. As for a more completed description for properties of size and power of the proposed test statistics, we show their size-power curves \citep{DM_1998} in Appendix 7.3.


\subsection{Conditional expectile forecasts}

In this subsection, we present simulation results for forecasting the conditional $\alpha-$expectile of a random variable $Y_{t+1}$ at each period $t$: $e_{t+1|t}\left(\alpha\right):=\upsilon$, where $\upsilon$ satisfies
\[
\frac{E_{t}\left[\left(\upsilon-Y_{t+1}\right)_{+}\right]}{E_{t}\left[\left(Y_{t+1}-\upsilon\right)_{+}\right]} = \frac{\alpha}{1-\alpha},
\]
and $E_{t}\left[.\right]=E\left[.|I_{t}\right]$ is the conditional
expectation operator at period $t$ and $I_{t}$ is the information
set up to period $t$. Again we let $X_{1t}:=f_{1,t+1|t}$ be the benchmark and $X_{2t}:=f_{2,t+1|t}$ be the competing forecasts.

\subsubsection{A comparison of consistent loss functions and the proposed test}
We first consider a simulation when different consistent loss functions provide different ranking results for two competing forecasts on the conditional expectation of $Y_{t+1}$: $E_{t}\left[Y_{t+1}\right]=e_{t+1|t}(0.5)$.
The consistent loss functions we consider here are the squared error
loss 
and the exponential Bregman loss. 
The random variable $Y_{t+1}$ has the following data generating process 
\begin{equation}
Y_{t+1}=\gamma+\beta_{1}W_{1t}+\beta_{2}W_{2t}+\varepsilon_{t+1},
\label{sim_mse_expb}
\end{equation}
where $W_{1t}\sim i.i.d.N\left(0,\sigma_{W_{1}}^{2}\right)$, $W_{2t}\sim i.i.d.N\left(0,\sigma_{W_{2}}^{2}\right)$
and $\varepsilon_{t+1}\sim i.i.d.N\left(0,1\right)$. $W_{1t}$, $W_{2t}$ and $\varepsilon_{t+1}$ are mutually independent.
We set $\gamma=0.4$, $\beta_{1}=0.5$, $\beta_{2}=0.2$ and $\sigma_{W_{1}}^{2}=\sigma_{W_{2}}^{2}=1$.
The benchmark forecast is $X_{1t}=c_{1}+b_{1}W_{1t}$ and the
competitor is $X_{2t}=c_{2}+b_{2}W_{2t}$. We consider three scenarios for parameter settings: (1) $c_{1}=c_{2}=2\gamma$, $b_{1}=2\beta_{1}$ and $b_{2}=2\beta_{2}$; (2) $c_{1}=2\gamma$, $c_{2}=\gamma$, $b_{1}=2\beta_{1}$ and $b_{2}=\beta_{2}$; (3) $c_{1}=\gamma$, $c_{2}=2\gamma$, $b_{1}=\beta_{1}$ and $b_{2}=2\beta_{2}$. The three scenarios result in different forecast rankings when the squared error loss is used. Let $MSE(X,Y):=E[(X-Y)^{2}]$ denote the expected squared error loss of the random variable $Y$ and forecast $X$.
As shown in Appendix 7.5, scenario (1) implies $MSE(X_{1t},Y_{t+1})=MSE(X_{2t},Y_{t+1})$; scenario (2) implies $MSE(X_{1t},Y_{t+1})>MSE(X_{2t},Y_{t+1})$ and scenario (3) implies $MSE(X_{1t},Y_{t+1})<MSE(X_{2t},Y_{t+1})$. 

In the left panel of Figure \ref{figure3}, we plot differences of the expected
exponential Bregman loss for the two forecasts under the three scenarios
with parameter $a\in\left[-1,1\right]$. The right panel of Figure \ref{figure3} shows differences of the expected extremal consistent loss for the two forecasts with parameter $\theta\in\left[-5,5\right]$.
	
In scenario (1), the two forecasts have the same
expected squared error loss, 
but as can be seen from Figure \ref{figure3}, they have different expected exponential Bregman loss for $a\neq0$.\footnote{Note that for $a=0$, the exponential Bregman loss becomes the squared error loss (scaled by 0.5).} The difference is positive for $a>0$ and negative for $a<0$. In this scenario, if we use an accuracy test with the squared error loss, say the Diebold and Marino (DM) test, we will have a low rejection frequency since it is the least favorable configuration (l.f.c.) of the test. On contrary if the exponential Bregman loss with $a>0$ is used in the accuracy test, we may have a very high rejection frequency. As for the extremal consistent loss, the difference of their expected values has a positive maximum. 
It suggests that the null of (\ref{null}) should be rejected. 
		
In scenario (2), the competitor outperforms the benchmark under both the squared error loss and exponential Bregman loss, as can be seen from Figure \ref{figure3}. For the expected extremal consistent loss, again the difference has a positive maximum, 
which suggests that the null of (\ref{null}) should be rejected. But it is interesting to note that the difference also has a negative minimum, which suggests that the competitor may perform worse than the benchmark under a certain consistent loss function other than the squared error loss and exponential Bregman loss.
		
In scenario (3), 
the benchmark outperforms the competitor under the squared error loss and exponential Bregman loss. Furthermore, 
the difference of the expected extremal consistent loss is nonpositive for all $\theta$ considered here. It suggests that no matter which consistent loss function is used, the benchmark will still perform no worse than the competitor and the null of (\ref{null}) should not be rejected.
		
In the upper panel of Table \ref{table3}, we show rejection frequencies of the proposed test and the DM test with the squared error loss for scenarios (1) to (3). 
The significant levels we choose are 0.01, 0.05 and 0.1. The simulation results confirm what Figure \ref{figure3} shows. For scenario (1), 
rejection frequencies of the DM test are close to the corresponding significant levels, which is expected, since scenario (1) is the least favorable configuration for the DM test when the squared error loss is used. But in this scenario, rejection frequencies of the proposed test are much higher than the corresponding significant levels and increase with the number of generated forecasts $T_{P}$. 
For scenario (2), rejection frequencies of the proposed test and the DM test both increase with $T_{P}$. 
For scenario (3), the proposed test and the DM test both obtain no rejection, which again confirm what Figure \ref{figure3} shows. 

In the bottom panel of Table 3, we show simulation results for a ``reverse situation'' in which $X_{1t}$ is the competitor and $X_{2t}$ is the benchmark. In this situation, results for scenarios (1) and (3) are expected. The proposed test statistic and the DM test statistic behave as before in scenario (1). While in scenario (3), now the test statistics both have a high probability to reject the null. 
In scenario (2), as mentioned, the difference of the expected extremal consistent loss functions has a negative minimum, which implies that $X_{2t}$ may perform worse than $X_{1t}$ under a certain consistent loss function other than the squared error loss and exponential Bregman loss. Our proposed test statistic thus has a high probability to reject
the null in this case. However, using the DM test statistic has a
very low probability to reject the null since $X_{2t}$ performs better
than $X_{1t}$ under the squared error loss. 

\subsubsection{Model E1}
For this simulation, $Y_{t+1}|\mu_{t+1|t}  \sim  i.i.d.N\left(\mu_{t+1|t},1\right)$, 
where the conditional expectation $\mu_{t+1|t} \sim  i.i.d.N\left(0,1\right)$. Let $e^{Z}\left(\alpha\right)$ denote the $\alpha-$expectile
of a standard normal random variable $Z$. The conditional $\alpha-$expectile of $Y_{t+1}$ at period
$t$ is $e_{t+1|t}\left(\alpha\right)=\mu_{t+1|t}+e^{Z}\left(\alpha\right)$. 
We set the benchmark forecast for $e_{t+1|t}\left(\alpha\right)$
as $X_{1t} =  \mu_{t+1|t}+e^{Z}\left(\alpha\right)+\varsigma\left(\alpha\right)Z_{1t}$, where $Z_{1t} \sim  i.i.d.N\left(0,0.25\right)$ and
\[
\varsigma\left(\alpha\right)=\frac{\sqrt{E\left[\left(1\left\{ Z<e^{Z}\left(\alpha\right)\right\}-\alpha\right)^{2}\left(Z-e^{Z}\left(\alpha\right)\right)^{2}\right]}}{E\left[\left|1\left\{ Z<e^{Z}\left(\alpha\right)\right\}-\alpha\right|\right]}.
\]
The benchmark forecast $X_{1t}$ can be viewed as a noisy forecast
for the conditional $\alpha-$expectile $e_{t+1|t}\left(\alpha\right)$.
For the noise $Z_{1t}$, we scale it with $\varsigma\left(\alpha\right)$
to reflect the fact that accuracy of forecasting conditional expectiles
generally depends on $\alpha$.\footnote{Note that $\varsigma^{2}\left(\alpha\right)/n$ is the asymptotic variance of the empirical $\alpha-$expectile for $n$ i.i.d. normal samples, see \cite{NP_1987}.} We use the following settings to generate the competing forecast $X_{2t}$:
(1) $X_{2t}=\mu_{t+1|t}+e^{Z}\left(\alpha\right)$;
(2) $X_{2t}=\mu_{t+1|t}+e^{Z}\left(\alpha\right)+\varsigma\left(\alpha\right)Z_{it}$,
$Z_{it}\sim i.i.d.N\left(0,\sigma_{i}^{2}\right)$ and $\sigma_{i}^{2}=0.04,$
0.25 and 1 for $i=2$, 3, 4; (3) $X_{2t}=e^{Z}\left(\alpha\right)+\varsigma\left(\alpha\right)Z_{it}$,
$Z_{it}\sim i.i.d.N\left(0,\sigma_{i}^{2}\right)$, where $\sigma_{i}^{2}=0.25$
and 1 for $i=3$, 4.

In setting (1), $X_{2t}$ is the true conditional $\alpha-$expectile. 
In setting (2), like 
$X_{1t}$, $X_{2t}$ can be viewed as a noisy forecast for the conditional
$\alpha-$expectile. In particular, $X_{1t}$ and $X_{2t}=\mu_{t+1|t}+e^{Z}\left(\alpha\right)+\varsigma\left(\alpha\right)Z_{3t}$
shall be equivalent since their noisy terms both follow $N\left(0,0.25\right)$, and this case is the least favorable configuration for the test. When $X_{2t}=\mu_{t+1|t}+e^{Z}\left(\alpha\right)+\varsigma\left(\alpha\right)Z_{2t}$
($\mu_{t+1|t}+e^{Z}\left(\alpha\right)+\varsigma\left(\alpha\right)Z_{4t}$),
$X_{2t}$ is on average a more accurate (less accurate) forecast than
$X_{1t}$, since the noise $Z_{2t}$ ($Z_{4t}$) has a smaller (larger)
variance than $Z_{1t}$ does. In setting (3), $X_{2t}$ can be viewed as a noisy forecast
when the conditional expectation $\mu_{t+1|t}$ is replaced with the unconditional
expectation (zero). Also the noise has the same or a larger variance
than $Z_{1t}$ does. Thus in this case, $X_{2t}$ is expected to perform
worse than $X_{1t}$.

\subsubsection{Model E2}

For this simulation, we generate data from a VAR(1) model: 
\begin{eqnarray}
Y_{t+1} & = & 0.1+0.3Y_{t}+\beta_{2}W_{1t}+\varepsilon_{1,t+1},\nonumber \\
W_{1,t+1} & = & 0.2+0.6W_{1t}+\varepsilon_{2,t+1},\nonumber\\ 
W_{2,t+1} & = & 0.3+0.4W_{2t}+\varepsilon_{3,t+1},\nonumber 
\end{eqnarray}
where 
\begin{eqnarray*}
\left[\begin{array}{c}
\varepsilon_{1,t+1}\\
\varepsilon_{2,t+1}\\
\varepsilon_{3,t+1}
\end{array}\right] & \sim & i.i.d.MN\left(\mathbf{0},\Omega_{\varepsilon}\right),\\
\Omega_{\varepsilon} & = & \left[\begin{array}{ccc}
1 & 0 & 0\\
0 & 1 & \sigma_{23}\\
0 & \sigma_{23} & 1
\end{array}\right],
\end{eqnarray*}
and $MN\left(\mathbf{0},\Omega_{\varepsilon}\right)$ denotes a multivariate normal distribution with mean vector $\mathbf{0}$ and covariance matrix $\Omega_{\varepsilon}$. Here we focus on evaluating forecasts of the conditional expectation of $Y_{t+1}$ at each period $t$. The parameter $\beta_{2}$ controls the importance of $W_{1t}$ for the forecast. 
For $W_{2t}$, it does not directly affect $Y_{t+1}$ and may not be helpful on the
forecast. However, if its correlation with $W_{1t}$ (measured by $\sigma_{23}$) is high and $W_{1t}$ is not available, $W_{2t}$ can be a suitable alternative predictor. In the simulation, we will vary $\beta_{2}$ and $\sigma_{23}$ and see how such variations
affect performances of the proposed test statistic. The forecasts are all generated with estimated models in which the estimated coefficients at period $t$ 
are obtained from using the OLS and rolling window scheme with window length $l=100$.

The benchmark forecast is $X_{1t} =  \left(\hat{\gamma_{t}}+Z_{1t}\right)+\left(\hat{\beta}_{1t}+Z_{2t}\right)Y_{t}$, where $Z_{1t}  \sim  i.i.d.N\left(0,0.0025\right)$, $Z_{2t} \sim i.i.d.N\left(0,0.0225\right)$,
and $\hat{\gamma}_{t}$ and $\hat{\beta}_{1t}$ are the estimated coefficients at period $t$. 
The benchmark is from a misspecified model in which the coefficients are the OLS estimates plus noises. 
We use the following six settings to generate the competing forecast $X_{2t}$: (1) $\left(\beta_{2},\sigma_{23}\right)=\left(0.45,0\right)$,
$X_{2t}=\tilde{\gamma}_{t}+\tilde{\beta}_{1t}Y_{t}$, $\tilde{\gamma}=\hat{\gamma}+Z_{3t}$,
$\tilde{\beta}_{1t}=\hat{\beta}_{1t}+Z_{4t}$. $Z_{3t}\sim i.i.d.N\left(0,0.0025\right)$
and $Z_{4t}\sim i.i.d.N\left(0,0.0225\right)$. 
For settings (2) to (4), we set $\sigma_{23}=0$,
$\beta_{2}=0.1$, 0.45 and 0.75, and $X_{2t}=\hat{\gamma_{t}}+\hat{\beta}_{1t}Y_{t}+\hat{\beta}_{2t}^{k}W_{1t}$, where $\hat{\beta}_{2t}^{k}$ is the estimated coefficient at period $t$ and $k=low,$ $med$ and $high$ correspond to $\beta_{2}=0.1$,
0.45 and 0.75. 
For settings (5) and (6), we set $\sigma_{23}=0.3$ and 0.8, $\beta_{2}=0.45$,
and $X_{2t}=\hat{\gamma}_{t}+\hat{\beta}_{1t}Y_{t}+\hat{\beta}_{3t}W_{2t}^{h}$, where $\hat{\beta}_{3t}$ is the estimated coefficient at period $t$ and $h=lcr$ and $hcr$ correspond to $\sigma_{23}=0.3$ and 0.8. 

In setting (1), similar as the benchmark $X_{1t}$, $X_{2t}$ is also from a misspecified model in which the estimated coefficients are perturbed by noises. Since the noises in the benchmark and this setting follow the same distribution, $X_{1t}$ and $X_{2t}$ shall be equivalent forecasts. Hence setting (1) is the least favorable configuration (l.f.c.) for the test. In settings (2) to (4), we vary the coefficient $\beta_{2}$ at three different levels and keep $W_{1t}$ and $W_{2t}$ uncorrelated. The model used here is correctly specified. Comparing to the benchmark forecast $X_{1t}$, it is expected that as magnitude
of $\beta_{2}$ becomes strong, $W_{1t}$ will become more important
in the forecast, and $X_{2t}$ will outperform $X_{1t}$. Finally, in settings (5) and (6), we vary correlation between $W_{1t}$
and $W_{2t}$ at two different levels but keep $\beta_{2}$ constant. Although the model used in settings (5) and (6) is not correctly specified, it is expected that as the correlation between $W_{1t}$ and $W_{2t}$ increases, $W_{2t}$ may become more useful on the forecast. Hence $X_{2t}$ may perform better than $X_{1t}$ in this case. 

\subsubsection{Model E3}
For this simulation, we generate data by using a GARCH(1,1) model.
We focus on evaluating forecasts of the conditional expectation of $Y_{t+1}=V_{t+1}^{2}$ at each period $t$, where $V_{t+1}|\sigma_{t+1|t}^{2} \sim  N\left(0,\sigma_{t+1|t}^{2}\right)$ and $\sigma_{t+1|t}^{2} = 0.05+0.75\sigma_{t|t-1}^{2}+0.2V_{t}^{2}$. Note that $E_{t}\left[Y_{t+1}\right]=E_{t}\left[V_{t+1}^{2}\right]=\sigma_{t+1|t}^{2}$. 
The benchmark forecast is $X_{1t} = \exp(-0.045)U_{1t}Y_{t}$, where $\ln U_{1t} \sim i.i.d.N\left(0,0.09\right)$. Note that $E\left[\exp(-0.045)U_{1t}\right]=1$ and the benchmark forecast is an unbiased forecast. 
Let $\hat{\sigma}_{t+1|t}^{2}\left(p,q\right)=\hat{a}_{t}+\sum_{i=1}^{p}\hat{b}_{it}\hat{\sigma}_{t+1-i|t-i}^{2}+\sum_{j=1}^{q}\hat{c}_{jt}V_{t+1-j}^{2}$
denote a one-period ahead forecast for $\sigma_{t+1|t}^{2}$, in which $\hat{a}_{t}$, $\hat{b}_{it}$ and $\hat{c}_{jt}$ are the estimated coefficients at period $t$ obtained from using the maximized likelihood (ML). 
We use the following settings to generate the competing forecast $X_{2t}$: (1) $X_{2t}=\exp(-0.045)U_{2t}Y_{t}$, $\ln U_{2t}\sim i.i.d.N\left(0,0.09\right)$; 
(2) $X_{2t}=\hat{\sigma}_{t+1|t}^{2}\left(0,1\right)$; (3) $X_{2t}=\hat{\sigma}_{t+1|t}^{2}\left(1,1\right)$;
(4) $X_{2t}=\hat{\sigma}_{t+1|t}^{2}\left(2,2\right)$.

In setting (1), similar as the benchmark forecast, $X_{2t}$ is a random walk forecast scaled by a log-normal noise multiplying $\exp(-0.045)$. Since the noises in the benchmark and this setting follow the same distribution, $X_{1t}$ and $X_{2t}$ shall be equivalent forecasts and setting (1) is the least favorable configuration (l.f.c.) for the test. In
setting (3), $X_{2t}$ is a forecast from the correctly specified
GARCH(1,1) model and it is expected to outperform the benchmark forecast
$X_{1t}$. In setting (2) and (4), $X_{2t}$ is a forecast from misspecified models ARCH(1) and GARCH(2,2), respectively. 

\subsubsection{Simulation results}


Table \ref{table4} shows rejection frequencies of
the test statistic for using model E1. 
We can see that when the competing forecast $X_{2t}$ is either $\mu_{t+1|t}+e^{Z}(\alpha)$ or $\mu_{t+1|t}+e^{Z}(\alpha)+\varsigma(\alpha)Z_{2t}$, rejection frequency of the
test statistic increases as the length of forecast generated $T_{P}$ increases. The results are expected, since $\mu_{t+1|t}+e^{Z}(\alpha)$ is the true conditional expectation and $\mu_{t+1|t}+e^{Z}(\alpha)+\varsigma(\alpha)Z_{2t}$ has a smaller noisy perturbation than the
benchmark $X_{1t}$. In the least favorable configuration ($X_{2t}=\mu_{t+1|t}+e^{Z}(\alpha)+\varsigma(\alpha)Z_{3t}$), when $T_{P}$ is low, rejection frequency is slightly lower than the corresponding nominal size. But when $T_{P}$ increases, size of the test statistic is improved, as can be seen that the rejection frequency approaches to the corresponding significant level. For the other three settings, the results are very similar: over
different $T_{P}$ and significant levels, the rejection frequency
is at zero or a very low level. The results are also expected, since
these competing forecasts are worse forecasts than the benchmark forecast. 

Table \ref{table5} shows rejection frequencies of
the test statistic for using model E2. In the least favorable configuration, the rejection frequency behaves well. For the other five cases, the rejection frequency increases with the length of generated forecast $T_{P}$. 
As the magnitude of $\beta_{2}$ increases, on average the rejection frequency
increases. When $W_{2t}$ becomes more correlated
with $W_{1t}$, on average the rejection frequency also increases.
To sum, these results suggest that statistical power of the proposed test statistic increases when $W_{1t}$ becomes more important for $Y_{t+1}$ or correlation between $W_{1t}$ and $W_{2t}$ rises. Table \ref{table6} show rejection frequencies of
the test statistic for using model E3. As can be seen from the table, in the least favorable configuration, the rejection frequency is slightly lower than the corresponding significant level, which suggests that some size distortions occur here. For the other three cases, the rejection frequencies increase with $T_{P}$. 
\subsection{Conditional quantile forecasts}

In this subsection, we conduct simulations to understand how the proposed test statistic performs on evaluating forecasts of the conditional $\alpha-$quantile of the random variable
$Y_{t+1}$ at each period $t$. The conditional $\alpha-$quantile of $Y_{t+1}$
at period $t$ is defined as $q_{t+1|t}\left(\alpha\right):=\inf\left\{ \tau:P_{t}\left(Y_{t+1}\leq \tau\right)\geq\alpha\right\}$, 
where $P_{t}(.)$ is the conditional probability of $Y_{t+1}$ at period $t$. 

\subsubsection{Model Q1}

The data generating process for $Y_{t+1}$ used here is the same as in Subsection 4.1.2. Let $\varphi\left(x\right)$ and $\varPhi\left(x\right)$ denote density
and cumulative distribution functions of a standard normal random
variable. The conditional $\alpha-$quantile of $Y_{t+1}$ is $q_{t+1|t}(\alpha)=\mu_{t+1|t}+\varPhi^{-1}\left(\alpha\right)$, where $\varPhi^{-1}\left(\alpha\right)$
is the $\alpha-$quantile of the standard normal random variable. We set the benchmark forecast $X_{1t}=\mu_{t+1|t}+\varPhi^{-1}\left(\alpha\right)+\xi\left(\alpha\right)Z_{1t}$, where 
\[
\xi\left(\alpha\right)=\frac{\sqrt{\alpha\left(1-\alpha\right)}}{\varphi\left(\varPhi^{-1}\left(\alpha\right)\right)}
\]
and $Z_{1t}\sim i.i.d.N\left(0,0.25\right)$. 
The benchmark $X_{1t}$ is a noisy forecast for the true conditional
quantile. For the noise $Z_{1t}$, we scale it with $\xi\left(\alpha\right)$ to reflect the fact that accuracy of forecasting conditional quantiles generally depends on $\alpha$.\footnote{Note that $\xi^{2}\left(\alpha\right)/n$ is the asymptotic variance of the empirical $\alpha-$quantile for $n$ i.i.d. normal samples.} We use the following settings to generate competitors $X_{2t}$: (1) $X_{2t}=\mu_{t+1|t}+\varPhi^{-1}\left(\alpha\right)$;
(2) $X_{2t}=\mu_{t+1|t}+\varPhi^{-1}\left(\alpha\right)+\xi\left(\alpha\right)Z_{it}$,
$Z_{it}\sim i.i.d.N\left(0,\sigma_{i}^{2}\right)$ and $\sigma_{i}^{2}=0.04$,
0.25 and 1 for $i=2$, 3, and 4; (3) $X_{2t}=\varPhi^{-1}\left(\alpha\right)+\xi\left(\alpha\right)Z_{it}$,
$Z_{it}\sim i.i.d.N\left(0,\sigma_{i}^{2}\right)$, $\sigma_{i}^{2}=0.25$
and 1 for $i=3$ and 4.

In setting (1), $X_{2t}$ is the true conditional quantile. 
In setting (2), like $X_{1t}$, $X_{2t}$ can
be viewed as a noisy forecast for the true conditional quantile.
In particular, $X_{1t}$ and $X_{2t}=\mu_{t+1|t}+\varPhi^{-1}\left(\alpha\right)+\xi\left(\alpha\right)Z_{3t}$
shall be equivalent forecasts since their noisy terms both follow $N\left(0,0.25\right)$, and this case is the least favorable configuration for the test. When $X_{2t}=\mu_{t+1|t}+\varPhi^{-1}\left(\alpha\right)+\xi\left(\alpha\right)Z_{2t}$
($X_{2t}=\mu_{t+1|t}+\varPhi^{-1}\left(\alpha\right)+\xi\left(\alpha\right)Z_{4t}$),
$X_{2t}$ on average is a more accurate (less accurate) forecast than $X_{1t}$, since the noise $Z_{2t}$ ($Z_{4t}$) has a smaller (larger)
variance than $Z_{1t}$ does. In setting (3), $X_{2t}$ can be viewed
as a noisy forecast when the conditional expectation $\mu_{t+1|t}$ is replaced with the
unconditional one (zero). Also the noise has the same or a larger
variance than $Z_{1t}$ does. Thus in this case, $X_{2t}$ is expected to perform worse than $X_{1t}$. 

\subsubsection{Model Q2}
For this simulation, we set $Y_{t+1}=0.5+1.2W_{1t}+1.5W_{2t}+\varepsilon_{t+1}$, 
where $W_{1t}$, $W_{2t}$ and $\varepsilon_{t+1}$ $\sim i.i.d.N\left(0,1\right)$. We estimate
the conditional $\alpha-$quantile $q_{t+1|t}\left(\alpha\right)$ of $Y_{t+1}$ at period $t$ with $\hat{q}_{t}\left(\alpha\right)=\hat{\mu}_{t+1|t}+\hat{q}_{t}^{\varepsilon}\left(\alpha\right)$. 
Here $\hat{\mu}_{t+1|t}$ is a forecast for $E_{t}[Y_{t+1}]$ at period $t$ from a predictive regression. The predictive regression has different specifications and is estimated with the OLS with the rolling window scheme. $\hat{q}_{t}^{\varepsilon}\left(\alpha\right)$ is
the sample quantile of residuals $\hat{\varepsilon}_{i}^{t}$,
$i=t-l+1,\ldots,t$, of the predictive regression and $l=100$ is the rolling window length. 
The benchmark forecast $X_{1t}$ is given by $X_{1t} = \hat{\gamma_{t}}+\hat{\beta}_{1t}W_{1t}+\hat{q}_{t}^{\varepsilon}\left(\alpha\right)+Z_{1t}$, where $Z_{1t} \sim i.i.d.N\left(0,1\right)$ 
and $\hat{\gamma}_{t}$ and $\hat{\beta}_{1t}$ are the estimated coefficients at period $t$. 
In this case, $\hat{\mu}_{t+1|t}=\hat{\gamma_{t}}+\hat{\beta}_{1t}W_{1t}$ is a conditional expectation forecast from a misspecified predictive regression. 
The benchmark $X_{1t}$ thus can be viewed as a conditional quantile forecast from a misspecified model plus a noise $Z_{1t}$. We use
the following settings to generate the competitors $X_{2t}$: (1) $X_{2t}=\hat{\gamma_{t}}+\hat{\beta}_{1t}W_{1t}+\hat{q}_{t}^{\varepsilon}\left(\alpha\right)+Z_{2t}$,
$Z_{2t}\sim i.i.d.N\left(0,1\right)$;
(2) $X_{2t}=\hat{\gamma_{t}}+\hat{\beta}_{1t}W_{1t}+\hat{q}_{t}^{\varepsilon}\left(\alpha\right)$;
(3) $X_{2t}=\hat{\gamma_{t}}+\hat{\beta}_{1t}W_{1t}+\hat{\beta}_{2t}W_{2t}+\hat{q}_{t}^{\varepsilon}\left(\alpha\right)$;
(4) $X_{2t}=\hat{\gamma_{t}}+\hat{\beta}_{1t}W_{1t}+1.5W_{2t}+\hat{q}_{t}^{\varepsilon}\left(\alpha\right)$;
(5) $X_{2t}=0.5+1.2W_{1t}+1.5W_{2t}+\hat{q}_{t}^{\varepsilon}\left(\alpha\right)$.

In setting (1), $X_{2t}$ an equivalent forecast of $X_{1t}$, since they have the same $\hat{\mu}_{t+1|t}$ and the two noises
$Z_{1t}$ and $Z_{2t}$ have the same distribution. Hence setting
(1) is the least favorable configuration for the test. In
setting (2), $X_{2t}$ is the same as the benchmark but without the noise term. In setting (3), $\hat{\mu}_{t+1|t}$
is estimated from the correctly specified predictive regression. In setting
(4), $\hat{\mu}_{t+1|t}$ is a combination of two components: $\hat{\gamma_{t}}+\hat{\beta}_{1t}W_{1t}$
and $1.5W_{2t}$. The former is the same as the conditional expectation forecast
in setting (1) and the latter is $W_{2t}$ with its true coefficient.
In setting (5), $\hat{\mu}_{t+1|t}$ is the true conditional expectation of $Y_{t+1}$. From above, it can be seen that $X_{2t}$ in settings (2) to (5) are expected to outperform $X_{1t}$ in forecasting the conditional quantile of $Y_{t+1}$.

\subsubsection{Simulation results}

We report rejection frequencies of the proposed test statistic for using model Q1 in Table \ref{table7}. From the table, we can see that when
the competing forecast $X_{2t}$ is either $\mu_{t+1|t}+\Phi^{-1}\left(\alpha\right)$
or $\mu_{t+1|t}+\Phi^{-1}\left(\alpha\right)+\xi\left(\alpha\right)Z_{2t}$,
rejection frequency of the test statistic increases as the length
of generated forecast $T_{P}$ increases. The results are expected, 
since the two are more accurate forecasts than the benchmark $X_{1t}$. For the least favorable configuration ($X_{2t}=\mu_{t+1|t}+\Phi^{-1}\left(\alpha\right)+\xi\left(\alpha\right)Z_{3t}$
), the sizes are overall controlled well as $T_{P}$ increases. As for the other three settings, which are considered as worse forecasts than the benchmark, the results are very similar: over different $T_{P}$ and significant levels, the rejection frequency is at zero or a very low level.

Table \ref{table8} shows rejection frequencies of the proposed test statistic for using model Q2. From the table, we can see that for the least favorable configuration, overall the sizes are well controlled. We also can see that for settings (2) to (5), when $T_{P}$ is low, the rejection frequencies for the low quantiles ($\alpha=0.01$ and 0.05) are lower than those for the high quantile ($\alpha=0.5$). But as $T_{P}$ increases, the rejection frequencies increase. For settings (3) to (5), which use the correct model specification, the rejection frequencies for different quantiles approach to a satisfied level as $T_{P}$ increases. But for setting (2), which uses an incorrect model specification, the rejection frequencies for different quantiles still have some differences as $T_{P}$ increases. Overall the results suggest that as the competing forecast becomes more accurate than the benchmark, the proposed test statistic has more statistical power to detect the performance difference.


\section{Empirical applications}
\subsection{Forecasting equity risk premium of the S\&P500 Index}
In this subsection, we use the proposed test to evaluate abilities
of some predictors on forecasting risk premium of the S\&P500 index.
\citet{GW_2008} claim that some predictors which were suggested by
academic research often perform worse than the historical average
excess return on forecasting risk premium of the S\&P500 index, either
in-sample or out-of-sample. Here we re-examine the claim and focus
on the out-of-sample performances of the predictors. The main statistics used in \citet{GW_2008} for evaluating the out-of-sample forecasts are the out-of-sample R-square and difference of the root mean squared errors (dRMSE), which are based on the squared error loss function or its variant. We use the proposed test statistic to see whether the predictors can possibly outperform
the historical average excess return under other consistent loss functions. 

We consider sixteen predictors: (1) the default yield spread (dfy); (2)
inflation (infl); (3) stock variance (svar); 
(4) log dividend payout ratio (de); (5) long term yield (lty); (6) the term
spread (tms); (7) treasury-bill rates (tbl); (8) default return spread
(dfr); (9) log dividend price ratio (dp); (10) log dividend yield (dy); (11)
long term return (ltr); (12) log earnings price ratio (ep); (13) the book-to-market ratio (bm); (14) net equity expansion (ntis); (15) investment to capital ratio (ik); (16) percent equity issuing (eqis). For detailed explanations
on the predictors, please see \citet{GW_2008}. The data have three frequencies: annual (from 1927 to
2015), quarterly (from Q1-1927 to Q4-2015) and monthly (from January-1927
to December-2015).\footnote{For some predictors, their quarterly and/or monthly data are not available. Quarterly data are not available for percent equity issuing (eqis). Monthly data are not available for eqis and investment to capital ratio (ik). In addition, yearly and quarterly data for ik are only available after 1947.} The data set can be downloaded from Amit Goyal's
website: \texttt{http://www.hec.unil.ch/agoyal/}. 

\subsubsection{Single-variable predictive regressions}
The variable to be forecasted is the one-period-ahead risk premium
(expected excess return) of the S\&P500 index. To calculate the excess
return, we use the simple return (including the dividend) of the index
and then subtract the U.S. treasury bill rate from it. We use the
historical average excess return of the S\&P500 index as the benchmark
forecast. The competing forecast is constructed by using a single-variable
linear regression (including the intercept term), which is estimated with the OLS. The forecasts may be viewed as the ones that are generated from misspecified models. 
Thus using different consistent loss functions may yield different ranking results \citep{Patton_2015}. 

We use a rolling window scheme to generate the forecasts. 
The window length for the annual data is 20 years; for the quarterly data,
it is 80 quarters and for the monthly data, it is 240 months. Accordingly,
the forecasting period for the annual data is from 1947 to 2015 (69
years); for the quarterly data, it is from Q1-1947 to Q4-2015 (276
quarters)\footnote{For investment to capital ratio (ik), the forecasting period for the quarterly data is from Q1-1967 to Q4-2015 (196 quarters).} and for the monthly data, it is from January-1947 to December-2015 (828 months). 

In Table \ref{table9}, we show values of the proposed test statistic
for forecasting the conditional expectation (50\%-expectile) and the corresponding
empirical p-values. 
For comparisons, we also show p-values of the Diebold and Marino (DM) test statistic with the squared error loss and the difference of the root mean squared error loss (dRMSE) scaled by 100. The DM test statistic is obtained with the Newey-West standard error of the difference of the squared error loss.

From the table, it can be seen that the proposed test statistic is
not statistically significant at 5\% level, except in three cases
of forecasting the annual risk premium (dp, ik and eqis). 
For the DM test statistic, it is also not statistically significant 5\% level for all cases. These results suggest that there is still weak evidence to say that these predictors can effectively outperform the historical average excess return on forecasting the risk premium of the S\&P500 index, even a much larger class of consistent loss functions are considered for the forecast evaluations. 

\subsubsection{Multivariate predictive regressions}

While the results of the single-variable predictive regressions are overall not positive for the considered predictors, different combinations of them might provide improved outcomes. We next apply the proposed test on a completed list of predictive regressions generated from combinations of the predictors.

Some filtrations are conducted before the empirical analysis. First, we only focus on the cases of quarterly and monthly data since they can provide enough samples for the rolling window estimations when the predictive regressions are multivariate. We also exclude investment to capital ratio (ik) from the predictors since its sample length is shorter than others. Thus for
each of the quarterly and monthly data used here, we have fourteen predictors. Ideally we can have $2^{14}-1=16,383$ predictive regressions generated from combinations of these predictors. However, among the predictors, some of them are a linear combination of others. For example, term spread (tms) equals long term yield (lty) minus treasury-bill rates (tbl), and log earnings price ratio (ep) equals log dividend price ratio (dp) minus log dividend payout ratio (de). When these variables are simultaneously included in a predictive regression, it will result in the problem of muticollineraity in the estimation. Thus we exclude the predictive regressions in which all (lty, tms, tbl) or all (de, dp, ep) are included.

In Figure \ref{figure9} we show ordered values (from small to large) of the relevant four quantities for forecasts obtained from using the multivariate predictive regressions. 
The red crosses in each plot are values of the quantities for the single-variable predictive regressions shown in Table \ref{table9}. As can be seen from the second row of the figure, among these forecasts, only a small proportion of them have a very small p-value. For the quarterly data, only six forecasts generate empirical p-values less than 0.0025;\footnote{Since here we have a large number of candidate predictive regressions, to avoid data snooping and take multiplicity into account, we use a much more restricted criterion for the p-value than the conventional levels 0.05 and 0.01 used in the single-variable predictive regressions.} for the monthly data, the same number is 99. As shown in the third row of the figure, there are also only a few number of forecasts generating a positive dRMSE: for the quarterly data, the number is 4 (two of them are from using the single-variable regressions), and for the monthly data, the number is 13 (two of them are from using the single-variable regressions). For the DM test statistic, the p-values are all above 0.35 (0.18) for the quarterly (monthly) data. 

Finally, in Table \ref{table10} we show frequency that a predictor is included in the predictive regressions whose forecasts have the empirical p-values less than 0.0025, 0.005 and 0.01. Some predictors seem to be more often included in such predictive regressions than others (e.g., dfy and infl for the quarterly data, and dfy and ntis for the monthly data), which suggests that under certain non squared-error loss functions, using these predictors might be helpful on outperforming the historical average excess return on forecasting the risk premium of the S\&P500 index.

\subsection{Forecasting annual growth of U.S. real gross domestic product (RGDP)}
In this subsection, we use the proposed test to compare performances of experts' forecasts on annual growth of U.S. real gross domestic product (RGDP). The extremal consistent loss function used here is for the conditional expectation forecast. The data are from Survey of Professional Forecasters (SPF) conducted by Federal Reserve Bank of Philadelphia. We focus on comparing mean forecast from all experts (SPF average) and an expert's (with ID: 426) individual forecast. We use forecasts for next four quarter-to-quarter growth of U.S. RGDP to calculate forecast for the annual growth. We use both Q3-2017 vintage and the first release data of U.S. RGDP level data to calculate the realized annual growth. 
The sample period for the comparison is from Q1-1991 to Q2-2017 (106 quarters) and all the data used are in quarterly frequency. Figure \ref{figure10} shows time series plots of the Q3-2017 vintage and the first release data for annual growth of 
U.S. RGDP and the two forecasts. 

Upper panel of Table \ref{table11} shows summary statistics for the four time
series. The mean forecast can be viewed as an average of opinions
of the experts who were in the survey. It is known that such ``wisdom
of crowds'' on average has a superior performance than an individual
forecast. Results of our proposed test confirm this. As can be seen
in bottom panel of Table \ref{table11}, when the mean forecast is either the
benchmark or the competitor, empirical p-values of the proposed
test suggest that the mean forecast should at least perform equally well
or better than the individual forecast, no matter whether the Q3-2017
vintage or first release data are used as the realized target random
variable. Furthermore, when the mean forecast is the benchmark, the test result
suggests that underperformance of the individual forecast is insensitive
to the choice of consistent loss function for the conditional expectation forecast.

In upper panel of Figure \ref{figure11}, with the Q3-2017 vintage data, we plot empirical differences of consistent loss functions (SPF average minus ID: 426): exponential and homogeneous Bregman with $\alpha=0.5$, over a range of parameter values.\footnote{The plots for the case of using the first release data are very similar, so they are not shown here.} As can be seen
from the plots, the consistent loss functions chosen here all show
non-positive empirical differences, which are in line with the test results.

\subsection{Estimating Value at Risk of the daily S\&P500 index}
Value at risk (VaR) is an estimated amount of possible investment loss during a certain period. In risk management, the VaR is one of the most important measures used by regulators for quantifying banks' and financial institutions' exposures to risk. Suppose the amount of investment at the end of period $t$ is $I_{t}$ and log return of the investment at period $t+1$ is $R_{t+1}$. At period $t$, the VaR at level $\alpha$ for period $t+1$: $VaR_{\alpha,t+1}$ can be formally defined as the conditional $\alpha-$quantile of $I_{t}\times R_{t+1}$. 
For simplicity, we assume $I_{t}=\$1$ for all $t$ and thus $VaR_{\alpha,t+1}$ is equivalent to the conditional $\alpha-$quantile of $R_{t+1}$. In this subsection, we use the proposed test for conditional quantile forecasts to compare performances of four methods on estimating daily $VaR_{\alpha,t+1}$ of the S\&P500 index. 

The first method is to use sample quantile of an asset's daily log return. The second one is to assume that the asset's daily log return follows a normal distribution and the VaR is calculated with the estimated mean and variance. The two methods are simple and can be viewed as benchmarks on estimating the daily VaR. The third and fourth methods
are based on the conditional autoregressive value at risk (CAViaR)
models of \citet{EM_2004}. In the CAViaR models, $VaR_{\alpha,t+1}$
follows an AR process augmented with a function of a finite number
of lagged observable variables. Here we consider the following two specifications for the CAViaR models: 
\begin{eqnarray}
	VaR_{\alpha,t+1} & = & a+b\times VaR_{\alpha,t}+c\left|R_{t}\right|,\label{caviar_sy}\\
	VaR_{\alpha,t+1} & = & a+b\times VaR_{\alpha,t}+c_{1}\left|R_{t}\right|\mathbf{1}\left\{ R_{t}>0\right\} +c_{2}\left|R_{t}\right|\mathbf{1}\left\{ R_{t}\leq0\right\} .\label{caviar_asy}
\end{eqnarray}
The CAViaR models of (\ref{caviar_sy}) and (\ref{caviar_asy}) are
termed ``symmetric absolute value'' and ``asymmetric slope'' in
\citet{EM_2004}, and thus we use CAViaR-sy and CAViaR-asy to denote
them. Coefficients of the two CAViaR models are estimated with minimizing an average of (empirical) tick loss. 
We solve the minimization problem with the Nelder and Mead simplex algorithm.

We consider $\alpha=0.01,$ 0.025 and 0.05, which are the most often
used VaR levels in practice. All of the four methods are conducted
with a rolling window scheme with window length equal to 500. The
estimated daily $VaR_{\alpha,t+1}$ is generated as an out-of-sample forecast
of the conditional $\alpha-$quantiles of the daily S\&P500 log return. The sample period of the daily S\&P500 index data is from Jan-08-2002
to Dec-29-2017 (4,024 days) and the forecasting period is from Jan-02-2004
to Dec-29-2017 (3,524 days). Figure \ref{figure12} shows time-series
plots of the daily S\&P500 log return and the estimated daily $VaR_{\alpha,t+1}$
generated with CAViaR-sy and CAViaR-asy. Table \ref{table12} presents
summary statistics, hit proportion and value of averaged tick
loss of the estimated daily $VaR_{\alpha,t+1}$ generated with the four
methods and summary statistics of the daily S\&P500 log return. The hit proportion is an average of number of days when the
daily S\&P500 log return is no greater than the estimated daily $VaR_{\alpha,t+1}$,
which estimates the unconditional probability of an exceedance event.
From the table, it can be seen that the two CAViaR models on average
generate a lower value of tick loss than the two simple methods. 

We report values of the proposed test statistic, the corresponding
empirical p-values and p-values of the Diebold-Marino test statistic in Table \ref{table13}. The loss function used for calculating the DM test statistic is
the tick loss. The performances are compared pairwisely. In the table, methods shown in rows are benchmarks and those shown in columns
are competitors in the tests. It can be seen that when the two simple
methods are the benchmarks and the two CAViaR models are the competitors,
under the conventional significant level 0.05, the null hypotheses
are all rejected for the proposed test. But when the two CAViaR models
are the benchmarks and the two simple methods are the competitors, all the null
hypotheses are not rejected under the conventional significant level 0.05 (the smallest corresponding p-value is 0.610). The results suggest that the two CAViaR models perform at least equally well as or better than the two simple methods on estimating the daily $VaR_{\alpha,t+1}$ of the S\&P500 index under all consistent loss functions for forecasting the conditional $\alpha-$quantiles when $\alpha = 0.01$, 0.025 and 0.05. Using the DM test also show similar results. Finally, turning to a comparison of the two CAViaR models themselves, the test results suggest that CAViaR-asy seems to be more adequate than CAViaR-sy on estimating the daily $VaR_{\alpha,t+1}$ when $\alpha=$ 0.025 and 0.05.

\section{Conclusions}

In this paper, we develop statistical tests for evaluating performances of expectile and quantile forecasts of a random variable. Based on
the extremal consistent loss functions proposed by \citet{EGJK_2016}, we construct test statistics for the tests. If the null hypothesis holds, the benchmark forecast will at least perform equally well as the competing one regardless which consistent loss function is used. 
For implementing the tests, we propose to use the re-centered bootstrap to obtain empirical p-values of the test statistics. We derive asymptotic results for the proposed test statistics and for using the stationary bootstrap to construct the empirical p-values. In the simulation study, we show the proposed test statistics work reasonably well under various situations. 

We apply the proposed test on re-examining abilities of some predictors on forecasting risk premiums of the S\&P500 index. When the predictors are used individually, we find that they seldom can outperform the historical average of excess return, no matter which consistent loss functions for forecasting conditional expectation is used for evaluating the forecast performances. When we consider possible combinations of the predictors, for forecasting the quarterly and monthly risk premiums, we find a few number of them might outperform the historical average of excess return under certain consistent loss functions. With the proposed test, we also demonstrate that for forecasting U.S. RGDP annual growth, mean forecasts from all experts has a superior performance than an individual forecast, and the result is insensitive to which consistent loss function for forecasting conditional expectation is chosen. As for comparisons of estimated daily value at risk of the S\&P500 index, results from the proposed test suggest that the CAViaR type models perform better than the two benchmark methods, no matter which consistent loss function for the conditional quantile forecasts is used for the performance evaluations.
\clearpage
\bibliographystyle{ECTA}
\bibliography{ref_scoring_function}

\clearpage
\begin{sidewaystable}
	\caption{Examples for $L^{E}\left(x,y\right)$} 
	\scalebox{.9}{
	\begin{tabular}{cccll}
		\hline 
		$\phi\left(t\right)$  & Domain of $t$  & $L^{E}\left(x,y\right)$  & Name for $L^{E}_{\alpha}\left(x,y\right)$, $\alpha = 0.5$   
		& Reference\tabularnewline
		\hline 
		$t^{2}$  & $t\in\mathbb{R}$  & $\left(x-y\right)^{2}$  & Squared error loss  & -\tabularnewline
		&  &  &  & \tabularnewline
		$t\log(t)+\left(1-t\right)\log(1-t)$  & $t\in [0,1]$  & $-\log x$ if $y=1$, $-\log\left(1-x\right)$ if $y=0$  & Negative log likelihood for $Y\in\{0,1\}$ & -\tabularnewline
		&  &  &  & \tabularnewline
		$\left|t\right|^{b}$, $b>1$  & $t\in\mathbb{R}$  & $\left|y\right|^{b}-\left|x\right|^{b}-b\times sign\left(x\right)\left|x\right|^{b-1}\left(y-x\right)$  & Homogeneous Bregman loss  & Gneiting (2011)\tabularnewline
		&  &  &  & \tabularnewline
		$\frac{1}{a^{2}}\exp(at)$, $a\neq0$  & $t\in\mathbb{R}$  & $\frac{1}{a^{2}}\left[\exp\left(ay\right)-\exp\left(ax\right)\right]-\frac{1}{a}\exp\left(ax\right)\left(y-x\right)$  & Exponential (non-homogeneous)  & Patton (2015)\tabularnewline
		&  &  & Bregman loss  & \tabularnewline
		&  &  &  & \tabularnewline
		$-\log(t)$  & $t>0$  & $\frac{y}{x}-\log(\frac{y}{x})-1$  & QLIKE loss (homogeneous loss  & Patton (2011)\tabularnewline
		&  &  & with order $c=0$)  & \tabularnewline
		&  &  &  & \tabularnewline
		$t\log(t)$  & $t>0$  & $y\log\frac{y}{x}-\left(y-x\right)$  & Homogeneous loss with order $c=1$  & Patton (2011)\tabularnewline
		&  &  &  & \tabularnewline
		$\frac{1}{c^{2}-c}t^{c}$, $c\notin\left\{ 0,1\right\} $  & $t>0$  & $\frac{1}{c^{2}-c}\left(y^{c}-x^{c}\right)-\frac{1}{c-1}x^{c-1}\left(y-x\right)$  & Homogeneous loss with order $c\notin\left\{ 0,1\right\} $  & Patton (2011)\tabularnewline
		&  &  &  & \tabularnewline
		$\left(t-\theta\right)_{+}$, $\theta\in\Theta\subseteq\mathbb{R}$  & $t\in\mathbb{R}$  & $\left(y-\theta\right)_{+}-\left(x-\theta\right)_{+}-1\left\{ \theta<x\right\} \left(y-x\right)$  & Extremal consistent loss (for expectile)  & Ehm et al. (2016)\tabularnewline
		\hline 
	\end{tabular}
}
	\label{table1} 
\end{sidewaystable}

\begin{sidewaystable}
	\caption{Examples for $L^{Q}\left(x,y\right)$}
	\begin{tabular}{cccll}
		\hline 
		$\zeta\left(t\right)$  & Domain of $t$  & $L^{Q}\left(x,y\right)$  & Name for $L^{Q}_{\alpha}\left(x,y\right)$, $\alpha \in (0,1)$  
		& Reference\tabularnewline
		\hline 
		$t$  & $t\in\mathbb{R}$  & $x-y$ & Lin-lin (tick) loss & -\tabularnewline
		&  &  &  & \tabularnewline
		$t^{c}/\left|c\right|$, $c\neq0$ & $t>0$  & $\left(x^{c}-y^{c}\right)/\left|c\right|$  & Homogeneous (power) loss  & Gneiting (2011)\tabularnewline
		&  &  & with order $c\neq0$  & \tabularnewline
		&  &  &  & \tabularnewline
		$\log\left(t\right)$ & $t>0$  & $\log x-\log y$  & Homogeneous (power) loss  & Gneiting (2011)\tabularnewline
		&  &  & with order $c=0$  & \tabularnewline
		&  &  &  & \tabularnewline
		$t/\alpha$  & $t\in\mathbb{R}$ & $\left(x-y\right)/\alpha$ & Scaled lin-lin loss  & Holzmann and\tabularnewline
		&  &  &  & Eulert (2013)\tabularnewline
		&  &  &  & \tabularnewline
		$1\left\{ \theta<t\right\} $, $\theta\in\Theta\subseteq\mathbb{R}$ & $t\in\mathbb{R}$ & $1\left\{ \theta<x\right\} -1\left\{ \theta<y\right\}$ & Extremal consistent loss (for quantile) & Ehm et al. (2016)\tabularnewline
		\hline 
	\end{tabular}
	\label{table2} 
\end{sidewaystable}

\begin{table}
	\caption{The table shows rejection frequencies of the proposed test and the
		Diebold-Marino test with the squared error loss. The critical values
		of the proposed test are constructed by using the re-centered bootstrap.
		The variable to be forecasted is $E_{t}\left[Y_{t+1}\right]$, where
		$Y_{t+1}=\gamma+\beta_{1}W_{1t}+\beta_{2}W_{2t}+\varepsilon_{t+1}$,
		and $W_{1t}\sim i.i.d.N\left(0,\sigma_{W_{1}}^{2}\right)$, $W_{2t}\sim i.i.d.N\left(0,\sigma_{W_{2}}^{2}\right)$
		and $\varepsilon_{t+1}\sim i.i.d.N\left(0,\sigma_{\varepsilon}^{2}\right)$.
		$W_{1t}$, $W_{2t}$ and $\varepsilon_{t+1}$ are mutually independent.
		We set $\gamma=0.4$, $\beta_{1}=0.5$, $\beta_{2}=0.2$ and $\sigma_{W_{1}}^{2}=\sigma_{W_{2}}^{2}=1$.
		The benchmark forecast is $X_{1t}=c_{1}+b_{1}W_{1t}$ and the competing
		forecast is $X_{2t}=c_{2}+b_{2}W_{2t}$. Scenarios (1) to (3) correspond
		to different parameter settings in Section 4.1.1. We report the rejection
		frequencies at three different significant levels: 0.01, 0.05 and
		0.1. We set length of forecast $T_{p}=100$, 300 and 1000, bootstrap
		sample size $M=400$. Each scenario is simulated 1000 times.}
	\centering \scalebox{1}{ %
		\begin{tabular}{lcccccccc}
			\hline 
			&  & \multicolumn{7}{c}{Benchmark: $X_{1t}$, Competitor: $X_{2t}$ }\tabularnewline
			\hline 
			&  & \multicolumn{3}{c}{The proposed test} &  & \multicolumn{3}{c}{DM}\tabularnewline
			\cline{2-5} \cline{7-9} 
			& $T_{P}$  & 0.01  & 0.05  & 0.1  &  & 0.01  & 0.05  & 0.1\tabularnewline
			\hline 
			& 100  & 0.047  & 0.207  & 0.347  &  & 0.011  & 0.052  & 0.120\tabularnewline
			Scenario (1)  & 300  & 0.237  & 0.519  & 0.716  &  & 0.015  & 0.052  & 0.092\tabularnewline
			& 1000  & 0.968  & 1.000  & 1.000  &  & 0.007  & 0.048  & 0.102\tabularnewline
			&  &  &  &  &  &  &  & \tabularnewline
			& 100  & 0.120  & 0.317  & 0.511  &  & 0.097  & 0.272  & 0.397\tabularnewline
			Scenario (2)  & 300  & 0.419  & 0.721  & 0.875  &  & 0.237  & 0.479  & 0.608\tabularnewline
			& 1000  & 0.998  & 1.000  & 1.000  &  & 0.736  & 0.888  & 0.953\tabularnewline
			&  &  &  &  &  &  &  & \tabularnewline
			& 100  & 0.000  & 0.000  & 0.000  &  & 0.000  & 0.000  & 0.000\tabularnewline
			Scenario (3)  & 300  & 0.000  & 0.000  & 0.000  &  & 0.000  & 0.000  & 0.000\tabularnewline
			& 1000  & 0.000  & 0.000  & 0.000  &  & 0.000  & 0.000  & 0.000\tabularnewline
			\hline 
			&  & \multicolumn{7}{c}{Benchmark: $X_{2t}$, Competitor: $X_{1t}$ }\tabularnewline
			\hline 
			&  & \multicolumn{3}{c}{The proposed test} &  & \multicolumn{3}{c}{DM}\tabularnewline
			\cline{2-5} \cline{7-9} 
			& $T_{P}$  & 0.01  & 0.05  & 0.1  &  & 0.01  & 0.05  & 0.1\tabularnewline
			\hline 
			& 100  & 0.362 & 0.611 & 0.721 &  & 0.015 & 0.045 & 0.095\tabularnewline
			Scenario (1)  & 300  & 0.828 & 0.958 & 0.983 &  & 0.007 & 0.057 & 0.122\tabularnewline
			& 1000  & 1.000 & 1.000 & 1.000 &  & 0.012 & 0.057 & 0.105\tabularnewline
			&  &  &  &  &  &  &  & \tabularnewline
			& 100  & 0.217 & 0.479 & 0.599 &  & 0.000 & 0.000 & 0.001\tabularnewline
			Scenario (2)  & 300  & 0.559 & 0.791 & 0.888 &  & 0.000 & 0.000 & 0.001\tabularnewline
			& 1000  & 0.980 & 1.000 & 1.000 &  & 0.000 & 0.000 & 0.000\tabularnewline
			&  &  &  &  &  &  &  & \tabularnewline
			& 100  & 0.611 & 0.845 & 0.908 &  & 0.648 & 0.863 & 0.925\tabularnewline
			Scenario (3)  & 300  & 0.988 & 0.998 & 1.000 &  & 0.993 & 1.000 & 1.000\tabularnewline
			& 1000 & 1.000 & 1.000 & 1.000 &  & 1.000 & 1.000 & 1.000\tabularnewline
			\hline 
		\end{tabular}} \label{table3} 
	\end{table}

\begin{sidewaystable}
	\caption{The table shows rejection frequencies of the proposed test when critical
		values are constructed by using the re-centered bootstrap. The variable
		to be forecasted is the conditional $\alpha-$expectile of $Y_{t+1}$,
		where $Y_{t+1}|\mu_{t+1|t}\sim N\left(\mu_{t+1|t},1\right)$ and $\mu_{t+1|t}\sim i.i.d.N\left(0,1\right)$.
		We consider $\alpha=0.01$, 0.05 and 0.5. The benchmark forecast is
		$X_{1t} =\mu_{t+1|t}+e^{Z}\left(\alpha\right)+\varsigma\left(\alpha\right)Z_{1t},$
		where $e^{Z}\left(\alpha\right)$ is the $\alpha-$expectile of the
		standard normal random variable $Z$, $\varsigma\left(\alpha\right)=\sqrt{E\left[\left(1\left\{ Z<e^{Z}\left(\alpha\right)\right\}-\alpha\right)^{2}\left(Z-e^{Z}\left(\alpha\right)\right)^{2}\right]}/E\left[\left|1\left\{ Z<e^{Z}\left(\alpha\right)\right\}-\alpha\right|\right].$
		and $Z_{1t}\sim N\left(0,0.25\right).$ The first column shows six
		competing forecasts $X_{2t}$. Here $Z_{it}\sim N\left(0,\sigma_{i}^{2}\right)$
		and $\sigma_{i}^{2}=0.04,$ 0.25, 1 for $i=2$, 3, 4. We report the
		rejection frequencies at three different significant levels: 0.01,
		0.05 and 0.1. We set length of forecast $T_{p}=100$, 300 and 1000
		and bootstrap sample size $M=400$. Each scenario is simulated 1000
		times.}
	\centering %
	\begin{tabular}{lcccccccccccccc}
		\hline 
		&  &  &  & \multicolumn{3}{c}{$T_{p}=100$} &  & \multicolumn{3}{c}{$T_{p}=300$} &  & \multicolumn{3}{c}{$T_{p}=1000$}\tabularnewline
		\cline{5-7} \cline{9-11} \cline{13-15} 
		$X_{2t}$  &  & $\alpha$  &  & 0.01  & 0.05  & 0.1  &  & 0.01  & 0.05  & 0.1  &  & 0.01  & 0.05  & 0.1\tabularnewline
		\hline 
		&  & 0.01  &  & 0.002 & 0.062 & 0.314 &  & 0.080 & 0.519 & 0.800 &  & 0.963 & 1.000 & 1.000\tabularnewline
		$\mu_{t+1|t}+e^{Z}(\alpha)$  &  & 0.05  &  & 0.005 & 0.100 & 0.374 &  & 0.147 & 0.596 & 0.868 &  & 0.980 & 1.000 & 1.000\tabularnewline
		&  & 0.5  &  & 0.036 & 0.226 & 0.386 &  & 0.380 & 0.709 & 0.877 &  & 0.971 & 1.000 & 1.000\tabularnewline
		&  &  &  &  &  &  &  &  &  &  &  &  &  & \tabularnewline
		&  & 0.01  &  & 0.000 & 0.082 & 0.282 &  & 0.065 & 0.454 & 0.788 &  & 0.900 & 0.998 & 1.000\tabularnewline
		$\mu_{t+1|t}+e^{Z}(\alpha)+\varsigma\left(\alpha\right)Z_{2t}$  &  & 0.05  &  & 0.000 & 0.050 & 0.282 &  & 0.085 & 0.459 & 0.731 &  & 0.828 & 0.995 & 1.000\tabularnewline
		&  & 0.5  &  & 0.031 & 0.121 & 0.295 &  & 0.240 & 0.555 & 0.736 &  & 0.876 & 0.985 & 0.995\tabularnewline
		&  &  &  &  &  &  &  &  &  &  &  &  &  & \tabularnewline
		&  & 0.01  &  & 0.000 & 0.012 & 0.070 &  & 0.002 & 0.025 & 0.087 &  & 0.015 & 0.052 & 0.102\tabularnewline
		$\mu_{t+1|t}+e^{Z}(\alpha)+\varsigma\left(\alpha\right)Z_{3t}$ (l.f.c.)  &  & 0.05  &  & 0.000 & 0.010 & 0.077 &  & 0.002 & 0.052 & 0.112 &  & 0.007 & 0.062 & 0.107\tabularnewline
		&  & 0.5  &  & 0.011 & 0.026 & 0.061 &  & 0.014 & 0.053 & 0.100 &  & 0.015 & 0.066 & 0.105\tabularnewline
		&  &  &  &  &  &  &  &  &  &  &  &  &  & \tabularnewline
		&  & 0.01  &  & 0.000 & 0.000 & 0.000 &  & 0.000 & 0.000 & 0.000 &  & 0.000 & 0.000 & 0.000\tabularnewline
		$\mu_{t+1|t}+e^{Z}(\alpha)+\varsigma\left(\alpha\right)Z_{4t}$  &  & 0.05  &  & 0.000 & 0.000 & 0.000 &  & 0.000 & 0.000 & 0.000 &  & 0.000 & 0.000 & 0.000\tabularnewline
		&  & 0.5  &  & 0.000 & 0.000 & 0.000 &  & 0.000 & 0.000 & 0.000 &  & 0.000 & 0.000 & 0.000\tabularnewline
		&  &  &  &  &  &  &  &  &  &  &  &  &  & \tabularnewline
		&  & 0.01  &  & 0.000 & 0.002 & 0.012 &  & 0.000 & 0.002 & 0.007 &  & 0.000 & 0.001 & 0.002\tabularnewline
		$e^{Z}(\alpha)+\varsigma\left(\alpha\right)Z_{3t}$  &  & 0.05  &  & 0.000 & 0.002 & 0.002 &  & 0.000 & 0.000 & 0.000 &  & 0.000 & 0.000 & 0.000\tabularnewline
		&  & 0.5  &  & 0.000 & 0.000 & 0.000 &  & 0.000 & 0.000 & 0.000 &  & 0.000 & 0.000 & 0.000\tabularnewline
		&  &  &  &  &  &  &  &  &  &  &  &  &  & \tabularnewline
		&  & 0.01  &  & 0.000 & 0.000 & 0.000 &  & 0.000 & 0.000 & 0.000 &  & 0.000 & 0.000 & 0.000\tabularnewline
		$e^{Z}(\alpha)+\varsigma\left(\alpha\right)Z_{4t}$  &  & 0.05  &  & 0.000 & 0.000 & 0.000 &  & 0.000 & 0.000 & 0.000 &  & 0.000 & 0.000 & 0.000\tabularnewline
		&  & 0.5  &  & 0.000 & 0.000 & 0.000 &  & 0.000 & 0.000 & 0.000 &  & 0.000 & 0.000 & 0.000\tabularnewline
		\hline 
	\end{tabular}\label{table4} 
\end{sidewaystable}

\begin{table}
\caption{The table shows rejection frequencies of the proposed test when critical
values are constructed by using the re-centered bootstrap. The variable
to be forecasted is $E_{t}\left[Y_{t+1}\right]$. Data generating
processes for the relevant variables $Y_{t+1}$, $W_{1,t+1}$ and
$W_{2,t+1}$ are shown in Section 4.1.3. The benchmark forecast is
$X_{1t}:=f_{1,t+1|t}=\left(\hat{\gamma_{t}}+Z_{1t}\right)+\left(\hat{\beta}_{1t}+Z_{2t}\right),$
where $\hat{\gamma}_{t}$ and $\hat{\beta}_{1t}$ are the coefficients
estimated from using the OLS and rolling window scheme with window
length $T_{R}=100$, $Z_{1t}\sim i.i.d.N\left(0,0.0025\right)$ and
$Z_{2t}\sim i.i.d.N\left(0,0.0225\right)$. The first column shows
seven competing forecasts $X_{2t}:=f_{2,t+1|t}$. We report the
rejection frequencies at three different significant levels: 0.01,
0.05 and 0.1. We set length of forecast $T_{p}=100$, 300 and 1000
and bootstrap sample size $M=400$. Each scenario is simulated 1000
times.}
\centering \scalebox{0.9}{ %
\begin{tabular}{lcccccccccccc}
\hline 
 &  & \multicolumn{3}{c}{$T_{p}=100$} &  & \multicolumn{3}{c}{$T_{p}=300$} &  & \multicolumn{3}{c}{$T_{p}=1000$}\tabularnewline
\cline{3-5} \cline{7-9} \cline{11-13} 
$X_{2t}$  &  & 0.01  & 0.05  & 0.1  &  & 0.01  & 0.05  & 0.1  &  & 0.01  & 0.05  & 0.1\tabularnewline
\hline 
$\tilde{\gamma}_{t}+\tilde{\beta}_{1t}Y_{t}$ (l.f.c.)  &  & 0.011  & 0.051  & 0.125  &  & 0.015  & 0.049  & 0.086  &  & 0.011  & 0.054  & 0.101\tabularnewline
$\hat{\gamma_{t}}+\hat{\beta}_{1t}Y_{t}+\hat{\beta}_{2t}^{low}W_{1t}$  &  & 0.051  & 0.146  & 0.245  &  & 0.066  & 0.177  & 0.297  &  & 0.124  & 0.352  & 0.543\tabularnewline
$\hat{\gamma_{t}}+\hat{\beta}_{1t}Y_{t}+\hat{\beta}_{2t}^{med}W_{1t}$  &  & 0.413  & 0.721  & 0.869  &  & 0.881  & 0.985  & 1.000  &  & 1.000  & 1.000  & 1.000\tabularnewline
$\hat{\gamma_{t}}+\hat{\beta}_{1t}Y_{t}+\hat{\beta}_{2t}^{high}W_{1t}$  &  & 0.705  & 0.918  & 0.989  &  & 0.997  & 1.000  & 1.000  &  & 1.000  & 1.000  & 1.000\tabularnewline
$\hat{\gamma_{t}}+\hat{\beta}_{1t}Y_{t}+\hat{\beta}_{3t}W_{2t}^{lcr}$  &  & 0.025  & 0.126  & 0.241  &  & 0.025  & 0.176  & 0.292  &  & 0.134  & 0.383  & 0.525\tabularnewline
$\hat{\gamma_{t}}+\hat{\beta}_{1t}Y_{t}+\hat{\beta}_{3t}W_{2t}^{hcr}$  &  & 0.192  & 0.465  & 0.662  &  & 0.503  & 0.805  & 0.922  &  & 0.991  & 1.000  & 1.000\tabularnewline
\hline 
\end{tabular}} \label{table5} 
\end{table}

\begin{table}
\caption{The table shows rejection frequencies of the proposed test when critical
values are constructed by using the re-centered bootstrap. The variable
to be forecasted is $E_{t}\left[Y_{t+1}\right]$, where $Y_{t+1}=V_{t+1}^{2}$,
$V_{t+1}\sim i.i.d.N\left(0,\sigma_{t+1|t}^{2}\right)$. Data generating
processes for the relevant variables $V_{t+1}$ and $\sigma_{t+1|t}^{2}$
are shown in Section 4.1.4. The benchmark forecast is $X_{1t}:=f_{1,t+1|t}=\exp(-0.045)U_{1t}Y_{t},$
where $\ln U_{1t}\sim i.i.d.N\left(0,0.09\right)$. The first column
shows four competing forecasts $X_{2t}:=f_{2,t+1|t}$. We report
the rejection frequencies at three different significant levels: 0.01,
0.05 and 0.1. We set length of forecast $T_{p}=100$, 300 and 1000
and bootstrap sample size $M=400$. Each scenario is simulated 1000
times.}
\centering %
\begin{tabular}{lcccccccccccc}
\hline 
 &  & \multicolumn{3}{c}{$T_{p}=100$} &  & \multicolumn{3}{c}{$T_{p}=300$} &  & \multicolumn{3}{c}{$T_{p}=1000$}\tabularnewline
\cline{3-5} \cline{7-9} \cline{11-13} 
$X_{2t}$  &  & 0.01  & 0.05  & 0.1  &  & 0.01  & 0.05  & 0.1  &  & 0.01  & 0.05  & 0.1\tabularnewline
\hline 
$\exp(-0.045)U_{2t}Y_{t}$ (l.f.c.)  &  & 0.005  & 0.031  & 0.082  &  & 0.000  & 0.016  & 0.051  &  & 0.010  & 0.027  & 0.054\tabularnewline
$\hat{\sigma}_{t+1|t}^{2}\left(0,1\right)$  &  & 0.267  & 0.564  & 0.758  &  & 0.645  & 0.891  & 0.953  &  & 0.903  & 0.960  & 0.971\tabularnewline
$\hat{\sigma}_{t+1|t}^{2}\left(1,1\right)$  &  & 0.281  & 0.601  & 0.881  &  & 0.645  & 0.870  & 0.965  &  & 0.883  & 0.956  & 0.977\tabularnewline
$\hat{\sigma}_{t+1|t}^{2}\left(2,2\right)$  &  & 0.273  & 0.602  & 0.875  &  & 0.633  & 0.881  & 0.965  &  & 0.878  & 0.954  & 0.975\tabularnewline
\hline 
\end{tabular}\label{table6} 
\end{table}

\begin{sidewaystable}
\caption{The table shows rejection frequencies of the proposed test when critical
values are constructed by using the re-centered bootstrap. The variable
to be forecasted is the conditional $\alpha-$quantile of $Y_{t+1}$,
where $Y_{t+1}|\mu_{t+1|t}\sim N\left(\mu_{t+1|t},1\right)$ and $\mu_{t+1|t}\sim i.i.d.N\left(0,1\right)$.
We consider $\alpha=0.01$, 0.05 and 0.5. The benchmark forecast is
$X_{1t}:=f_{1,t+1|t}=\mu_{t+1|t}+\Phi^{-1}\left(\alpha\right)+\xi\left(\alpha\right)Z_{1t},$
where $\xi\left(\alpha\right)=\sqrt{\alpha\left(1-\alpha\right)}/\phi\left(\Phi^{-1}\left(\alpha\right)\right)$
and $Z_{1t}\sim N\left(0,0.25\right).$ The first column shows six
competing forecasts $X_{2t}:=f_{2,t+1|t}$. Here $Z_{it}\sim N\left(0,\sigma_{i}^{2}\right)$
and $\sigma_{i}^{2}=0.04,$ 0.25, 1 for $i=2$, 3, 4. We report the
rejection frequencies at three different significant levels: 0.01,
0.05 and 0.1. We set length of forecast $T_{p}=100$, 300 and 1000
and bootstrap sample size $M=400$. Each scenario is simulated 1000
times.}
\centering %
\begin{tabular}{lcccccccccccccc}
\hline 
 &  &  &  & \multicolumn{3}{c}{$T_{p}=100$} &  & \multicolumn{3}{c}{$T_{p}=300$} &  & \multicolumn{3}{c}{$T_{p}=1000$}\tabularnewline
\cline{5-7} \cline{9-11} \cline{13-15} 
$X_{2t}$  &  & $\alpha$  &  & 0.01  & 0.05  & 0.1  &  & 0.01  & 0.05  & 0.1  &  & 0.01  & 0.05  & 0.1\tabularnewline
\hline 
 &  & 0.01  &  & 0.065  & 0.441  & 0.713  &  & 0.865  & 0.988  & 1.000  &  & 1.000  & 1.000  & 1.000\tabularnewline
$\mu_{t+1|t}+\Phi^{-1}\left(\alpha\right)$  &  & 0.05  &  & 0.027  & 0.209  & 0.491  &  & 0.446  & 0.825  & 0.958  &  & 1.000  & 1.000  & 1.000\tabularnewline
 &  & 0.5  &  & 0.115  & 0.382  & 0.566  &  & 0.554  & 0.815  & 0.915  &  & 0.998  & 1.000  & 1.000\tabularnewline
 &  &  &  &  &  &  &  &  &  &  &  &  &  & \tabularnewline
 &  & 0.01  &  & 0.085  & 0.411  & 0.643  &  & 0.830  & 0.973  & 0.998  &  & 1.000  & 1.000  & 1.000\tabularnewline
$\mu_{t+1|t}+\Phi^{-1}\left(\alpha\right)+\xi\left(\alpha\right)Z_{2t}$  &  & 0.05  &  & 0.025  & 0.190  & 0.387  &  & 0.297  & 0.706  & 0.853  &  & 0.978  & 1.000  & 1.000\tabularnewline
 &  & 0.5  &  & 0.070  & 0.264  & 0.411  &  & 0.299  & 0.618  & 0.768  &  & 0.875  & 0.983  & 0.998\tabularnewline
 &  &  &  &  &  &  &  &  &  &  &  &  &  & \tabularnewline
 &  & 0.01  &  & 0.002  & 0.042  & 0.077  &  & 0.020  & 0.050  & 0.107  &  & 0.007  & 0.055  & 0.102\tabularnewline
$\mu_{t+1|t}+\Phi^{-1}\left(\alpha\right)+\xi\left(\alpha\right)Z_{3t}$
(l.f.c.)  &  & 0.05  &  & 0.000  & 0.022  & 0.092  &  & 0.002  & 0.047  & 0.095  &  & 0.007  & 0.050  & 0.087\tabularnewline
 &  & 0.5  &  & 0.010  & 0.027  & 0.062  &  & 0.005  & 0.047  & 0.097  &  & 0.012  & 0.042  & 0.085\tabularnewline
 &  &  &  &  &  &  &  &  &  &  &  &  &  & \tabularnewline
 &  & 0.01  &  & 0.000  & 0.000  & 0.000  &  & 0.000  & 0.000  & 0.000  &  & 0.000  & 0.000  & 0.000\tabularnewline
$\mu_{t+1|t}+\Phi^{-1}\left(\alpha\right)+\xi\left(\alpha\right)Z_{4t}$  &  & 0.05  &  & 0.000  & 0.000  & 0.000  &  & 0.000  & 0.000  & 0.000  &  & 0.000  & 0.000  & 0.000\tabularnewline
 &  & 0.5  &  & 0.002  & 0.002  & 0.005  &  & 0.000  & 0.000  & 0.000  &  & 0.000  & 0.000  & 0.000\tabularnewline
 &  &  &  &  &  &  &  &  &  &  &  &  &  & \tabularnewline
 &  & 0.01  &  & 0.000  & 0.002  & 0.030  &  & 0.000  & 0.007  & 0.030  &  & 0.002  & 0.017  & 0.027\tabularnewline
$\Phi^{-1}\left(\alpha\right)+\xi\left(\alpha\right)Z_{3t}$  &  & 0.05  &  & 0.000  & 0.005  & 0.007  &  & 0.000  & 0.000  & 0.002  &  & 0.000  & 0.007  & 0.015\tabularnewline
 &  & 0.5  &  & 0.000  & 0.000  & 0.000  &  & 0.000  & 0.000  & 0.000  &  & 0.000  & 0.000  & 0.000\tabularnewline
 &  &  &  &  &  &  &  &  &  &  &  &  &  & \tabularnewline
 &  & 0.01  &  & 0.000  & 0.000  & 0.000  &  & 0.000  & 0.000  & 0.000  &  & 0.000  & 0.000  & 0.000\tabularnewline
$\Phi^{-1}\left(\alpha\right)+\xi\left(\alpha\right)Z_{4t}$  &  & 0.05  &  & 0.000  & 0.000  & 0.000  &  & 0.000  & 0.000  & 0.000  &  & 0.000  & 0.000  & 0.000\tabularnewline
 &  & 0.5  &  & 0.000  & 0.000  & 0.000  &  & 0.000  & 0.000  & 0.000  &  & 0.000  & 0.000  & 0.000\tabularnewline
\hline 
\end{tabular}\label{table7} 
\end{sidewaystable}

\begin{sidewaystable}
	\caption{The table shows rejection frequencies of the proposed test when critical
		values are constructed by using the re-centered bootstrap.
		The variable to be forecasted is the conditional $\alpha-$quantile
		of $Y_{t+1}$, where $Y_{t+1}=0.5+1.2W_{1t}+1.5W_{2t}+\varepsilon_{t+1}$,
		where $W_{1t}$, $W_{2t}$ and $\varepsilon_{t+1}$ are i.i.d. and
		each follows $N\left(0,1\right)$. We consider $\alpha=0.01$, 0.05
		and 0.5. The benchmark forecast is $X_{1t}:=f_{1,t+1|t}=\hat{\gamma_{t}}+\hat{\beta}_{1t}W_{1t}+\hat{q}_{t}^{\varepsilon}(\alpha)+Z_{1t}$,
		where $\hat{q}_{t}^{\varepsilon}(\alpha)$ is the empirical quantile of residuals estimated at period $t$ and $Z_{1t}\sim N\left(0,1\right).$ The first column shows five competing forecasts $X_{2t}:=f_{2,t+1|t}$. We report the rejection frequencies at three different significant levels: 0.01, 0.05 and 0.1. We set length of forecast $T_{p}=100$, 300 and 1000 and bootstrap sample size $M=400$. Each scenario is simulated 1000 times.}
	\centering %
	\begin{tabular}{lcccccccccccccc}
		\hline 
		&  &  &  & \multicolumn{3}{c}{$T_{p}=100$} &  & \multicolumn{3}{c}{$T_{p}=300$} &  & \multicolumn{3}{c}{$T_{p}=1000$}\tabularnewline
		\cline{5-7} \cline{9-11} \cline{13-15} 
		$X_{2t}$  &  & $\alpha$  &  & 0.01  & 0.05  & 0.1  &  & 0.01  & 0.05  & 0.1  &  & 0.01  & 0.05  & 0.1\tabularnewline
		\hline 
		&  & 0.01  &  & 0.002 & 0.037 & 0.142 &  & 0.000 & 0.012 & 0.055 &  & 0.010 & 0.040 & 0.077\tabularnewline
		$\hat{\gamma_{t}}+\hat{\beta}_{1t}W_{1t}+\hat{q}_{t}^{\varepsilon}(\alpha)+Z_{1t}$
		(l.f.c.) &  & 0.05  &  & 0.000 & 0.042 & 0.140 &  & 0.000 & 0.047 & 0.107 &  & 0.010 & 0.042 & 0.102\tabularnewline
		&  & 0.5  &  & 0.020 & 0.060 & 0.120 &  & 0.015 & 0.060 & 0.100 &  & 0.017 & 0.047 & 0.092\tabularnewline
		&  &  &  &  &  &  &  &  &  &  &  &  &  & \tabularnewline
		&  & 0.01  &  & 0.002 & 0.042 & 0.165 &  & 0.007 & 0.070 & 0.224 &  & 0.135 & 0.516 & 0.733\tabularnewline
		$\hat{\gamma_{t}}+\hat{\beta}_{1t}W_{1t}+\hat{q}_{t}^{\varepsilon}(\alpha)$ &  & 0.05  &  & 0.017 & 0.137 & 0.302 &  & 0.060 & 0.287 & 0.526 &  & 0.524 & 0.873 & 0.958\tabularnewline
		&  & 0.5  &  & 0.127 & 0.414 & 0.594 &  & 0.429 & 0.713 & 0.853 &  & 0.970 & 1.000 & 1.000\tabularnewline
		&  &  &  &  &  &  &  &  &  &  &  &  &  & \tabularnewline
		&  & 0.01  &  & 0.022 & 0.269 & 0.591 &  & 0.137 & 0.541 & 0.826 &  & 0.960 & 1.000 & 1.000\tabularnewline
		$\hat{\gamma_{t}}+\hat{\beta}_{1t}W_{1t}+\hat{\beta}_{2t}W_{2t}+\hat{q}_{t}^{\varepsilon}(\alpha)$  &  & 0.05  &  & 0.354 & 0.788 & 0.925 &  & 0.955 & 1.000 & 1.000 &  & 1.000 & 1.000 & 1.000\tabularnewline
		&  & 0.5  &  & 0.991 & 1.000 & 1.000 &  & 1.000 & 1.000 & 1.000 &  & 1.000 & 1.000 & 1.000\tabularnewline
		&  &  &  &  &  &  &  &  &  &  &  &  &  & \tabularnewline
		&  & 0.01  &  & 0.025 & 0.289 & 0.606 &  & 0.165 & 0.531 & 0.820 &  & 0.945 & 0.998 & 1.000\tabularnewline
		$\hat{\gamma_{t}}+\hat{\beta}_{1t}W_{1t}+\beta_{2}W_{2t}+\hat{q}_{t}^{\varepsilon}(\alpha)$  &  & 0.05  &  & 0.327 & 0.781 & 0.933 &  & 0.963 & 0.998 & 1.000 &  & 1.000 & 1.000 & 1.000\tabularnewline
		&  & 0.5  &  & 0.988 & 1.000 & 1.000 &  & 1.000 & 1.000 & 1.000 &  & 1.000 & 1.000 & 1.000\tabularnewline
		&  &  &  &  &  &  &  &  &  &  &  &  &  & \tabularnewline
		&  & 0.01  &  & 0.025 & 0.307 & 0.631 &  & 0.160 & 0.554 & 0.828 &  & 0.965 & 1.000 & 1.000\tabularnewline
		$\gamma+\beta_{1}W_{1t}+\beta_{2}W_{2t}+\hat{q}_{t}^{\varepsilon}(\alpha)$  &  & 0.05  &  & 0.357 & 0.805 & 0.928 &  & 0.968 & 1.000 & 1.000 &  & 1.000 & 1.000 & 1.000\tabularnewline
		&  & 0.5  &  & 0.995 & 1.000 & 1.000 &  & 1.000 & 1.000 & 1.000 &  & 1.000 & 1.000 & 1.000\tabularnewline
		\hline
	\end{tabular}\label{table8} 
\end{sidewaystable}

\begin{sidewaystable}
	\caption{The table shows the value of the proposed test statistic for the conditional
		expectation (50\%-expectile), the corresponding empirical p-value, difference
		of the root mean squared loss function (dRMSE) and the p-value of the Diebold
		and Marino test statistic with the squared loss function (DM) for
		testing predictability of the risk premium of the S\&P500 index. The empirical p-value
		is obtained from using the re-centered bootstrap with bootstrap sample
		size $M=400$. For length of forecasting periods: annual: 69 years
		(1947 to 2015); quarterly: 276 quarters (Q1-1947 to Q4-2015)
		and monthly: 828 months (Jan-1947 to Dec-2015).
		}
	\centering %
	\begin{tabular}{lcccccccccccccc}
		\hline 
		& \multicolumn{4}{c}{Annual: 1927 to 2015} &  & \multicolumn{4}{c}{Quarterly: Q1-1927 to Q4-2015 } &  & \multicolumn{4}{c}{Monthly: Jan-1927 to Dec-2015}\tabularnewline
		\cline{2-15} 
		& Test stat.  & p-value  & dRMSE  & DM  &  & Test Stat.  & p-value  & dRMSE  & DM  &  & Test Stat.  & p-value  & dRMSE  & DM\tabularnewline
		\hline 
		dfy  & 0.058  & 0.355  & -1.806  & 0.922  &  & 0.015  & 0.795  & -0.332  & 0.979  &  & 0.007  & 0.838  & -0.039  & 0.856 \tabularnewline
		infl  & 0.074  & 0.120  & -1.251  & 0.976  &  & 0.031  & 0.320  & -0.056  & 0.690  &  & 0.026  & 0.060  & -0.003  & 0.551 \tabularnewline
		svar  & 0.041  & 0.753  & -1.530  & 0.983  &  & 0.031  & 0.293  & -0.854  & 0.834  &  & 0.002  & 0.983  & -0.069  & 0.949 \tabularnewline
		de  & 0.055  & 0.368  & -0.898  & 0.974  &  & 0.032  & 0.313  & -0.138  & 0.919  &  & 0.027  & 0.058  & -0.023  & 0.766 \tabularnewline
		lty  & 0.114  & 0.063  & -1.178  & 0.966  &  & 0.043  & 0.178  & -0.101  & 0.839  &  & 0.024  & 0.125  & -0.011  & 0.647 \tabularnewline
		tms  & 0.039  & 0.693  & -1.394  & 0.938  &  & 0.046  & 0.158  & -0.060  & 0.681  &  & 0.026  & 0.125  & -0.012  & 0.697 \tabularnewline
		tbl  & 0.098  & 0.110  & -1.650  & 0.971  &  & 0.040  & 0.195  & -0.191  & 0.871  &  & 0.021  & 0.255  & -0.035  & 0.875 \tabularnewline
		dfr  & 0.079  & 0.168  & -1.070  & 0.963  &  & 0.014  & 0.858  & -0.193  & 0.989  &  & 0.009  & 0.755  & -0.030  & 0.937 \tabularnewline
		dp  & 0.133  & \textbf{0.045}  & 0.615  & 0.204  &  & 0.050  & 0.140  & 0.039  & 0.355  &  & 0.024  & 0.255  & 0.004  & 0.425 \tabularnewline
		dy  & 0.131  & 0.078  & -0.282  & 0.629  &  & 0.048  & 0.183  & 0.037  & 0.400  &  & 0.023  & 0.283  & 0.010  & 0.334 \tabularnewline
		ltr  & 0.075  & 0.265  & -0.838  & 0.939  &  & 0.028  & 0.453  & -0.044  & 0.718  &  & 0.021  & 0.233  & -0.002  & 0.540 \tabularnewline
		ep  & 0.107  & 0.125  & -0.417  & 0.704  &  & 0.028  & 0.445  & -0.295  & 0.918  &  & 0.024  & 0.128  & -0.031  & 0.795 \tabularnewline
		b.m  & 0.113  & 0.110  & -0.767  & 0.746  &  & 0.018  & 0.728  & -0.139  & 0.847  &  & 0.014  & 0.585  & -0.028  & 0.888 \tabularnewline
		ntis  & 0.071  & 0.228  & -0.733  & 0.945  &  & 0.017  & 0.715  & -0.205  & 0.938  &  & 0.014  & 0.543  & -0.024  & 0.858\tabularnewline
		ik  & 0.135  & \textbf{0.028}  & 0.136  & 0.433  &  & 0.055  & 0.078  & 0.053  & 0.302  &  & -  & -  & -  & -\tabularnewline
		eqis  & 0.150  & \textbf{0.015}  & -0.076  & 0.535  &  & -  & -  & -  & -  &  & -  & -  & -  & -\tabularnewline
		\hline 
	\end{tabular}\label{table9} 
\end{sidewaystable}
\begin{table}
	\caption{The table shows frequencies that a predictor is included in the predictive
		regressions whose forecasts have the empirical p-values less than
		0.0025, 0.005 and 0.01. For the quarterly data, there are 6, 9 and 31 predictive
		regressions whose forecasts have the empirical p-values less than
		0.0025, 0.005 and 0.01; for the monthly data, the numbers are 99,
		192 and 422. }
	\centering
	\scalebox{0.9}{
	\begin{tabular}{cccccccc}
		\hline 
		& \multicolumn{3}{c}{Quarterly data} &  & \multicolumn{3}{c}{Monthly }\tabularnewline
		\cline{2-4} \cline{6-8} 
		& $\leq0.0025$ (6)  & $\leq0.005$ (9)  & $\leq0.01$(31)  &  & $\leq0.0025$ (99)  & $\leq0.005$ (192) & $\leq0.01$ (422)\tabularnewline
		\hline 
		dfy  & 6  & 8 & 24  &  & 87  & 162 & 348\tabularnewline
		infl  & 6  & 9 & 25  &  & 79  & 154 & 323\tabularnewline
		svar  & 0  & 2 & 10  &  & 30  & 51 & 114\tabularnewline
		de  & 5  & 6 & 21  &  & 59  & 114 & 247\tabularnewline
		lty  & 4  & 5 & 22  &  & 73  & 135 & 282\tabularnewline
		tms  & 3  & 6 & 17  &  & 36  & 66 & 154\tabularnewline
		tbl  & 5  & 7 & 22  &  & 38  & 76 & 165\tabularnewline
		dfr  & 0  & 0 & 2  &  & 69  & 128 & 275\tabularnewline
		dp  & 3  & 6 & 20  &  & 39  & 78 & 177\tabularnewline
		dy  & 1  & 2 & 8  &  & 48  & 97 & 208\tabularnewline
		ltr  & 0  & 0 & 0  &  & 13  & 26 & 71\tabularnewline
		ep  & 4  & 6 & 19  &  & 55  & 112 & 248\tabularnewline
		b.m  & 0  & 0 & 1  &  & 6  & 14 & 32\tabularnewline
		ntis  & 0  & 0 & 4  &  & 87  & 160 & 324\tabularnewline
		\hline 
	\end{tabular}}
	\label{table10} 
\end{table}

\begin{table}
		\caption{Upper panel of the table shows summary statistics of the Q3-2017 vintage and first release data for annual growth of U.S. real gross domestic product (RGDP) and two corresponding forecasts from Survey of Professional Forecasters conducted by Fed. Philadelphia: mean forecast from all experts (SPF average) and a forecast from an expert with ID. 426 (ID: 426). Bottom panel shows results of the proposed test when either SPF average or ID: 426 is the benchmark forecast. Both Q3-2017 vintage and first release data are used as the realized value of the target random variable. The data is in quarterly frequency and sample period is from Q1-1991 to Q2-2017 (106 quarters). }
		\centering
\begin{tabular}{lccccc}
	\hline 
	\multicolumn{6}{c}{Summary statistics}\tabularnewline
	\hline 
	&Q3-2017 vintage & First release &  & SPF average & ID: 426\tabularnewline
	\hline 
	Mean & 2.438 & 2.383 &  & 2.747 & 2.617\tabularnewline
	Std. & 1.775 & 1.428 &  & 0.530 & 0.786\tabularnewline
	Min. & -4.062 & -2.832 &  & 0.806 & 0.464\tabularnewline
	Max. & 5.266 & 5.300 &  & 4.006 & 4.198\tabularnewline
	\hline 
	\multicolumn{6}{c}{}\tabularnewline
	\multicolumn{6}{c}{Test results}\tabularnewline
	\hline 
	& \multicolumn{2}{c}{Q3-2017 vintage} &  & \multicolumn{2}{c}{First release}\tabularnewline
	\cline{2-3} \cline{5-6} 
	& Test stat. & p-value &  & Test stat. & p-value\tabularnewline
	\hline 
	$X_{1t}:$ SPF average & 0.000 & 1.000 &  & 0.000 & 1.000\tabularnewline
	$X_{2t}:$ ID: 426 &  &  &  &  & \tabularnewline
	&  &  &  &  & \tabularnewline
	$X_{1t}:$ ID: 426 & 2.380 & 0.010 &  & 1.917 & 0.012\tabularnewline
	$X_{2t}:$ SPF average &  &  &  &  & \tabularnewline
	\hline 
\end{tabular}
\label{table11}
\end{table}

\begin{table}
	\caption{The table shows summary statistics, hit proportion and value of averaged tick loss of the estimated $VaR_{\alpha,t+1}$ generated with the four methods and summary statistics of of the daily S\&P500 log return. The summary statistics are shown in percentage. The whole sample period of the daily S\&P500 log return is from Jan-08-2002 to Dec-29-2017 (4,024 days) and the forecast period is from Jan-02-2004 to Dec-29-2017 (3,524 days).}
	\centering
\begin{tabular}{lcccccc}
	\hline 
	& Mean & Std. & Min. & Max. & Hit prop. & tick loss\tabularnewline
	\hline 
	S\&P500 return & 0.021 & 1.192 & -9.470 & 10.957 & - & -\tabularnewline
	(whole sample) &  &  &  &  &  & \tabularnewline
	S\&P500 return & 0.025 & 1.162 & -9.470 & 10.957 & - & -\tabularnewline
	(forecast period) &  &  &  &  &  & \tabularnewline
	Sq. &  &  &  &  &  & \tabularnewline
	$\alpha=0.01$ & -3.029 & 1.484 & -6.317 & -1.500 & 0.015 & 0.046\tabularnewline
	$\alpha=0.025$ & -2.363 & 1.140 & -4.938 & -1.167 & 0.028 & 0.089\tabularnewline
	$\alpha=0.05$ & -1.797 & 0.759 & -3.532 & -0.962 & 0.051 & 0.140\tabularnewline
	Norm &  &  &  &  &  & \tabularnewline
	$\alpha=0.01$ & -2.536 & 1.154 & -5.299 & -1.413 & 0.024 & 0.051\tabularnewline
	$\alpha=0.025$ & -2.133 & 0.978 & -4.482 & -1.182 & 0.037 & 0.091\tabularnewline
	$\alpha=0.05$ & -1.786 & 0.826 & -3.779 & -0.984 & 0.054 & 0.141\tabularnewline
	CAViaR-sy &  &  &  &  &  & \tabularnewline
	$\alpha=0.01$ & -2.567 & 1.577 & -13.878 & -0.878 & 0.013 & 0.034\tabularnewline
	$\alpha=0.025$ & -2.129 & 1.268 & -11.794 & -0.646 & 0.028 & 0.070\tabularnewline
	$\alpha=0.05$ & -1.749 & 1.267 & -11.673 & -0.453 & 0.047 & 0.119\tabularnewline
	CAViaR-asy &  &  &  &  &  & \tabularnewline
	$\alpha=0.01$ & -2.490 & 1.683 & -15.270 & -0.279 & 0.016 & 0.032\tabularnewline
	$\alpha=0.025$ & -2.140 & 1.427 & -11.882 & -0.205 & 0.027 & 0.067\tabularnewline
	$\alpha=0.05$ & -1.743 & 1.240 & -10.946 & -0.157 & 0.049 & 0.116\tabularnewline
	\hline 
\end{tabular}
	\label{table12}
\end{table}

\begin{sidewaystable}
	\caption{The table shows the value of the proposed test statistic for the $\alpha$ conditional quantile, the corresponding empirical p-value and the p-value of the Diebold and Marino test statistic with the tick loss function (DM) for evaluating estimated daily $VaR_{\alpha,t+1}$ of the S\&P500 index. The tests are conducted pairwisely. The methods shown in rows are benchmarks and those shown in column are competitors in the tests. The empirical p-value is obtained from using the re-centered bootstrap with bootstrap sample
	size $M=400$. The length of forecasting periods is 3,524 days (from Jan-02-2004 to Dec-29-2017). 
	}
	\centering
	\scalebox{0.9}{
\begin{tabular}{llcccccccccccccc}
	\hline 
	&  & \multicolumn{4}{c}{$\alpha=0.01$} &  & \multicolumn{4}{c}{$\alpha=0.025$} &  & \multicolumn{4}{c}{$\alpha=0.05$}\tabularnewline
	\cline{3-6} \cline{8-11} \cline{13-16} 
	&  & Sq. & Norm & CAViaR-sy & CAViaR-asy &  & Sq. & Norm & CAViaR-sy & CAViaR-asy &  & Sq. & Norm & CAViaR-sy & CAViaR-asy\tabularnewline
	\hline 
	& Test stat. & - & 0.060 & 0.212 & 0.263 &  & - & 0.147 & 0.366 & 0.414 &  & - & 0.068 & 0.566 & 0.713\tabularnewline
	Sq. & p-value & - & 0.940 & 0.028 & 0.012 &  & - & 0.322 & 0.005 & 0.005 &  & - & 0.932 & 0.005 & 0.000\tabularnewline
	& DM & - & 0.959 & 0.003 & 0.001 &  & - & 0.888 & 0.001 & 0.001 &  & - & 0.928 & 0.001 & 0.000\tabularnewline
	&  &  &  &  &  &  &  &  &  &  &  &  &  &  & \tabularnewline
	& Test stat. & 0.211 & - & 0.296 & 0.346 &  & 0.179 & - & 0.459 & 0.523 &  & 0.243 & - & 0.685 & 0.622\tabularnewline
	Norm & p-value & 0.048 & - & 0.025 & 0.030 &  & 0.292 & - & 0.003 & 0.005 &  & 0.040 & - & 0.000 & 0.002\tabularnewline
	& DM & 0.041 & - & 0.006 & 0.004 &  & 0.112 & - & 0.001 & 0.001 &  & 0.072 & - & 0.000 & 0.000\tabularnewline
	&  &  &  &  &  &  &  &  &  &  &  &  &  &  & \tabularnewline
	& Test stat. & 0.079 & 0.046 & - & 0.105 &  & 0.042 & 0.042 & - & 0.248 &  & 0.001 & 0.001 & - & 0.285\tabularnewline
	CAViaR-sy & p-value & 0.610 & 0.952 & - & 0.278 &  & 0.982 & 0.970 & - & 0.002 &  & 0.998 & 0.998 & - & 0.010\tabularnewline
	& DM & 0.997 & 0.994 & - & 0.090 &  & 0.999 & 0.999 & - & 0.030 &  & 0.999 & 1.000 & - & 0.031\tabularnewline
	&  &  &  &  &  &  &  &  &  &  &  &  &  &  & \tabularnewline
	& Test stat. & 0.069 & 0.069 & 0.091 & - &  & 0.037 & 0.037 & 0.053 & - &  & 0.032 & 0.032 & 0.057 & -\tabularnewline
	CAViaR-asy & p-value & 0.778 & 0.895 & 0.390 & - &  & 0.992 & 0.982 & 0.908 & - &  & 0.998 & 0.995 & 0.940 & -\tabularnewline
	& DM & 0.999 & 0.996 & 0.910 & - &  & 0.999 & 0.999 & 0.970 & - &  & 1.000 & 1.000 & 0.969 & -\tabularnewline
	\hline 
\end{tabular}

}
	\label{table13}
\end{sidewaystable}

\clearpage
\begin{figure}[ht]
	\begin{center}
		\mbox{
			\subfigure{\includegraphics[height=7cm,width=8cm]{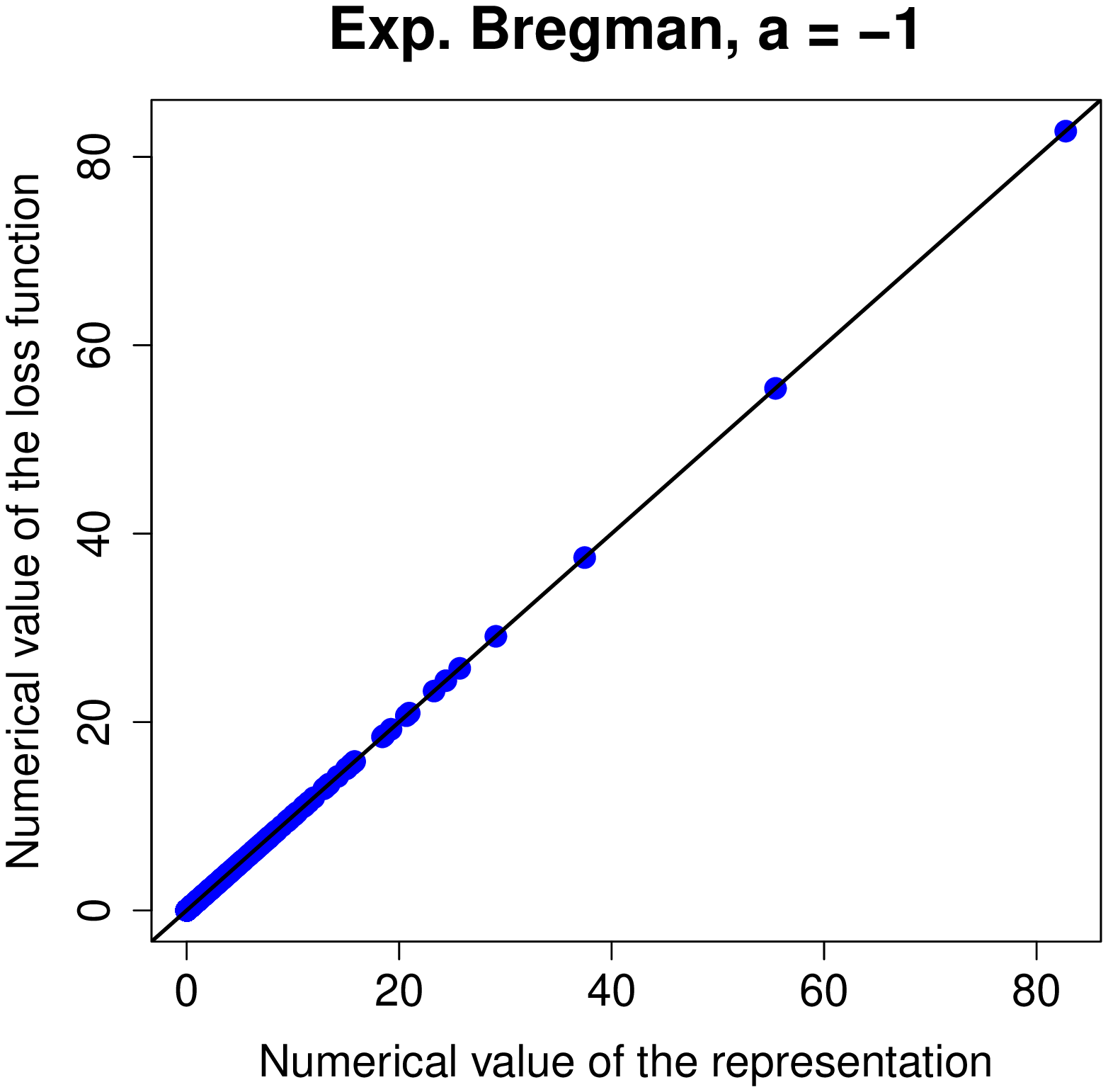}}
			\subfigure{\includegraphics[height=7cm,width=8cm]{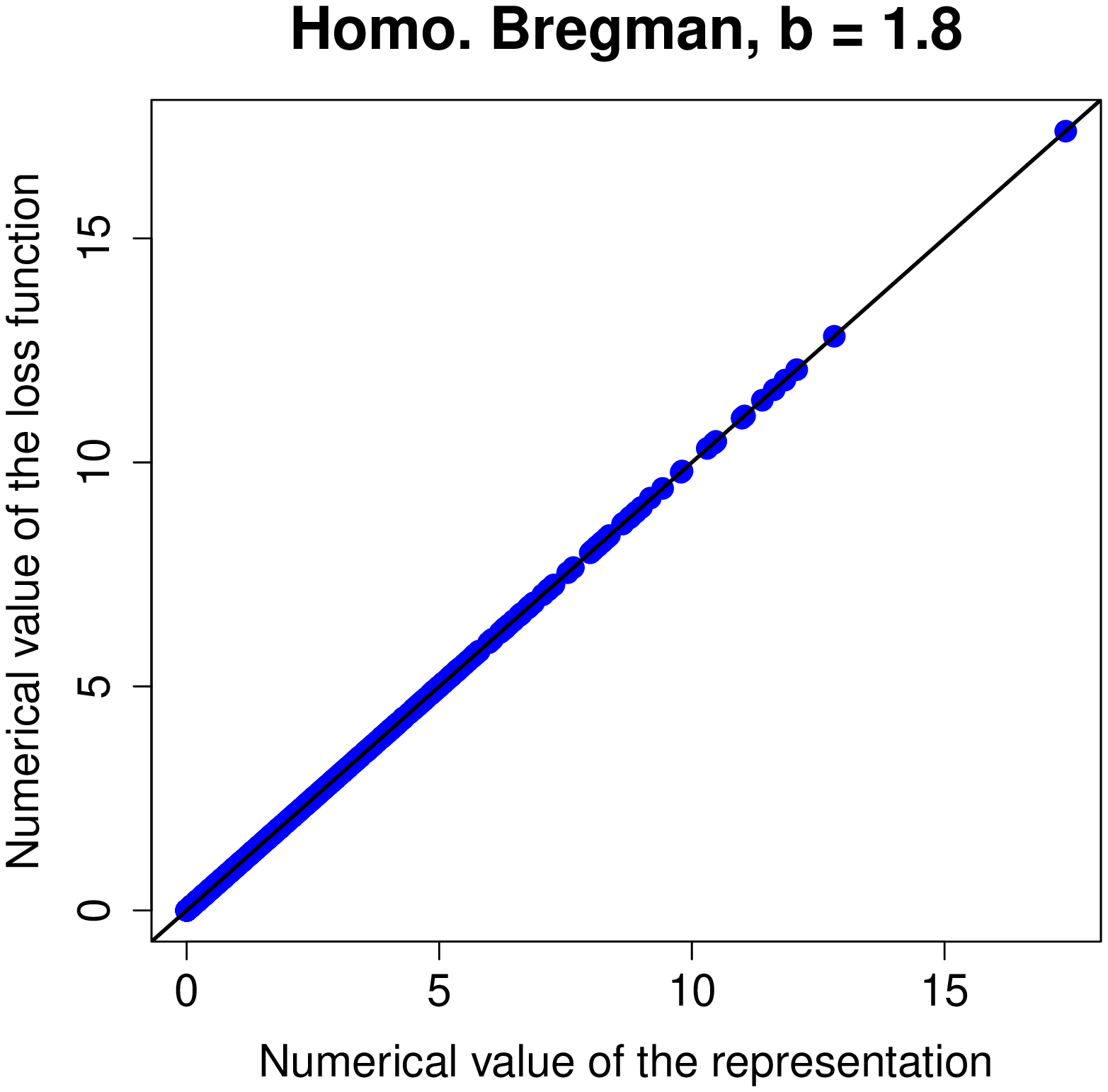}}
		} 
		\mbox{
		\subfigure{\includegraphics[height=7cm,width=8cm]{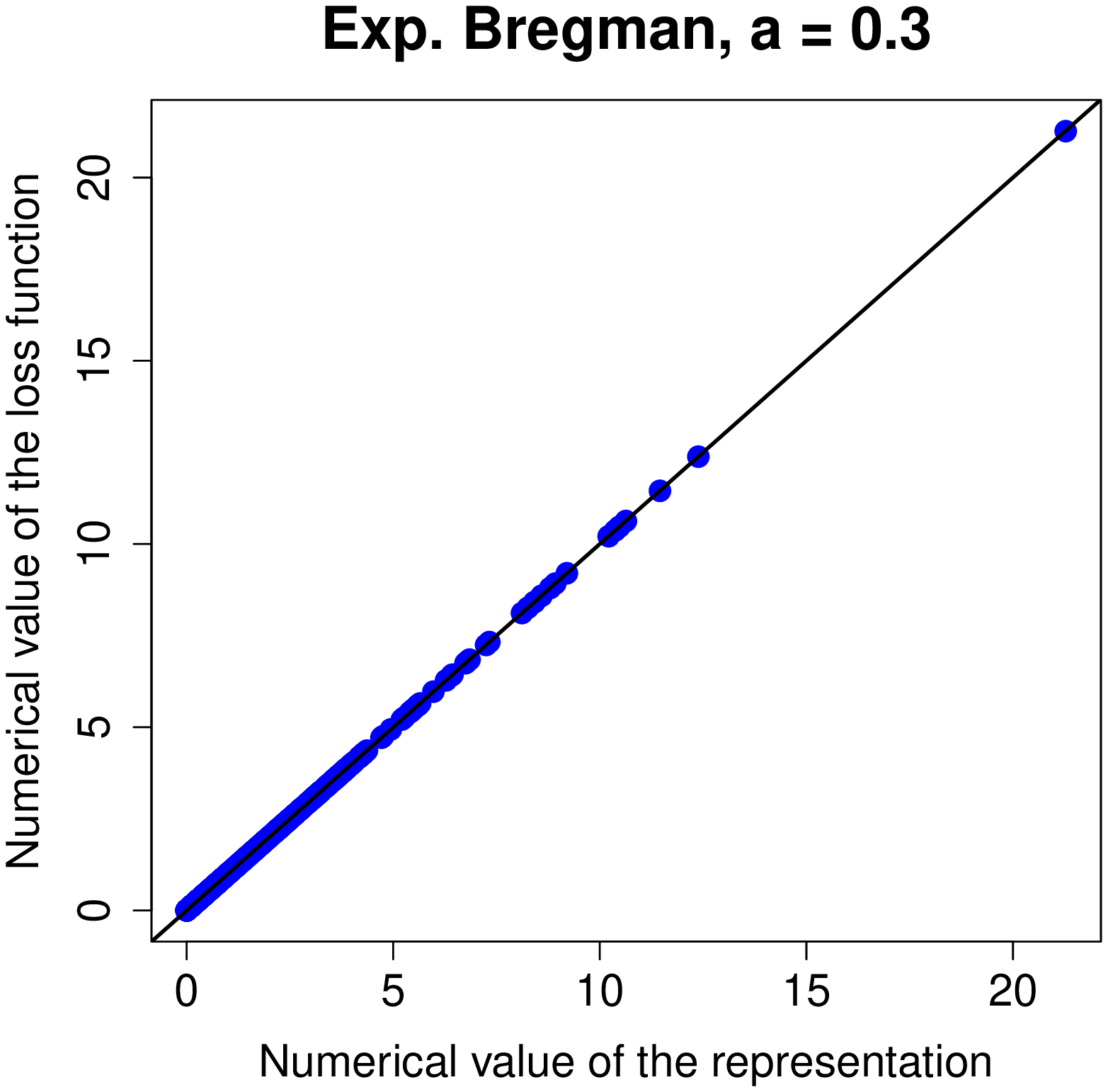}}
		\subfigure{\includegraphics[height=7cm,width=8cm]{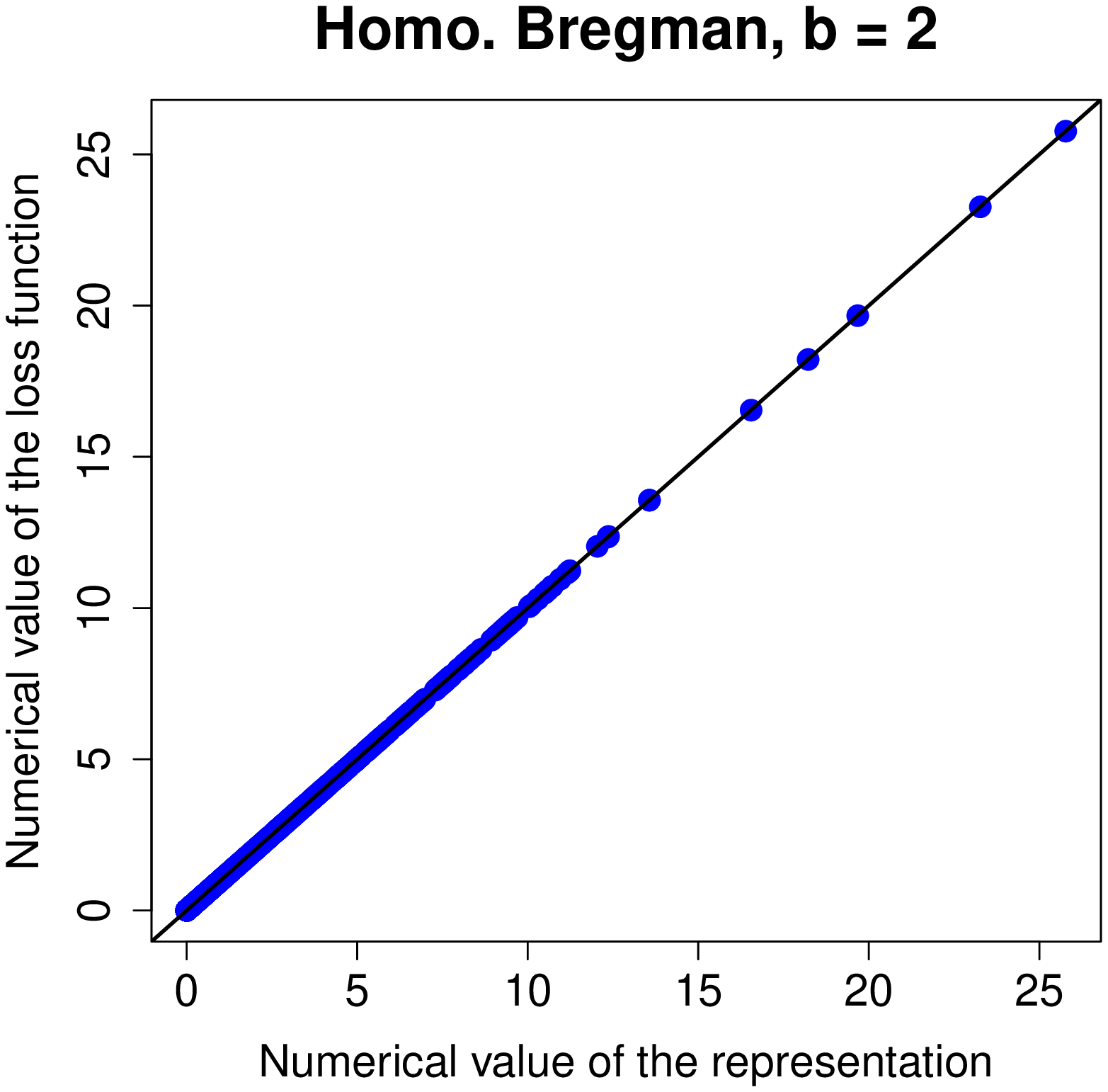}}
		}
		\mbox{
		\subfigure{\includegraphics[height=7cm,width=8cm]{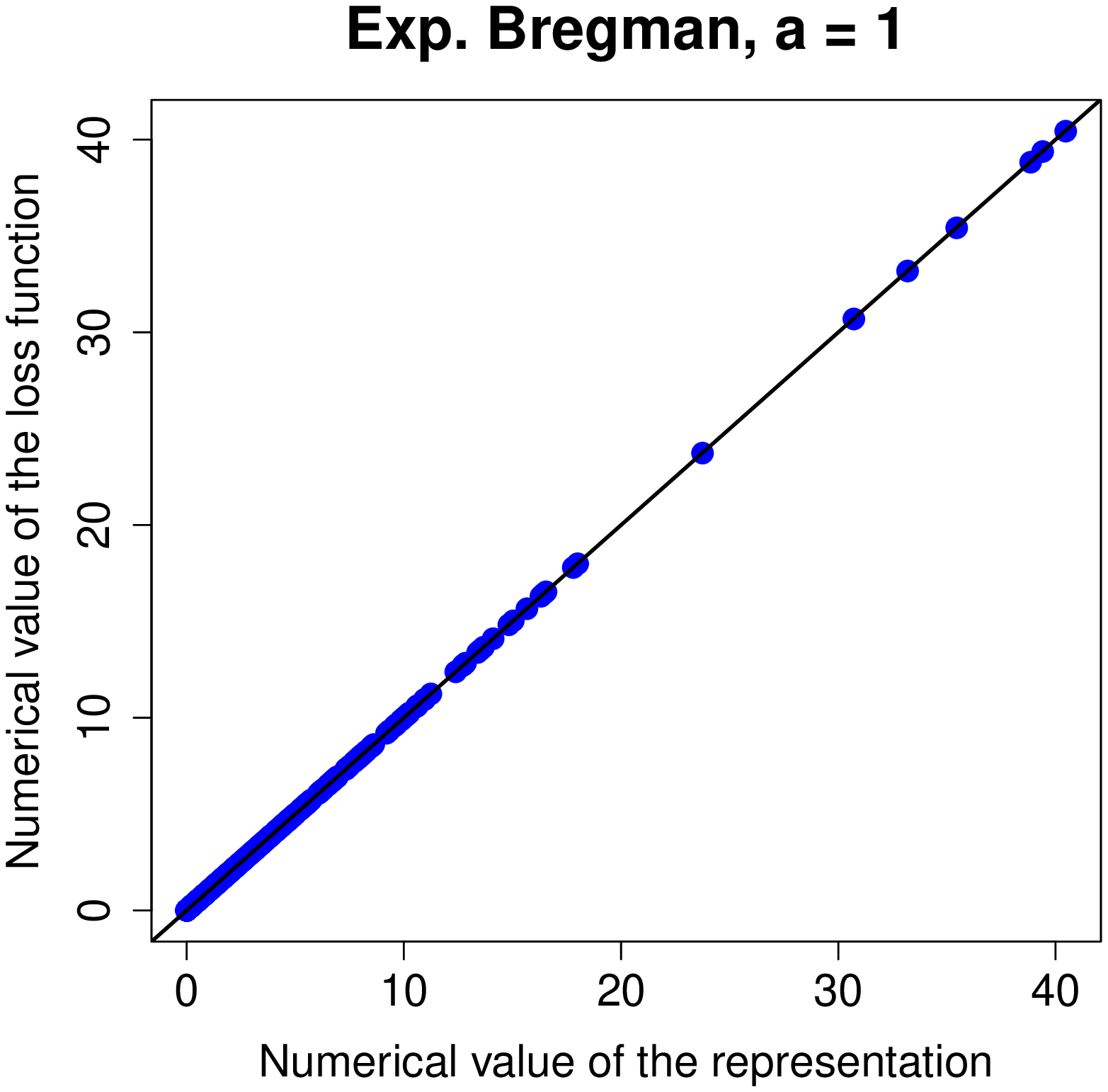}}
		\subfigure{\includegraphics[height=7cm,width=8cm]{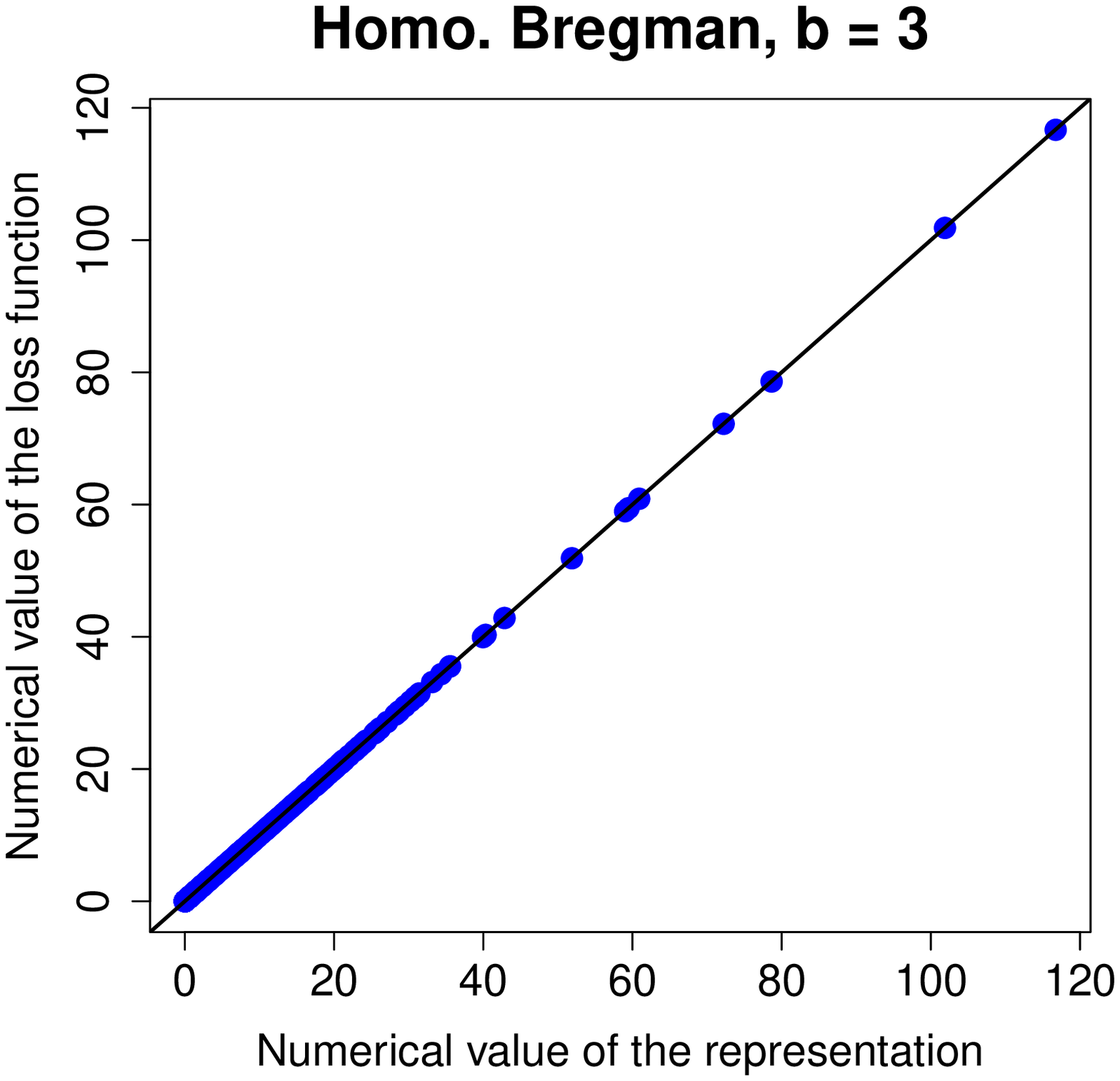}}
		}
	\end{center}
	\caption{The figure shows comparisons of numerical values
		of a consistent loss function for the $\alpha-$expectile and those obtained from using
		representation of (\ref{scoring_expectile1}) when $\alpha=0.5$. Left panel shows plots of numerical values of the exponential Bregman loss function vs. those obtained from using representation of (\ref{scoring_expectile1}) when $a=-1$, 0.3 and 1. Right panel shows the case of the homogeneous Bregman loss function with $b=1.8$, 2 and 3. The data for each comparison are 1000 pairs of $X\sim N(0,1)$ and
		$Y\sim N(0,1)$.}
	\label{figure1}
\end{figure}

\begin{figure}[ht]
	\begin{center}
		\mbox{
			\subfigure{\includegraphics[height=7cm,width=8cm]{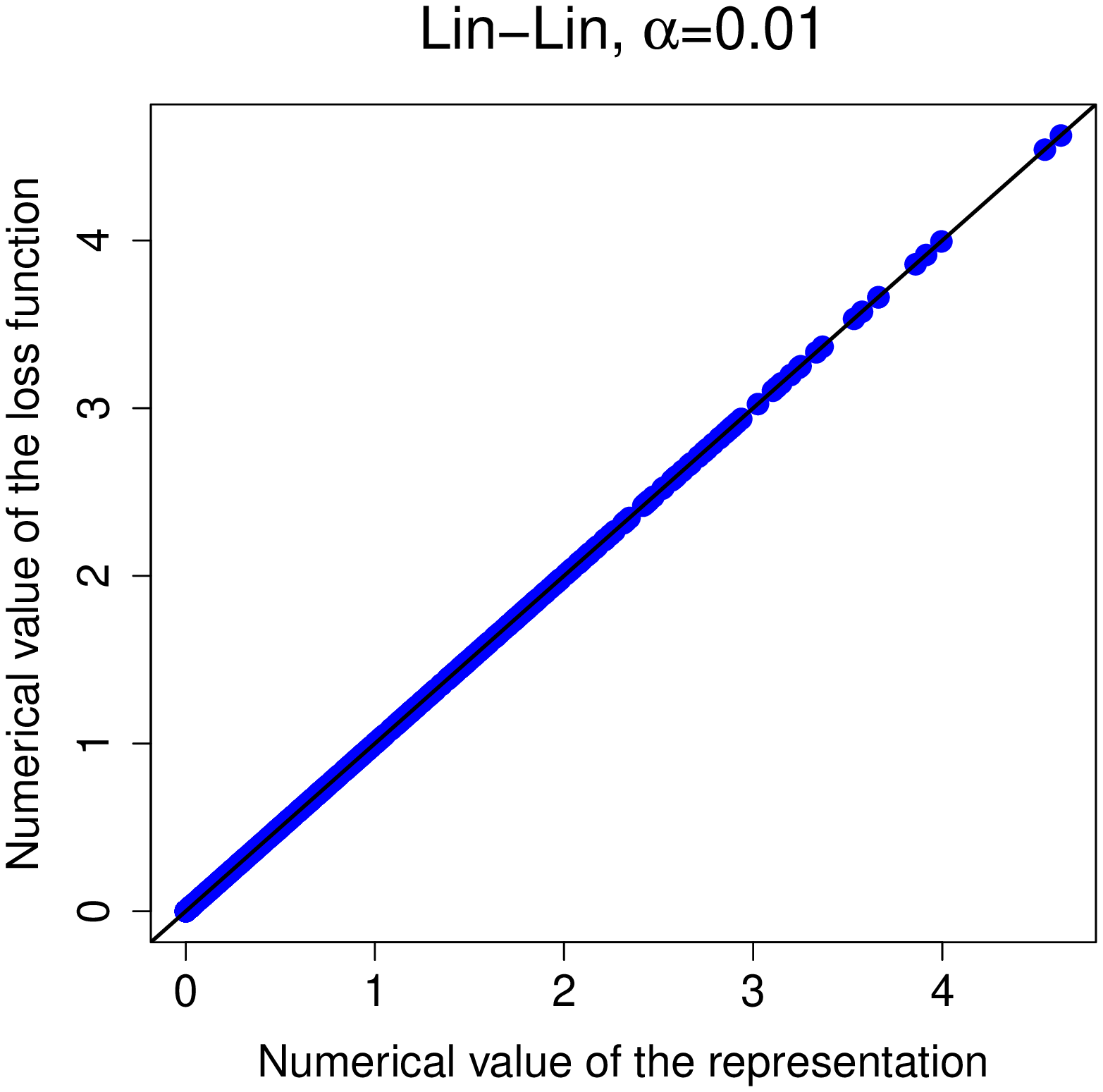}}
			\subfigure{\includegraphics[height=7cm,width=8cm]{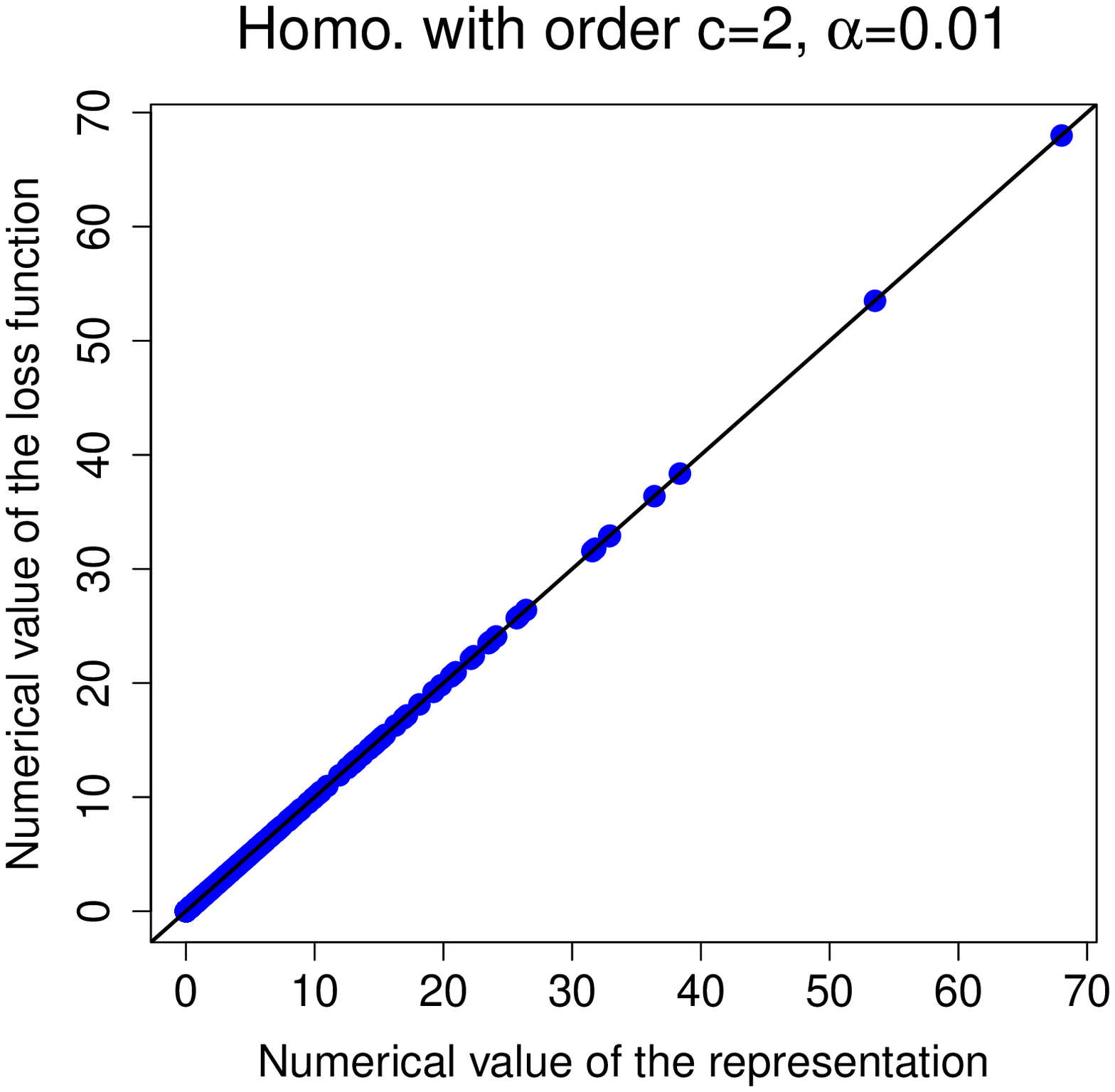}}
		} 
		\mbox{
			\subfigure{\includegraphics[height=7cm,width=8cm]{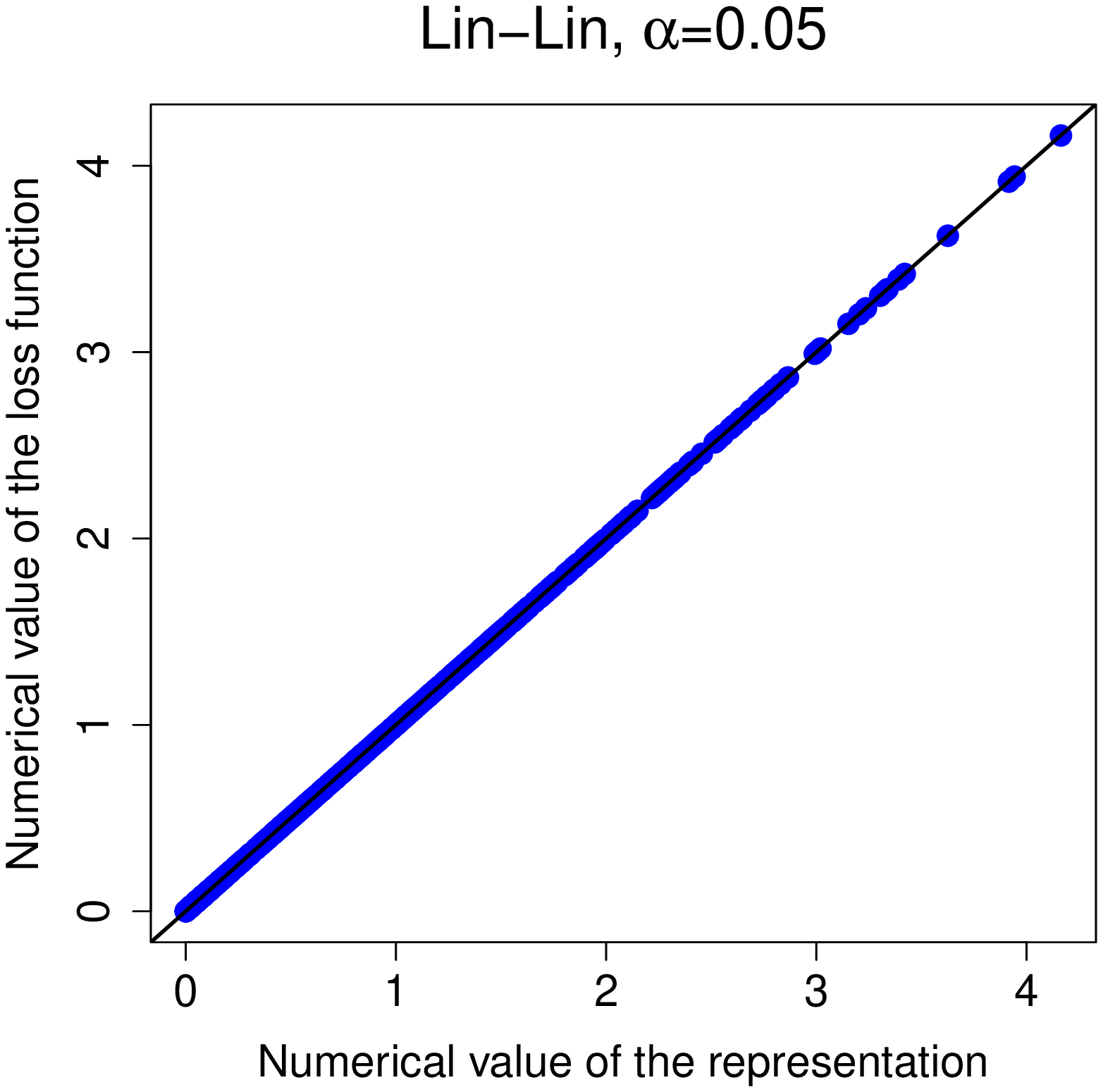}}
			\subfigure{\includegraphics[height=7cm,width=8cm]{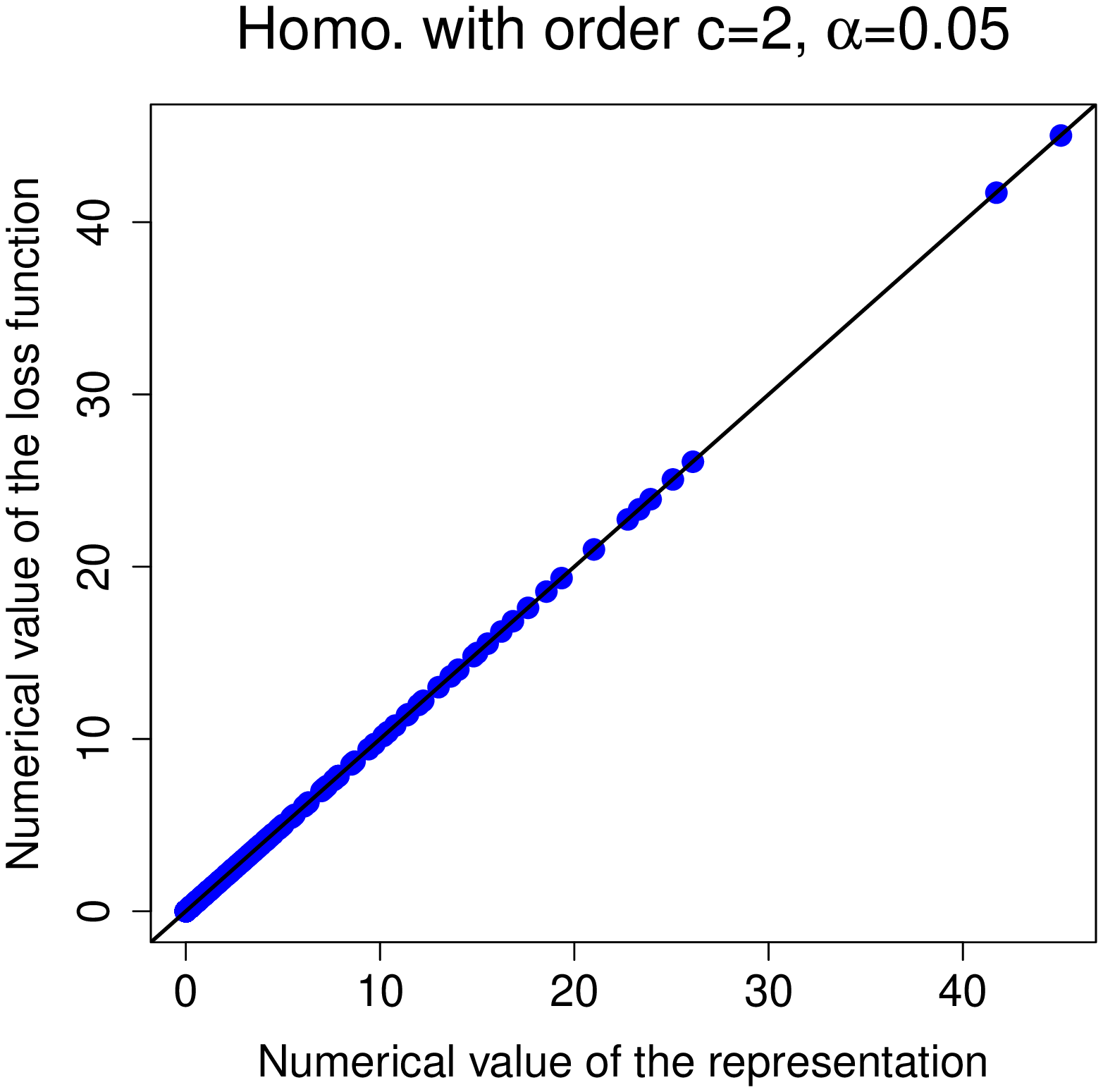}}
		}
		\mbox{
			\subfigure{\includegraphics[height=7cm,width=8cm]{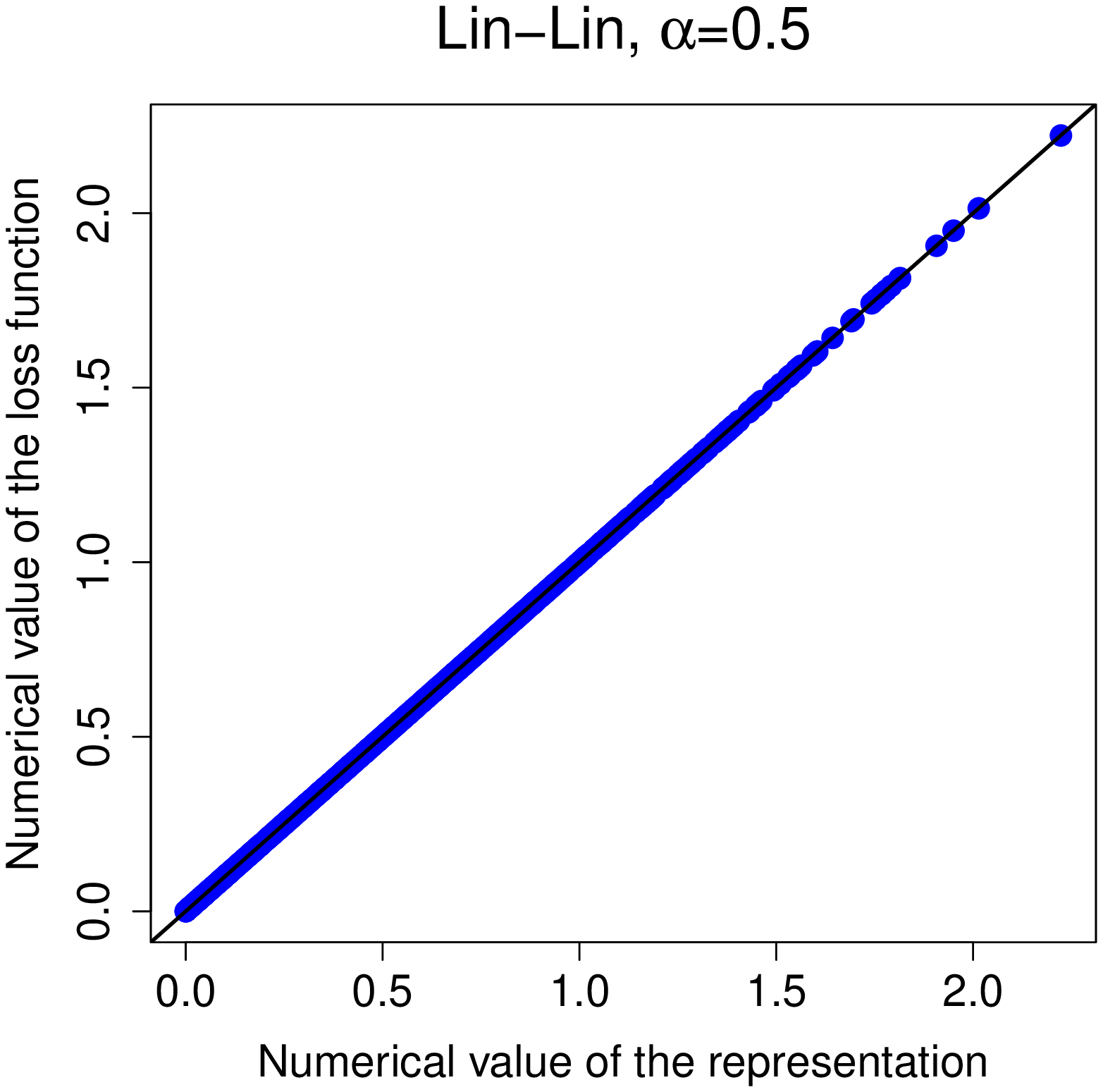}}
			\subfigure{\includegraphics[height=7cm,width=8cm]{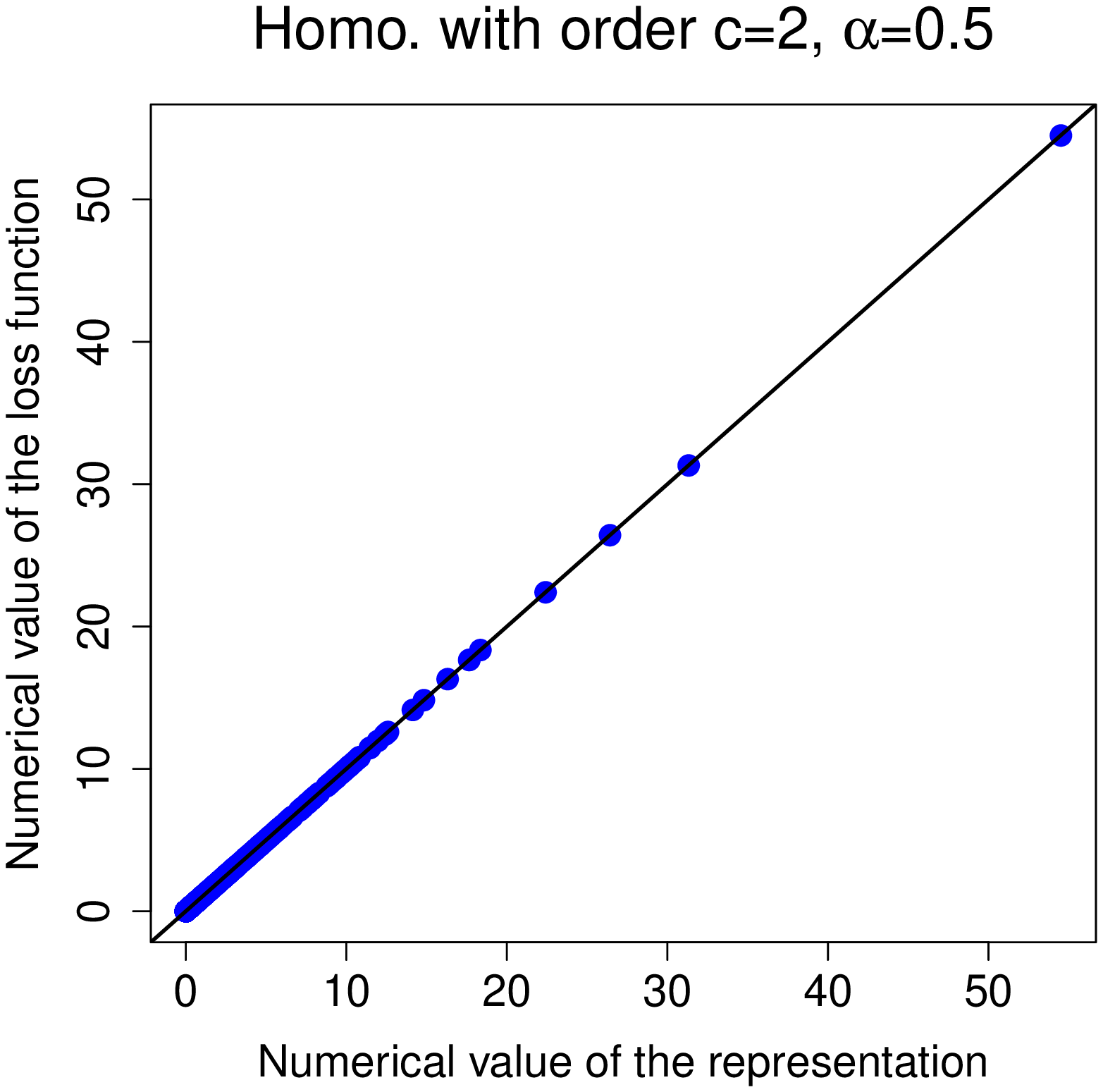}}
		}
	\end{center}
	\caption{The figure shows comparisons of numerical values
		of a consistent loss function for the $\alpha-$quantile and those obtained from using
		representation of (\ref{scoring_quantile1}) when $\alpha=0.01$, 0.05 and 0.5. Left panel shows plots of numerical values of the lin-lin loss function vs. those obtained from using representation of  (\ref{scoring_quantile1}). Right panel shows the case of the homogeneous loss function with order $c=2$. In the case of the lin-lin loss function, the data for each comparison are 1000 pairs of $X\sim N(0,1)$ and
		$Y\sim N(0,1)$. In the case of the homogeneous loss function with order $c=2$, the data for each comparison are 1000 pairs of $X\sim \chi^{2}\left(1\right)$ and $Y\sim \chi^{2}\left(1\right)$.}
	\label{figure2}
\end{figure}

\begin{figure}[ht]
	\begin{center}
		\mbox{
			\subfigure{\includegraphics[height=7cm,width=8cm]{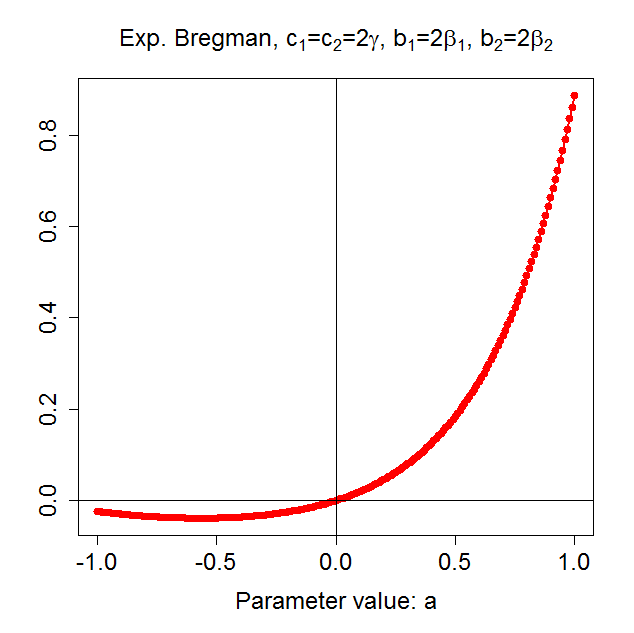}}
			\subfigure{\includegraphics[height=7cm,width=8cm]{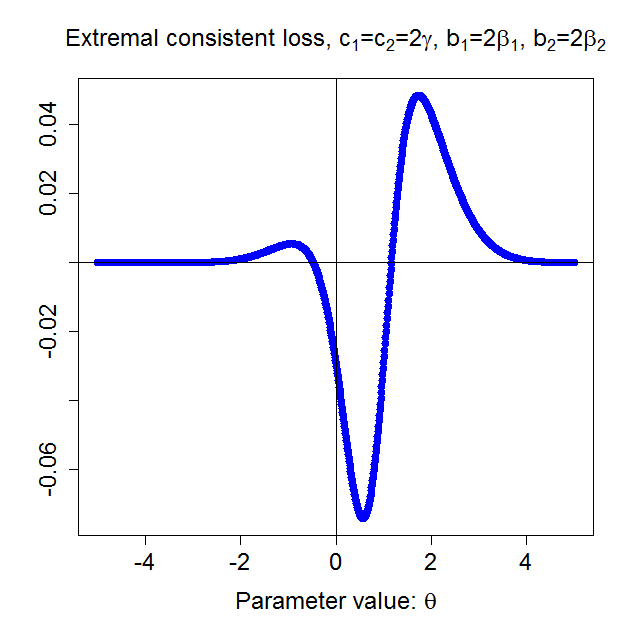}}
		} 
		\mbox{
			\subfigure{\includegraphics[height=7cm,width=8cm]{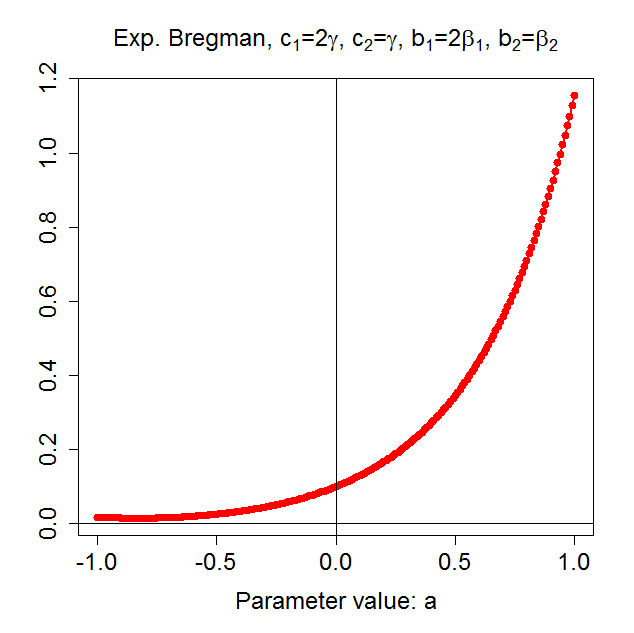}}
			\subfigure{\includegraphics[height=7cm,width=8cm]{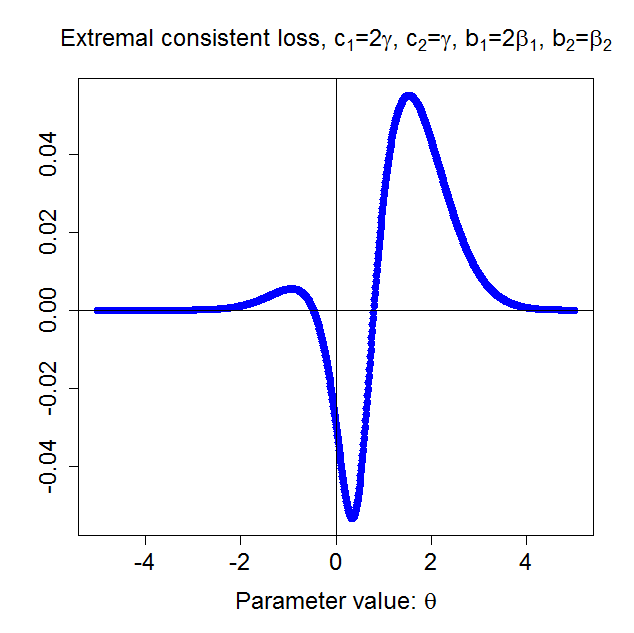}}
		}
		\mbox{
			\subfigure{\includegraphics[height=7cm,width=8cm]{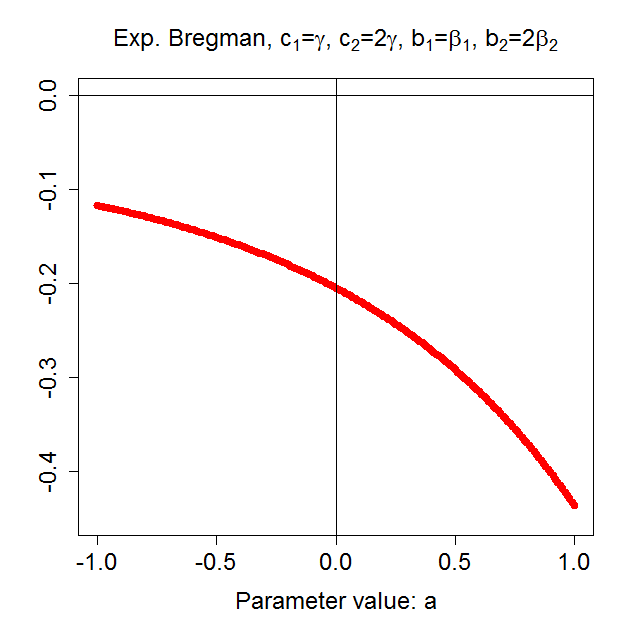}}
			\subfigure{\includegraphics[height=7cm,width=8cm]{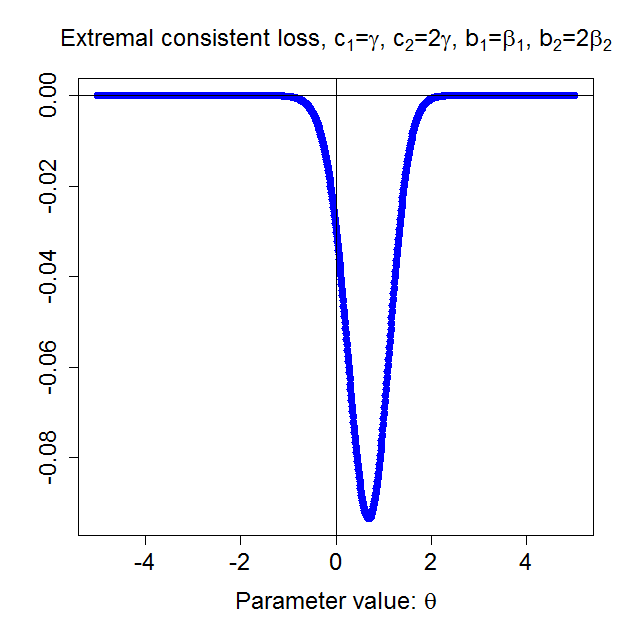}}
		}
	\end{center}
	\caption{The figure shows differences of the expected
		exponential Bregman loss with parameter $a\in\left[-1,1\right]$ (left panel) and differences of the expected extremal loss for the conditional expectation with parameter $\theta\in\left[-5,5\right]$ (right panel) for the two forecasts in cases (1) to (3) in Section 4.1.1.}
	\label{figure3}
\end{figure}

\begin{figure}[ht]
	\begin{center}
		\mbox{
			\subfigure{\includegraphics[height=5cm,width=7cm]{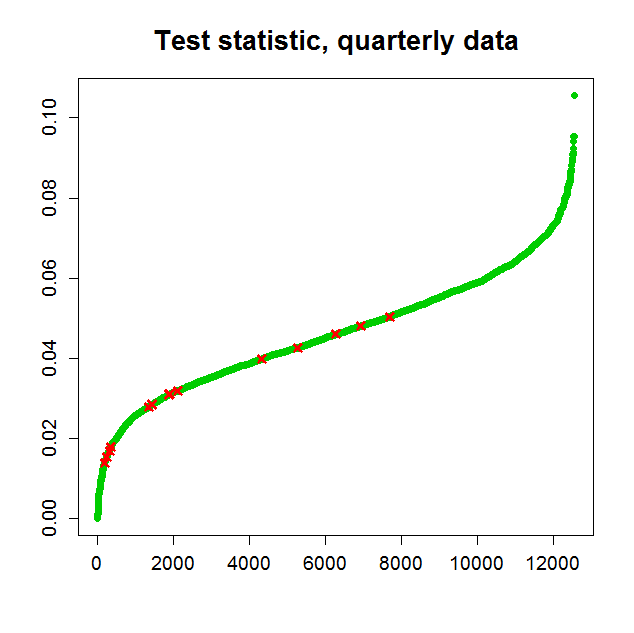}}
			\subfigure{\includegraphics[height=5cm,width=7cm]{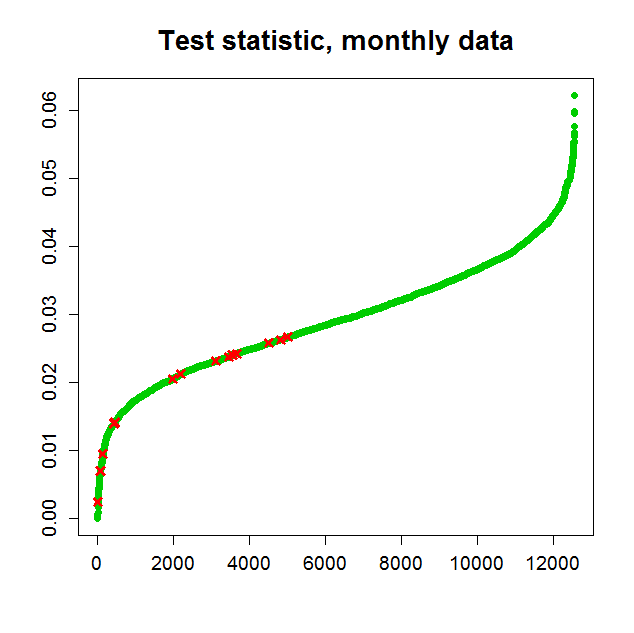}}
		} 
		\mbox{
			\subfigure{\includegraphics[height=5cm,width=7cm]{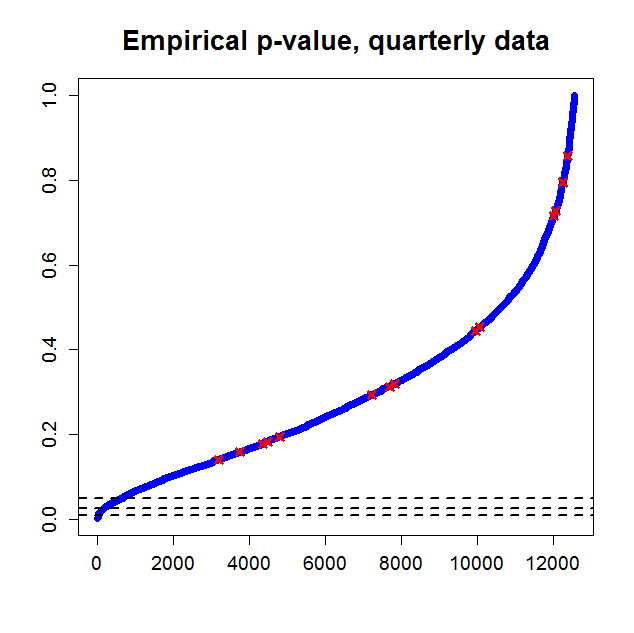}}
			\subfigure{\includegraphics[height=5cm,width=7cm]{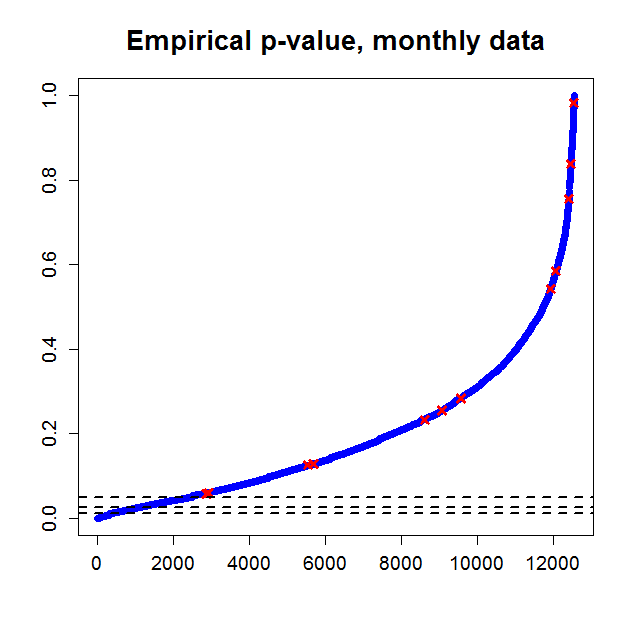}}
		}
		\mbox{
			\subfigure{\includegraphics[height=5cm,width=7cm]{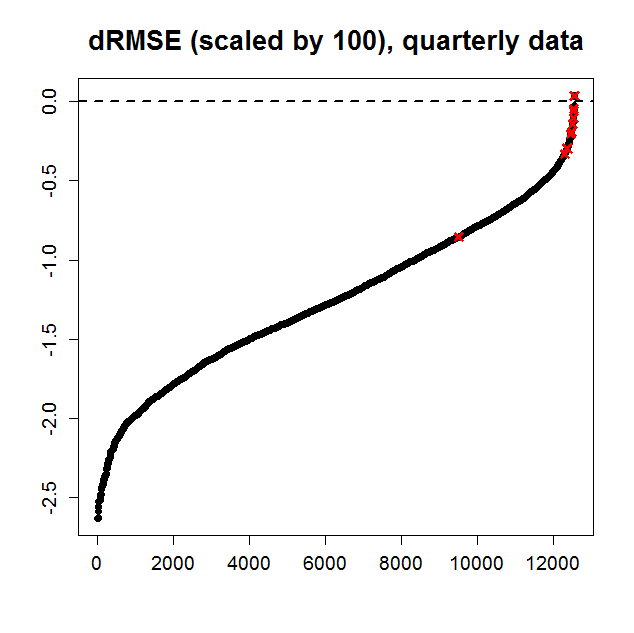}}
			\subfigure{\includegraphics[height=5cm,width=7cm]{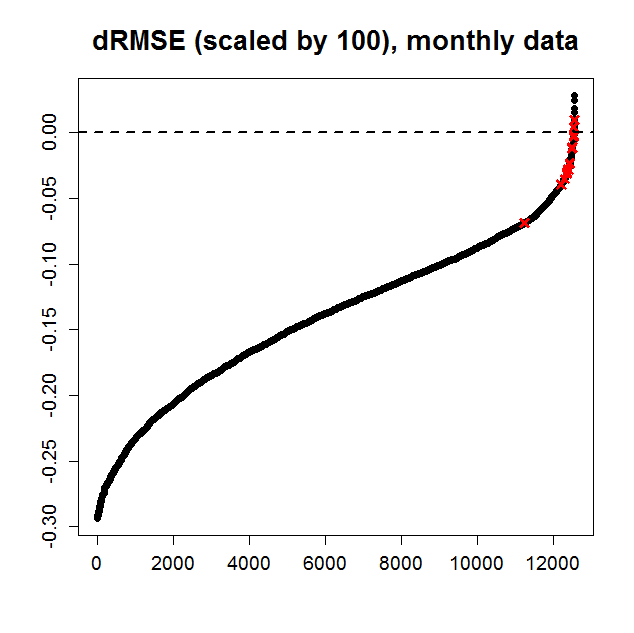}}
		}
		\mbox{
			\subfigure{\includegraphics[height=5cm,width=7cm]{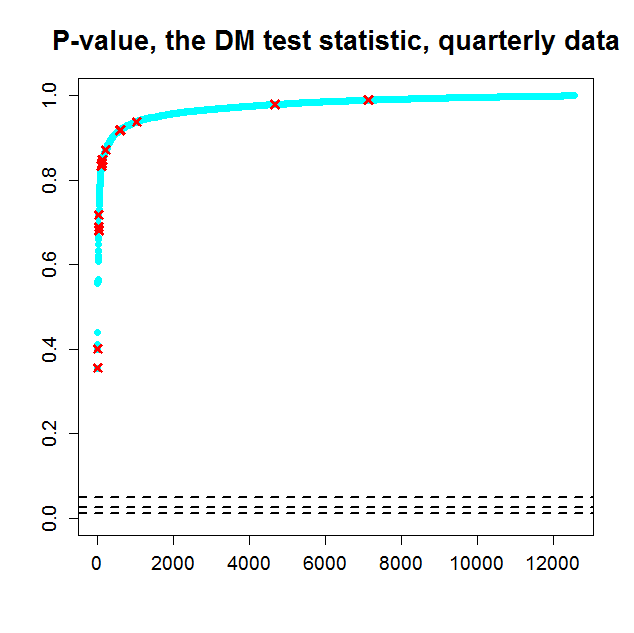}}
			\subfigure{\includegraphics[height=5cm,width=7cm]{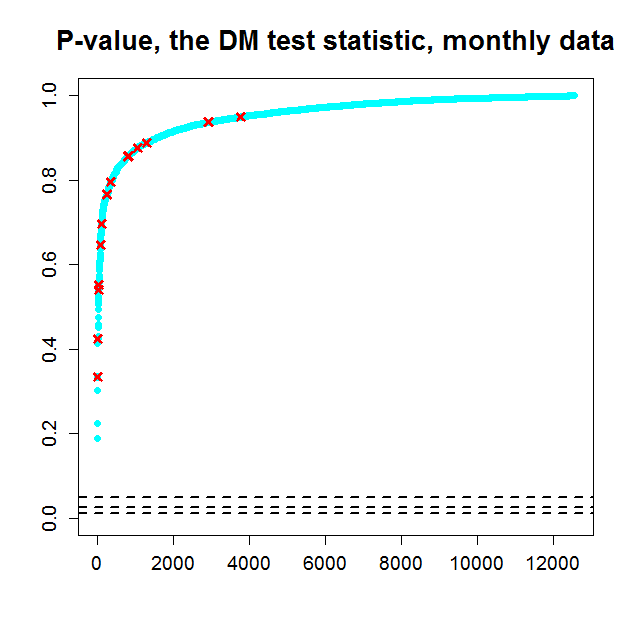}}
		}
	\end{center}
	\caption{The figure shows ordered values (from small to large) of the proposed test statistic for forecasting the conditional expectation, the corresponding empirical p-values, dRMSE scaled by 100 and the p-values of the DM test statistic with the squared error loss function for the multivariate predictive regressions. Left panel shows the cases of quarterly data and right panel shows the cases of monthly data. The red crosses in each plot are values of these quantities for the single-variable predictive regressions shown in Table \ref{table9}. }
	\label{figure9}
\end{figure}

\begin{sidewaysfigure}[ht]
	\centering
		\includegraphics[height= 18cm, width = 25cm]{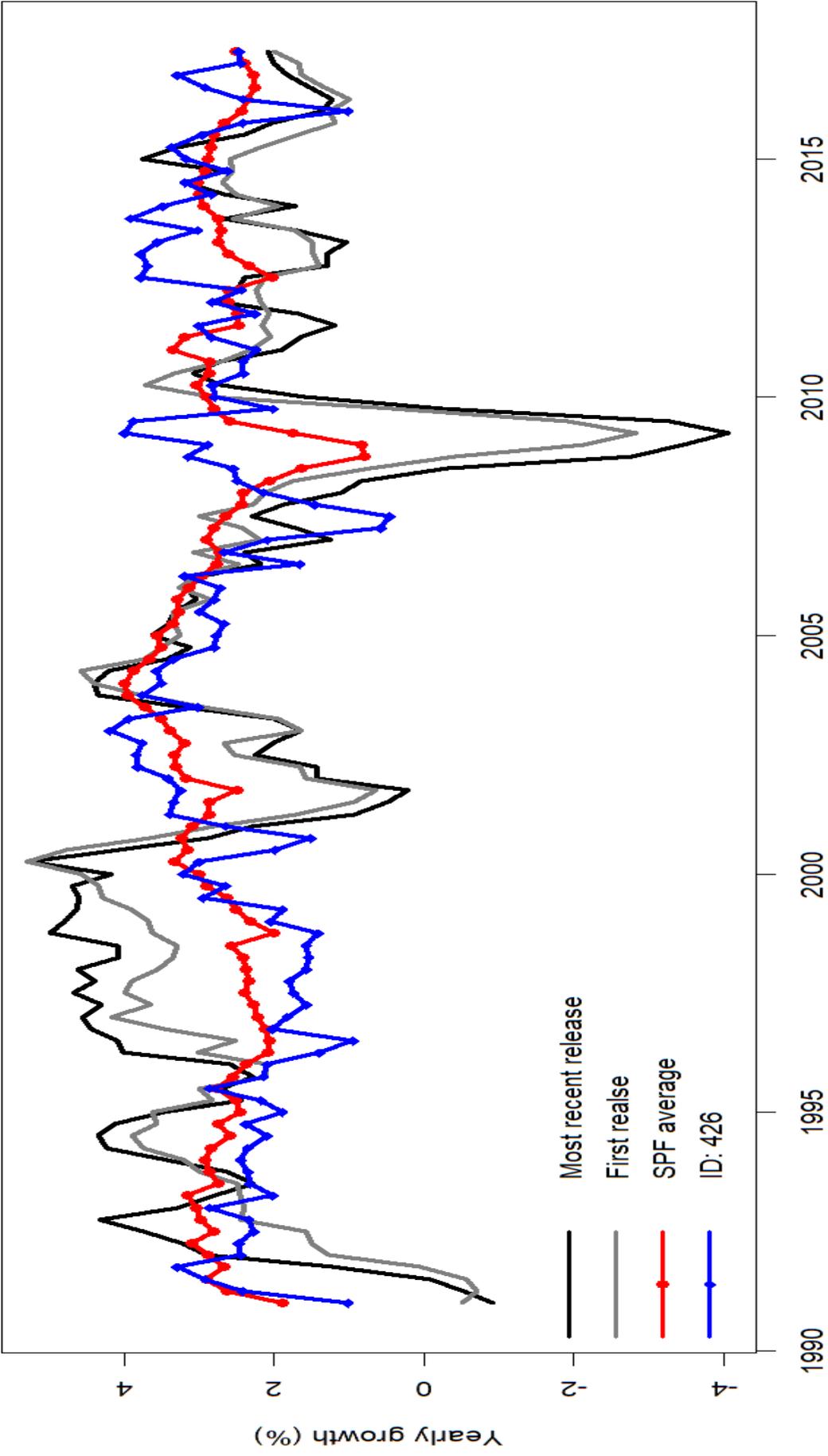}
	\caption{The figure shows time series plots of Q3-2017 vintage and first release for U.S. real gross domestic product (RGDP) annual growth and two corresponding forecasts from Survey of Professional Forecasters conducted by Fed. Philadelphia: mean forecast from all experts (SPF average) and a forecast from an expert with ID. 426 (ID: 426). The data is in quarterly frequency and sample period is from Q1-1991 to Q2-2017 (106 quarters).}
	\label{figure10}
\end{sidewaysfigure}

\begin{figure}[ht]
	\begin{center}
		\mbox{
			\subfigure{\includegraphics[height=7cm,width=8cm]{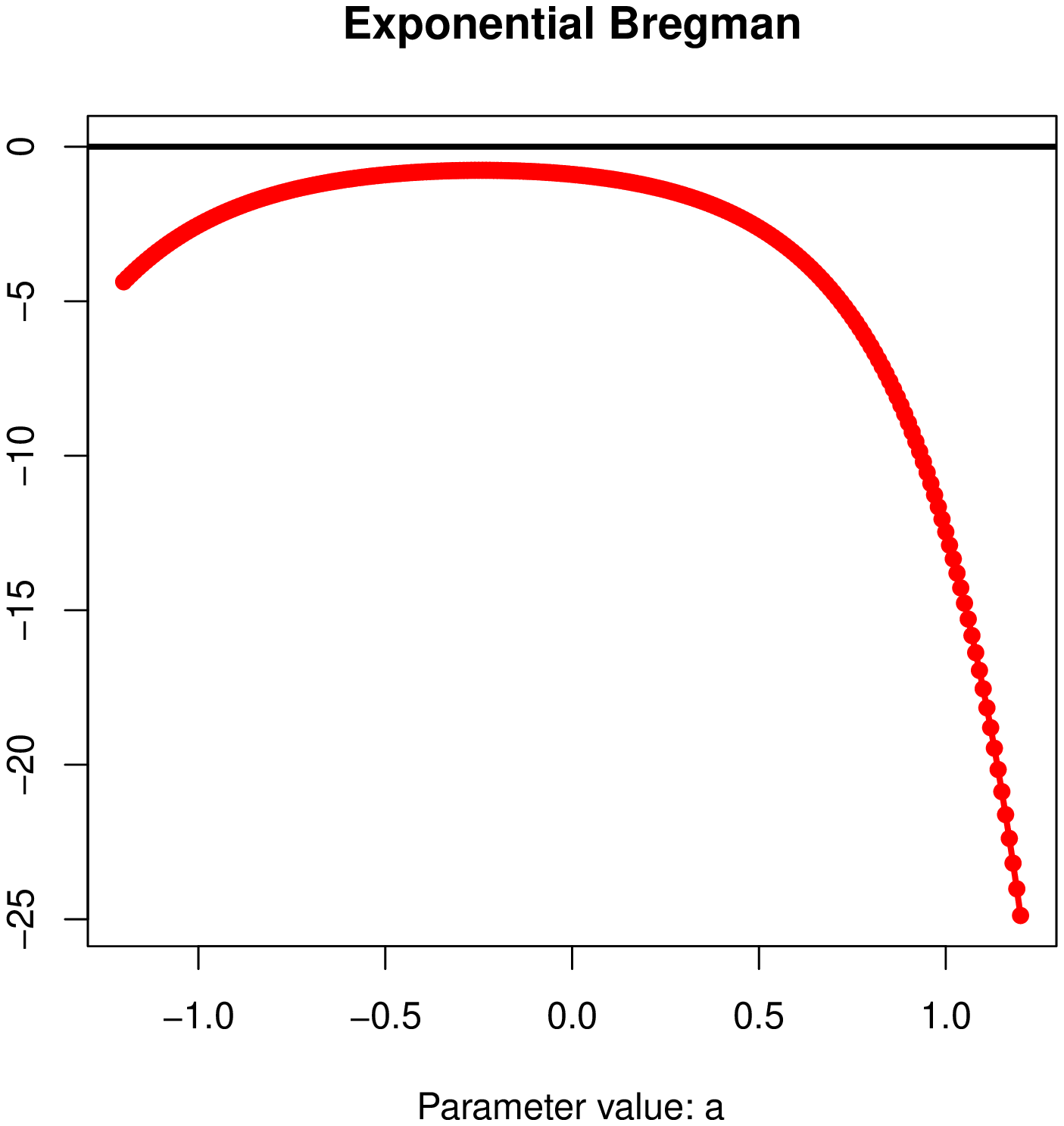}}
			\subfigure{\includegraphics[height=7cm,width=8cm]{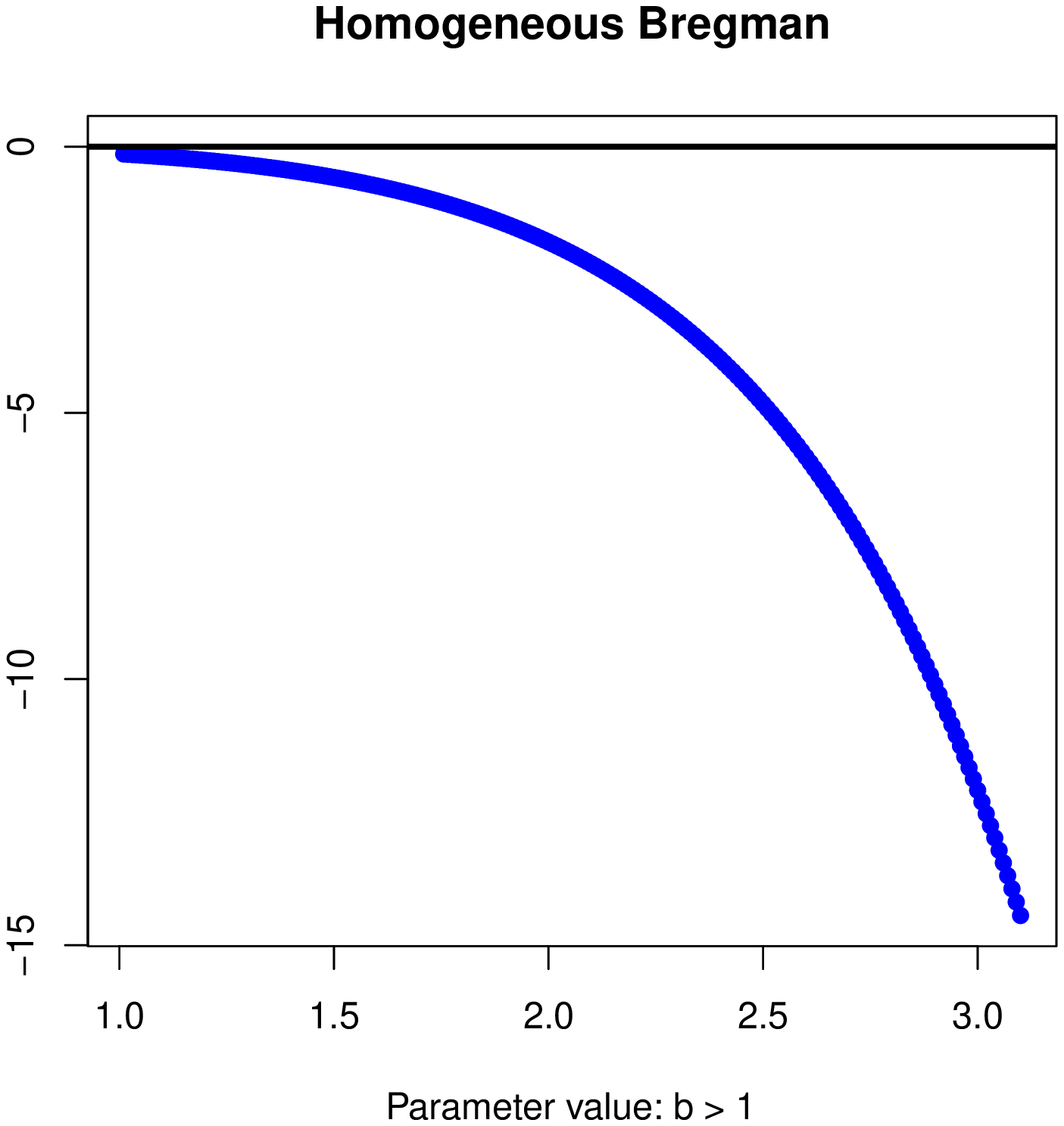}}
		} 
		\mbox{
			\subfigure{\includegraphics[height=7cm,width=8cm]{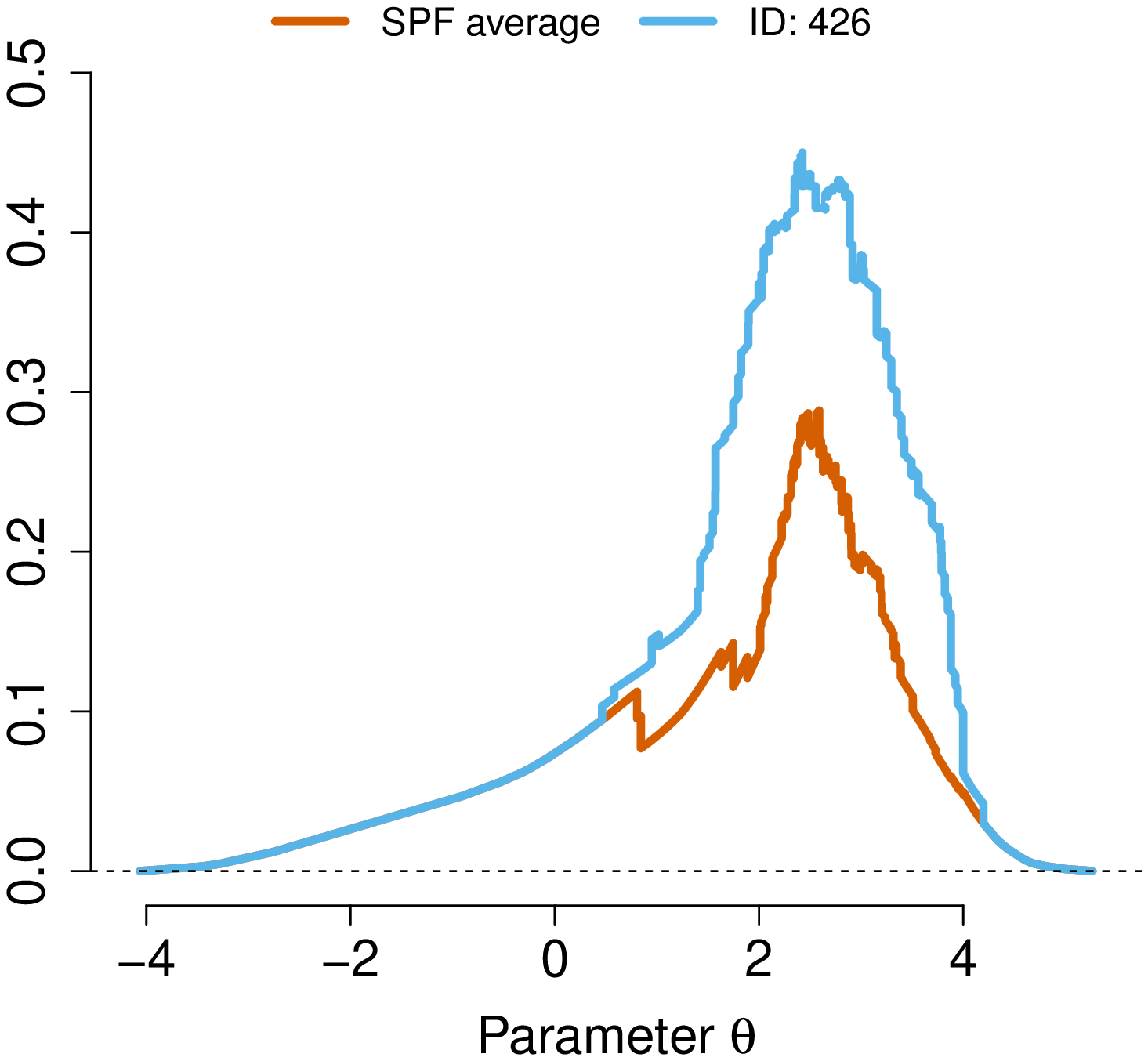}}
			\subfigure{\includegraphics[height=7cm,width=8cm]{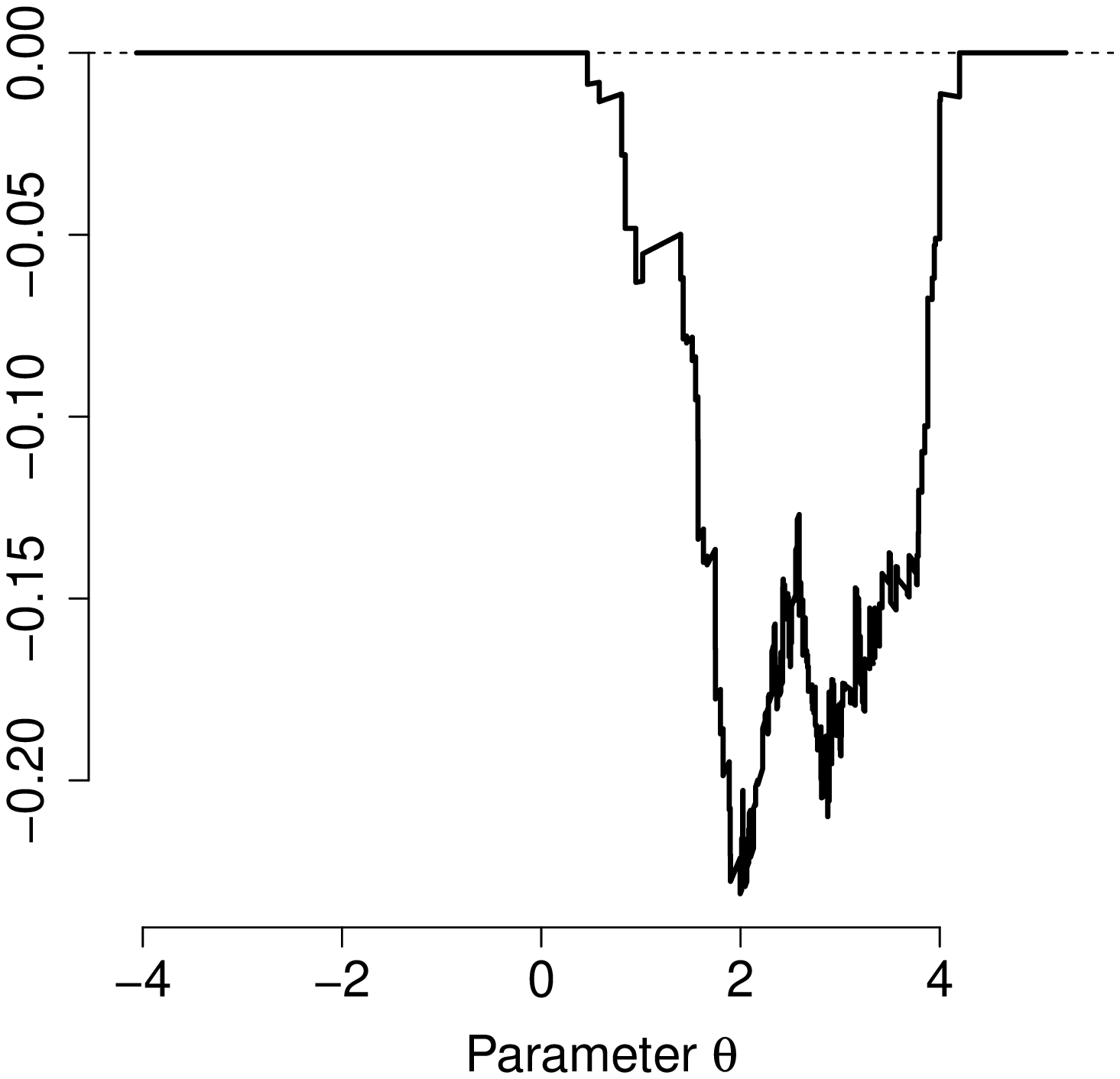}}
		}
	\end{center}
	\caption{The figure shows empirical value of the extremal consistent loss function for the expectile evaluated with two forecasts: SPF average and ID: 426 (bottom left) and empirical differences of the consistent loss functions (SPF average minus ID: 426): exponential Bregman loss (top left), homogeneous Bregman loss (top right) and the extremal consistent loss function for the expectile forecast with $\alpha=0.5$ (bottom right). The realized value of the target random variable is the Q3-2017 vintage for annual growth of U.S. RGDP. The data is in quarterly frequency and sample period is from Q1-1991 to Q2-2017 (106 quarters). The two plots in the bottom 
		are generated with \texttt{R} package \texttt{murphydiagram} \citep{EGJK_2016}.}
	\label{figure11}
\end{figure}

\begin{figure}[ht]
	\begin{center}
		\includegraphics[height=10cm,width=16cm]{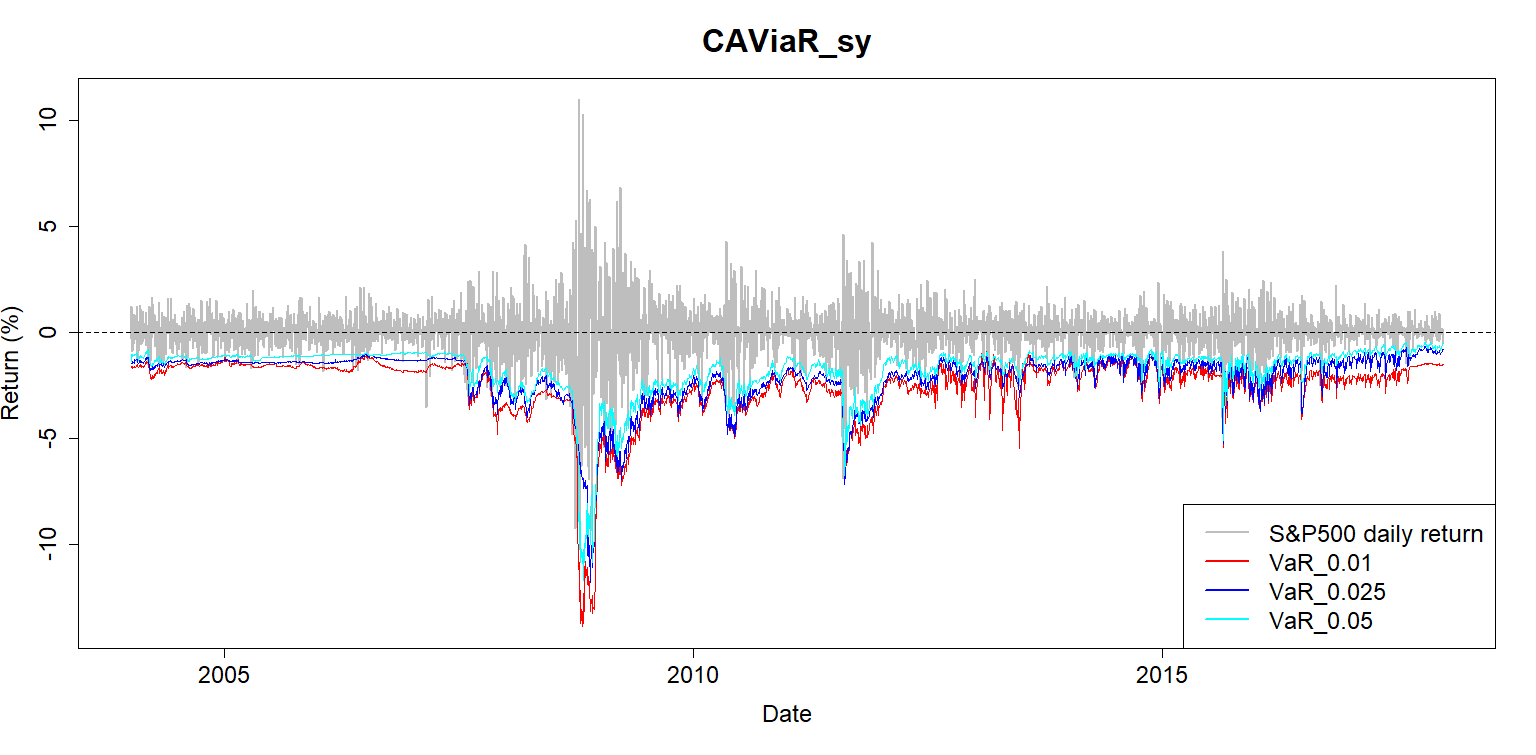}
		\includegraphics[height=10cm,width=16cm]{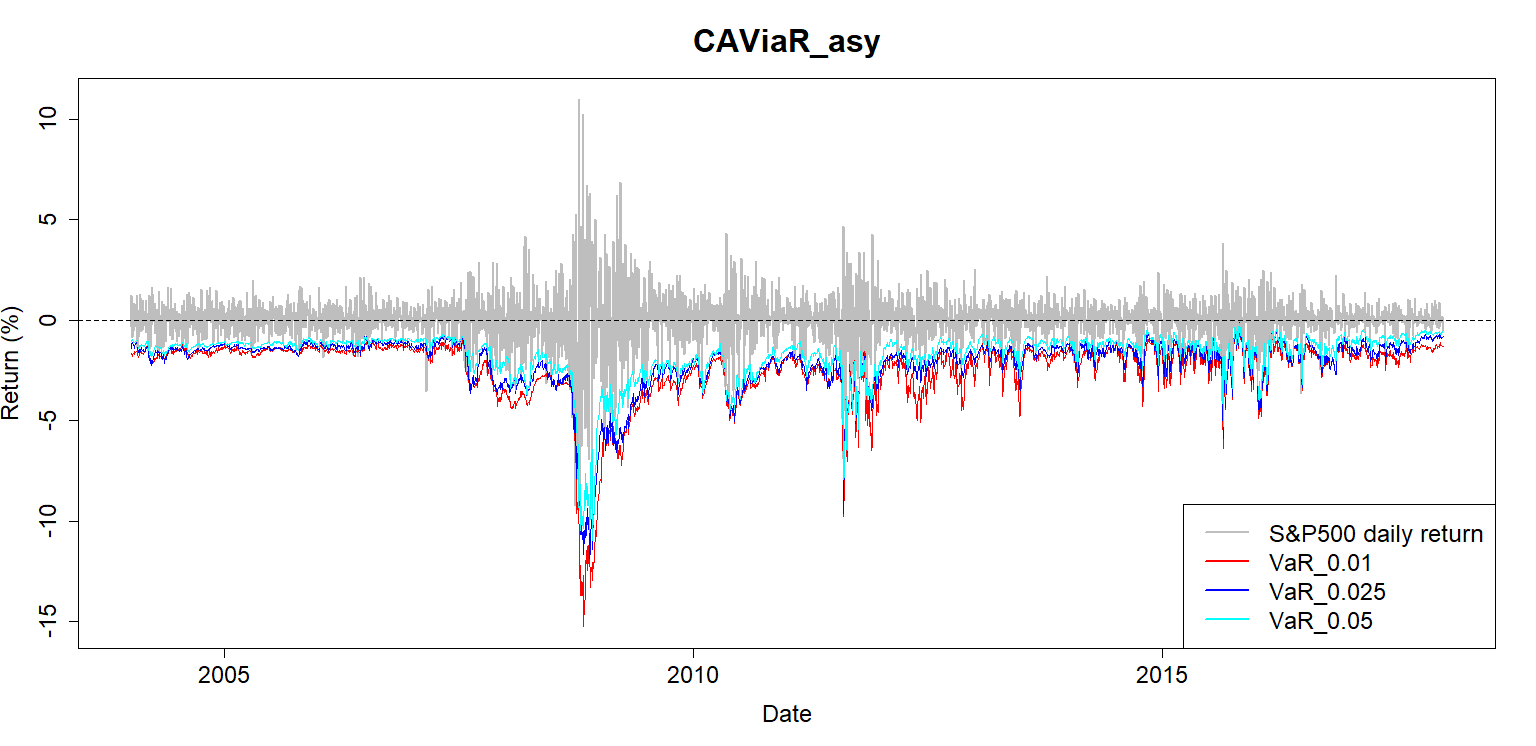}
	\end{center}
	\caption{The figure shows time-series plots of the daily S\&P500 log return and the estimated $VaR_{\alpha,t+1}$ generated with CAViaR-sy and CAViaR-asy. The forecast period is from Jan-02-2004 to Dec-29-2017 (3,524 days).}
	\label{figure12}
\end{figure}

\clearpage

\section{Appendix (For online publication only)}

\subsection{Some lemmas and proofs}
Here we restate some relevant definition and assumptions used in Subsection 3.2. Let $x\vee y=\max(x,y)$ and $x\wedge y=\min(x,y)$ and $\Rightarrow$ denote weak convergence of stochastic processes. 
\begin{definition}[the strong mixing coefficients $\alpha\left(n\right)$]
	Let $\underline{\mathcal{F}}_{k,-\infty}^{\overline{T}}$ denote the
	$\sigma-$field generated by $\left\{ Z_{kt},-\infty<t\leq\overline{T}\right\} $
	and $\overline{\mathcal{F}}_{k,\overline{T}}^{\infty}$ denote the
	$\sigma-$field generated by $\left\{ Z_{kt},\overline{T}\leq t<\infty\right\} $.
	The strong mixing coefficients $\alpha\left(n\right)$ are defined as 
	\[
	\sup_{\overline{T},k}\sup_{A\in\underline{\mathcal{F}}_{k,-\infty}^{\overline{T}},B\in\overline{\mathcal{F}}_{k,\overline{T}+n}^{\infty}}\left|P\left(A\cap B\right)-P\left(A\right)P\left(B\right)\right|=\alpha\left(n\right).
	\]
	The array $Z_{kt}$ satisfies the strong mixing condition if $\alpha\left(n\right)\downarrow0$
	as $n\rightarrow\infty$.
\end{definition}
Define empirical processes 
\begin{eqnarray*}
	v_{k,T_{P}}^{E}\left(\theta\right) & = & \sqrt{T_{P}}\left(\frac{1}{T_{P}}\sum_{t=T_{R}}^{T-h}\left(L_{\alpha,\theta}^{E}\left(X_{kt},Y_{t+h}\right)-E\left[L_{\alpha,\theta}^{E}\left(X_{kt},Y_{t+h}\right)\right]\right)\right),\\
	v_{k,T_{P}}^{Q}\left(\theta\right) & = & \sqrt{T_{P}}\left(\frac{1}{T_{P}}\sum_{t=T_{R}}^{T-h}\left(L_{\alpha,\theta}^{Q}\left(X_{kt},Y_{t+h}\right)-E\left[L_{\alpha,\theta}^{Q}\left(X_{kt},Y_{t+h}\right)\right]\right)\right).
\end{eqnarray*}
for $k=1,\ldots,K$ and for $\theta\in\Theta\subseteq\mathbb{R}$.
For $\left(X_{kt},Y_{t+1}\right)\in\mathbb{R}^{2}$, \cite{EGJK_2016}
show that $L_{\theta,\alpha}^{E}\left(X_{kt},Y_{t+h}\right)$ and
$L_{\theta,\alpha}^{Q}\left(X_{kt},Y_{t+h}\right)$ are right continuous,
non-negative and uniformly bounded with a bounded support function of
$\theta$. Let $\left\Vert X\right\Vert _{r}=\left(E\left[\left|X\right|^{r}\right]\right)^{\frac{1}{r}}$
denote a $L^{r}$-norm of a random variable $X$. Let $\varepsilon_{k,t+h}=Y_{t+h}-X_{kt}$
denote the forecast error, and $f_{Y_{t+h}}\left(y\right)$ and $f_{X_{kt}}\left(x\right)$
denote the marginal density functions of $Y_{t+h}$ and $X_{kt}$. 

\begin{lemma} There exists constants $s$, $q$ and $r\in\left[1,\infty\right]$
	and $1/s+1/q=1/r$ such that if $\left\Vert \varepsilon_{t+1}\right\Vert _{s}<\infty$
	and $f_{Y_{t+h}}\left(y\right)$ and $f_{X_{kt}}\left(x\right)$ are
	bounded density functions, 
	\begin{eqnarray*}
		\left\Vert L_{\alpha,\theta}^{E}\left(X_{kt},Y_{t+h}\right)-L_{\alpha,\theta^{\prime}}^{E}\left(X_{kt},Y_{t+h}\right)\right\Vert _{r} & \leq & C^{E}\left|\theta-\theta^{\prime}\right|^{\frac{1}{q}},\\
		\left\Vert L_{\alpha,\theta}^{Q}\left(X_{kt},Y_{t+h}\right)-L_{\alpha,\theta^{\prime}}^{Q}\left(X_{kt},Y_{t+h}\right)\right\Vert _{r} & \leq & C^{Q}\left|\theta-\theta^{\prime}\right|^{\frac{1}{q}},
	\end{eqnarray*}
	for $\left|\theta-\theta^{\prime}\right|\ll1$, where $C^{E}=2+\left\Vert \varepsilon_{k,t+h}\right\Vert _{p}\left(\max f_{X_{kt}}\left(x\right)\right)^{\frac{1}{q}}$
	and $C^{Q}=\left(\max f_{X_{kt}}\left(x\right)\right)^{\frac{1}{q}}\vee\left(\max f_{Y_{t+h}}\left(x\right)\right)^{\frac{1}{q}}$.
\end{lemma}

\begin{proof}[Proof of Lemma 1] Without loss of generality, assume $\theta^{\prime}<\theta$.
	For the case of $L_{\alpha,\theta}^{E}\left(X_{kt},Y_{t+h}\right)$,
	it can be shown that 
	\begin{eqnarray}
	\left\Vert L_{\alpha,\theta}^{E}\left(X_{kt},Y_{t+h}\right)-L_{\alpha,\theta^{\prime}}^{E}\left(X_{kt},Y_{t+h}\right)\right\Vert _{r} & \leq & \left\Vert 1\left\{ Y_{t+h}-X_{kt}<0\right\} -\alpha\right\Vert _{r}\nonumber \\
	&  & \times\left\Vert \left(Y_{t+h}-\theta\right)_{+}-\left(Y_{t+h}-\theta^{\prime}\right)_{+}\right.\nonumber \\
	&  & +\left(X_{kt}-\theta\right)_{+}-\left(X_{kt}-\theta^{\prime}\right)_{+}\nonumber \\
	&  & \left.+\left(Y_{t+h}-X_{kt}\right)\left(1\left\{ X_{kt}>\theta\right\} -1\left\{ X_{kt}>\theta^{\prime}\right\} \right)\right\Vert _{r}\nonumber \\
	& \leq & \left\Vert \left(Y_{t+h}-\theta\right)_{+}-\left(Y_{t+h}-\theta^{\prime}\right)_{+}\right\Vert _{r}\nonumber \\
	&  & +\left\Vert \left(X_{kt}-\theta\right)_{+}-\left(X_{kt}-\theta^{\prime}\right)_{+}\right\Vert _{r}\nonumber \\
	&  & +\left\Vert \left(Y_{t+h}-X_{t}\right)\left(1\left\{ X_{kt}>\theta\right\} -1\left\{ X_{kt}>\theta^{\prime}\right\} \right)\right\Vert _{r},\label{inequality}
	\end{eqnarray}
	by $\left|1\left\{ Y_{t+h}-X_{kt}<0\right\} -\alpha\right|\leq1$
	for any value of $X_{kt}$ and $Y_{t+h}$ and using Minkowski's inequality.
	Also the term $\left|1\left\{ Y_{t+h}-X_{kt}<0\right\} -\alpha\right|$
	does not involves with the parameter $\theta$. We now have a look
	of the first two terms of inequality of (\ref{inequality}). It can
	be shown that for a constant $x$, the function $\left(x-\theta\right)_{+}$
	is Lipschitz continuous for $\theta$, i.e., 
	\begin{equation}
	\left|\left(x-\theta\right)_{+}-\left(x-\theta^{\prime}\right)_{+}\right|\leq K\left|\theta-\theta^{\prime}\right|\label{inequality1}
	\end{equation}
	for some constant $K\geq0$ (Lipschitz constant). To see this, note
	that $\left(x-\theta\right)_{+}=\left(x-\theta\right)1\left\{ x>\theta\right\} $.
	Now if $\theta^{\prime}$, $\theta<x$ or $\theta$, $\theta^{\prime}>x$,
	the left hand side of (\ref{inequality1}) is 0 and the inequality
	of (\ref{inequality1}) always holds. Now if $\theta^{\prime}\leq x\leq\theta$,
	the left hand side of (\ref{inequality1}) is $\left|x-\theta^{\prime}\right|\leq\left|\theta-\theta^{\prime}\right|$.
	Thus the function $\left(x-\theta\right)_{+}$ satisfies Lipschitz
	continuity with Lipschitz constant $K=1$. The first two terms of
	(\ref{inequality}) is each bounded by $\left|\theta-\theta^{\prime}\right|$.
	For the third term of (\ref{inequality}), it can be shown that $1\left\{ X_{kt}>\theta\right\} -1\left\{ X_{kt}>\theta^{\prime}\right\} =1\left\{ \theta^{\prime}<X_{kt}\leq\theta\right\} $,
	since $\theta^{\prime}<\theta$ by assumption. By using the generalized
	H\"{o}lder's inequality, 
	\begin{eqnarray*}
		\left\Vert \left(Y_{t+h}-X_{kt}\right)\left(1\left\{ X_{kt}>\theta\right\} -1\left\{ X_{kt}>\theta^{\prime}\right\} \right)\right\Vert _{r} & \leq & \left\Vert \varepsilon_{k,t+h}\right\Vert _{s}\left\Vert 1\left\{ \theta^{\prime}<X_{kt}\leq\theta\right\} \right\Vert _{q}\\
		& = & \left\Vert \varepsilon_{k,t+h}\right\Vert _{s}\left(\int_{\theta^{\prime}}^{\theta}f_{X_{kt}}\left(x\right)dx\right)^{\frac{1}{q}}\\
		& \leq & \left\Vert \varepsilon_{k,t+h}\right\Vert _{s}\left(\max f_{X_{kt}}\left(x\right)\right)^{\frac{1}{q}}\left|\theta-\theta^{\prime}\right|^{\frac{1}{q}},
	\end{eqnarray*}
	where $s$, $q$ and $r\in\left[1,\infty\right]$ and $1/s+1/q=1/r$.
	With the above results, we can conclude that 
	\begin{eqnarray*}
		\left\Vert L_{\alpha,\theta}^{E}\left(X_{kt},Y_{t+h}\right)-L_{\alpha,\theta^{\prime}}^{E}\left(X_{kt},Y_{t+h}\right)\right\Vert _{r} & \leq & 2\left|\theta-\theta^{\prime}\right|+\left\Vert \varepsilon_{k,t+h}\right\Vert _{s}\left(\max f_{X_{kt}}\left(x\right)\right)^{\frac{1}{q}}\left|\theta-\theta^{\prime}\right|^{\frac{1}{q}}\\
		& \leq & C^{E}\left|\theta-\theta^{\prime}\right|^{\frac{1}{q}},
	\end{eqnarray*}
	when $\left|\theta-\theta^{\prime}\right|\ll1$, where $C^{E}=2+\left\Vert \varepsilon_{k,t+h}\right\Vert _{s}\left(\max f_{X_{kt}}\left(x\right)\right)^{\frac{1}{q}}$.
	
	For the case of $L_{\alpha,\theta}^{Q}\left(X_{kt},Y_{t+h}\right)$,
	by using the generalized H\"{o}lder's inequality, it can be shown that
	\begin{eqnarray*}
		\left\Vert L_{\alpha,\theta}^{Q}\left(X_{kt},Y_{t+h}\right)-L_{\alpha,\theta^{\prime}}^{Q}\left(X_{kt},Y_{t+h}\right)\right\Vert _{r} & \leq & \left\Vert 1\left\{ Y_{t+h}-X_{kt}<0\right\} -\alpha\right\Vert _{s}\\
		&  & \times\left\Vert 1\left\{ X_{kt}>\theta\right\} -1\left\{ X_{kt}>\theta^{\prime}\right\} \right.\\
		&  & \left.-\left(1\left\{ Y_{t+h}>\theta\right\} -1\left\{ Y_{t+h}>\theta^{\prime}\right\} \right)\right\Vert _{q}\\
		& \leq & \left\Vert 1\left\{ X_{kt}>\theta\right\} -1\left\{ X_{kt}>\theta^{\prime}\right\} \right\Vert _{q}\\
		&  & +\left\Vert 1\left\{ Y_{t+h}>\theta\right\} -1\left\{ Y_{t+h}>\theta\right\} \right\Vert _{q}\\
		& = & \left(\int_{\theta^{\prime}}^{\theta}f_{X_{kt}}\left(x\right)dx\right)^{\frac{1}{q}}+\left(\int_{\theta^{\prime}}^{\theta}f_{Y_{t+h}}\left(y\right)dy\right)^{\frac{1}{q}}\\
		& \leq & C^{Q}\times\left|\theta-\theta^{\prime}\right|^{\frac{1}{q}},
	\end{eqnarray*}
	where $s$, $q$ and $r\in\left[1,\infty\right]$ and $1/s+1/q=1/r$
	and $C^{Q}=\left(\max f_{X_{kt}}\left(x\right)\right)^{\frac{1}{q}}\vee\left(\max f_{Y_{t+h}}\left(y\right)\right)^{\frac{1}{q}}$.
	Note that in the first inequality since $\left|1\left\{ Y_{t+h}-X_{kt}<0\right\} -\alpha\right|\leq1$
	for any value of $X_{kt}$ and $Y_{t+h}$, the term$\left\Vert 1\left\{ Y_{t+h}-X_{kt}<0\right\} -\alpha\right\Vert _{s}\leq1$.
	Also $\left|1\left\{ Y_{t+h}-X_{kt}<0\right\} -\alpha\right|$ does
	not involves with the parameter $\theta$.\end{proof}


\begin{lemma} With the pseudometric 
	\[
	\rho_{*}^{E}\left(\theta,\theta^{\prime}\right)=\left\Vert L_{\alpha,\theta}^{E}\left(X_{kt},Y_{t+h}\right)-L_{\alpha,\theta^{\prime}}^{E}\left(X_{kt},Y_{t+h}\right)\right\Vert _{r},
	\]
	if $\left\Vert \varepsilon_{t+1}\right\Vert _{s}<\infty$ and $f_{Y_{t+h}}\left(y\right)$
	and $f_{X_{kt}}\left(x\right)$ are bounded density functions, then
	for every $\epsilon>0$, there exists $\delta>0$ such that 
	\begin{equation}
	\limsup_{T_{P}\rightarrow\infty}\left\Vert \sup_{\rho_{*}^{E}\left(\theta,\theta^{\prime}\right)<\delta}\left|v_{k,T_{P}}^{E}\left(\theta\right)-v_{k,T_{P}}^{E}\left(\theta^{\prime}\right)\right|\right\Vert _{r}<\epsilon\label{secE}
	\end{equation}
	holds for some $2\leq r<s$.
	
	With the pseudometric 
	\[
	\rho_{*}^{E}\left(\theta,\theta^{\prime}\right)=\left\Vert L_{\alpha,\theta}^{Q}\left(X_{kt},Y_{t+h}\right)-L_{\alpha,\theta^{\prime}}^{Q}\left(X_{kt},Y_{t+h}\right)\right\Vert _{r},
	\]
	if $f_{Y_{t+h}}\left(y\right)$ and $f_{X_{kt}}\left(x\right)$ are
	bounded density functions, then for every $\epsilon>0$, there exists
	$\delta>0$ such that 
	\begin{equation}
	\limsup_{T_{P}\rightarrow\infty}\left\Vert \sup_{\rho_{*}^{E}\left(\theta,\theta^{\prime}\right)<\delta}\left|v_{k,T_{P}}^{Q}\left(\theta\right)-v_{k,T_{P}}^{Q}\left(\theta^{\prime}\right)\right|\right\Vert _{r}<\epsilon,\label{secQ}
	\end{equation}
	holds for some $2\leq r<s$. \end{lemma}

\begin{proof}[Proof of Lemma 2] We first prove (\ref{secE}). For integers $l=1,2,\ldots,$
	let $N\left(l\right)=2^{la}$. Let $\Theta$ be a bounded subset of
	$\mathbb{R}^{a}$. In our case $a=1$. Let 
	\[\Theta^{l}=\left\{ \theta^{j}:\theta^{j}\in\Theta,\left|\theta-\theta^{j}\right|\leq Q2^{-l},Q<\infty,j=1,2,\ldots,N\left(l\right)\right\}.\]
	We choose $\theta^{\prime}\in\Theta^{l}$ so that $\left|\theta-\theta^{\prime}\right|\leq Q2^{-l}$.
	Note that the pseudometric $\rho_{*}^{E}\left(\theta,\theta^{\prime}\right)$
	is bounded for any $\left(\theta,\theta^{\prime}\right)$ since $\left|L_{\alpha,\theta}^{E}\left(X_{kt},Y_{t+h}\right)\right|\leq\max\left(\alpha,1-\alpha\right)\times\left|Y_{t+h}-X_{kt}\right|<\left|\varepsilon_{k,t+h}\right|$,
	\[
	\left\Vert L_{\alpha,\theta}^{E}\left(X_{kt},Y_{t+h}\right)-L_{\alpha,\theta^{\prime}}^{E}\left(X_{kt},Y_{t+h}\right)\right\Vert _{r}<2\left\Vert \varepsilon_{k,t+h}\right\Vert _{r}<2\left\Vert \varepsilon_{k,t+h}\right\Vert _{s}<\infty
	\]
	by the assumption that $\left\Vert \varepsilon_{k,t+h}\right\Vert _{s}<\infty.$
	The second inequality is by using the Lyapunov's inequality: for a
	random variable $X$, $\left\Vert X\right\Vert _{r}<\left\Vert X\right\Vert _{s}$
	for $1\leq r<s$. Let 
	\begin{eqnarray*}
		A_{k,T_{P}}^{E}\left(\theta,\theta^{\prime}\right) & = & \frac{1}{\sqrt{T_{P}}}\sum_{t=T_{R}}^{T-h}\left(L_{\alpha,\theta}^{E}\left(X_{kt},Y_{t+h}\right)-L_{\alpha,\theta^{\prime}}^{E}\left(X_{kt},Y_{t+h}\right)\right),\\
		B_{k,T_{P}}^{E}\left(\theta,\theta^{\prime}\right) & = & \frac{1}{\sqrt{T_{P}}}\sum_{t=T_{R}}^{T-h}\left(E\left[L_{\alpha,\theta}^{E}\left(X_{kt},Y_{t+h}\right)\right]-E\left[L_{\alpha,\theta^{\prime}}^{E}\left(X_{kt},Y_{t+h}\right)\right]\right).
	\end{eqnarray*}
	Then 
	\begin{eqnarray*}
		\left\Vert \sup_{\rho_{*}^{E}\left(\theta,\theta^{\prime}\right)<\delta}\left|v_{k,T_{P}}^{E}\left(\theta\right)-v_{k,T_{P}}^{E}\left(\theta^{\prime}\right)\right|\right\Vert _{r} & = & \left\Vert \sup_{\rho_{*}^{E}\left(\theta,\theta^{\prime}\right)<\delta}\left|A_{k,T_{P}}^{E}\left(\theta,\theta^{\prime}\right)-B_{k,T_{P}}^{E}\left(\theta,\theta^{\prime}\right)\right|\right\Vert _{r}\\
		& \leq & \left\Vert \sup_{\rho_{*}^{E}\left(\theta,\theta^{\prime}\right)<\delta}\left|A_{k,T_{P}}^{E}\left(\theta,\theta^{\prime}\right)\right|+\sup_{\rho_{*}^{E}\left(\theta,\theta^{\prime}\right)<\delta}\left|B_{k,T_{P}}^{E}\left(\theta,\theta^{\prime}\right)\right|\right\Vert _{r}\\
		& \leq & \left\Vert \sup_{\rho_{*}^{E}\left(\theta,\theta^{\prime}\right)<\delta}\left|A_{k,T_{P}}^{E}\left(\theta,\theta^{\prime}\right)\right|\right\Vert _{r}+\left\Vert \sup_{\rho_{*}^{E}\left(\theta,\theta^{\prime}\right)<\delta}\left|B_{k,T_{P}}^{E}\left(\theta,\theta^{\prime}\right)\right|\right\Vert _{r}
	\end{eqnarray*}
	For the second term of the above inequality, 
	\begin{eqnarray*}
		\left\Vert \sup_{\rho_{*}^{E}\left(\theta,\theta^{\prime}\right)<\delta}\left|B_{k,T_{P}}^{E}\left(\theta,\theta^{\prime}\right)\right|\right\Vert _{r} & \leq & \frac{1}{\sqrt{T_{P}}}\sum_{t=T_{R}}^{T-h}\left\Vert \sup_{\rho_{*}^{E}\left(\theta,\theta^{\prime}\right)<\delta}E\left[\left|L_{\alpha,\theta}^{E}\left(X_{kt},Y_{t+h}\right)-L_{\alpha,\theta^{\prime}}^{E}\left(X_{kt},Y_{t+h}\right)\right|\right]\right\Vert _{r}\\
		& = & \frac{1}{\sqrt{T_{P}}}\sum_{t=T_{R}}^{T-h}\sup_{\rho_{*}^{E}\left(\theta,\theta^{\prime}\right)<\delta}E\left[\left|L_{\alpha,\theta}^{E}\left(X_{kt},Y_{t+h}\right)-L_{\alpha,\theta^{\prime}}^{E}\left(X_{kt},Y_{t+h}\right)\right|\right]\\
		& < & \frac{1}{\sqrt{T_{P}}}\sum_{t=T_{R}}^{T-h}\sup_{\rho_{*}^{E}\left(\theta,\theta^{\prime}\right)<\delta}\left\Vert L_{\alpha,\theta}^{E}\left(X_{kt},Y_{t+h}\right)-L_{\alpha,\theta^{\prime}}^{E}\left(X_{kt},Y_{t+h}\right)\right\Vert _{r}\\
		& \leq & \frac{1}{\sqrt{T_{P}}}\sum_{t=T_{R}}^{T-h}\sup_{\rho_{*}^{E}\left(\theta,\theta^{\prime}\right)<\delta}C^{E}\left|\theta-\theta^{\prime}\right|^{\frac{1}{q}}.
	\end{eqnarray*}
	The third inequality is again by using the Lyapunov's inequality.
	The last inequality is by using Lemma 1 and the constant $C^{E}=2+\left\Vert \varepsilon_{k,t+h}\right\Vert _{s}\left|\max f_{X_{kt}}\left(x\right)\right|^{\frac{1}{q}},$
	where $\varepsilon_{k,t+1}=Y_{t+h}-X_{kt}$ and $s$, $q\in\left[1,\infty\right]$,
	$1/s+1/q=1/r$. For the first term, 
	\begin{eqnarray*}
		\left\Vert \sup_{\rho_{*}^{E}\left(\theta,\theta^{\prime}\right)<\delta}\left|A_{k,T_{P}}^{E}\left(\theta,\theta^{\prime}\right)\right|\right\Vert _{r} & \leq & \frac{1}{\sqrt{T_{P}}}\sum_{t=T_{R}}^{T-h}\left\Vert \sup_{\rho_{*}^{E}\left(\theta,\theta^{\prime}\right)<\delta}\left|L_{\alpha,\theta}^{E}\left(X_{kt},Y_{t+h}\right)-L_{\alpha,\theta^{\prime}}^{E}\left(X_{kt},Y_{t+h}\right)\right|\right\Vert _{r}
	\end{eqnarray*}
	It can be shown that 
	\begin{eqnarray*}
		\left|L_{\alpha,\theta}^{E}\left(X_{kt},Y_{t+h}\right)-L_{\alpha,\theta^{\prime}}^{E}\left(X_{kt},Y_{t+h}\right)\right| & \leq & \left|\left(Y_{t+1}-\theta\right)_{+}-\left(Y_{t+1}-\theta^{\prime}\right)_{+}\right|+\left|\left(X_{kt}-\theta\right)_{+}-\left(X_{kt}-\theta^{\prime}\right)_{+}\right|\\
		& + & \left|\left(Y_{t+1}-X_{kt}\right)\left(1\left\{ X_{kt}>\theta\right\} -1\left\{ X_{kt}>\theta^{\prime}\right\} \right)\right|\\
		& \leq & 2\left|\theta-\theta^{\prime}\right|+\left|Y_{t+1}-X_{kt}\right|1\left\{ \theta^{\prime}<X_{kt}\leq\theta\right\} .
	\end{eqnarray*}
	Thus 
	\begin{eqnarray*}
		\sup_{\rho_{*}^{E}\left(\theta,\theta^{\prime}\right)<\delta}\left|L_{\alpha,\theta}^{E}\left(X_{kt},Y_{t+h}\right)-L_{\alpha,\theta^{\prime}}^{E}\left(X_{kt},Y_{t+h}\right)\right| & \leq & 2\sup_{\rho_{*}^{E}\left(\theta,\theta^{\prime}\right)<\delta}\left|\theta-\theta^{\prime}\right|\\
		& + & \left|Y_{t+1}-X_{kt}\right|\sup_{\rho_{*}^{E}\left(\theta,\theta^{\prime}\right)<\delta}1\left\{ \theta^{\prime}<X_{kt}\leq\theta\right\} .
	\end{eqnarray*}
	Then 
	\begin{eqnarray*}
		\left\Vert \sup_{\rho_{*}^{E}\left(\theta,\theta^{\prime}\right)<\delta}\left|L_{\alpha,\theta}^{E}\left(X_{kt},Y_{t+h}\right)-L_{\alpha,\theta^{\prime}}^{E}\left(X_{kt},Y_{t+h}\right)\right|\right\Vert _{r} & \leq & 2\sup_{\rho_{*}^{E}\left(\theta,\theta^{\prime}\right)<\delta}\left|\theta-\theta^{\prime}\right|+\\
		& + & \left\Vert Y_{t+h}-X_{kt}\right\Vert _{s}\left\Vert \sup_{\rho_{*}^{E}\left(\theta,\theta^{\prime}\right)<\delta}1\left\{ \theta^{\prime}<X_{kt}\leq\theta\right\} \right\Vert _{q}.
	\end{eqnarray*}
	The second term of the above inequality is obtained with the generalized
	H\"{o}lder's inequality and $s$, $q$$\in\left[1,\infty\right]$ and
	$1/s+1/q=1/r$. With Assumptions 1 and 3, using similar arguments
	used in proving Lemma 1 of \citet{LMW_2005}, there exists a constant
	$C_{0}$ such that 
	\begin{eqnarray*}
		\left\Vert \sup_{\rho_{*}^{E}\left(\theta,\theta^{\prime}\right)<\delta}1\left\{ \theta^{\prime}<X_{kt}\leq\theta\right\} \right\Vert _{q} & = & \left(E\left|\sup_{\rho_{*}^{E}\left(\theta,\theta^{\prime}\right)<\delta}1\left\{ \theta^{\prime}<X_{kt}\leq\theta\right\} \right|^{q}\right)^{\frac{1}{q}}\\
		& \leq & \left(E\left|\sup_{\rho_{*}^{E}\left(\theta,\theta^{\prime}\right)<\delta}1\left\{ \theta^{\prime}<X_{kt}\leq\theta+\left(\theta-\theta^{\prime}\right)\right\} \right|^{q}\right)^{\frac{1}{q}}\\
		& \leq & \left(E\left|1\left\{ \left|X_{t}-\theta\right|\leq\left|\theta-\theta^{\prime}\right|\right\} \right|^{q}\right)^{\frac{1}{q}}\\
		& \leq & C_{0}\left|\theta-\theta^{\prime}\right|^{\frac{1}{q}},
	\end{eqnarray*}
	where $\theta$ and $\theta^{\prime}$ satisfy $\rho_{*}^{E}\left(\theta,\theta^{\prime}\right)<\delta$.
	If we take $\left|\theta-\theta^{\prime}\right|$ very small (say
	$\left|\theta-\theta^{\prime}\right|\ll1$), we may conclude that
	\[
	\left\Vert \sup_{\rho_{*}^{E}\left(\theta,\theta^{\prime}\right)<\delta}\left|L_{\alpha,\theta}^{E}\left(X_{kt},Y_{t+h}\right)-L_{\alpha,\theta^{\prime}}^{E}\left(X_{kt},Y_{t+h}\right)\right|\right\Vert _{r}\leq C_{1}\sup_{\rho_{*}^{E}\left(\theta,\theta^{\prime}\right)<\delta}\left|\theta-\theta^{\prime}\right|^{\frac{1}{q}},
	\]
	where $C_{1}=2+\left\Vert \varepsilon_{k,t+h}\right\Vert _{s}C_{0}$.
	Therefore 
	\[
	\left\Vert \sup_{\rho_{*}^{E}\left(\theta,\theta^{\prime}\right)<\delta}\left|A_{k,T_{P}}^{E}\left(\theta,\theta^{\prime}\right)\right|\right\Vert _{r}\leq\frac{1}{\sqrt{T_{P}}}\sum_{t=T_{R}}^{T-h}C_{1}\sup_{\rho_{*}^{E}\left(\theta,\theta^{\prime}\right)<\delta}\left|\theta-\theta^{\prime}\right|^{\frac{1}{q}}
	\]
	Combining the above results, we have 
	\[
	\left\Vert \sup_{\rho_{*}^{E}\left(\theta,\theta^{\prime}\right)<\delta}\left|v_{k,T_{P}}^{E}\left(\theta\right)-v_{k,T_{P}}^{E}\left(\theta^{\prime}\right)\right|\right\Vert _{r}\leq\frac{1}{\sqrt{T_{P}}}\sum_{t=T_{R}}^{T-h}C_{2}\sup_{\rho_{*}^{E}\left(\theta,\theta^{\prime}\right)<\delta}\left|\theta-\theta^{\prime}\right|^{\frac{1}{q}},
	\]
	where $C_{2}=C^{E}\vee C_{1}$. Note that here $\left|\theta-\theta^{\prime}\right|\leq Q/2^{l}$.
	Following \citet{Hansen_1996}, we can choose $l=l\left(T_{P}\right)$
	depending on $T_{P}$ such that $\sqrt{T_{P}}2^{-l\left(T_{P}\right)/q}\rightarrow0$
	as $T_{P}\rightarrow\infty$. Then the right hand side of the above
	inequality will becomes arbitrage small as $T_{P}\rightarrow\infty$.
	With a suitable choice for $Q$, we may set the corresponding $\delta=2^{-l\left(T_{P}\right)/q}$.
	Finally note that the condition of mixing coefficients in Assumption
	4 in \citet{Hansen_1996} is implied by Assumption 1. In addition,
	since $0\leq L_{\alpha,\theta}^{E}\left(X_{kt},Y_{t+h}\right)\leq\left(\alpha\vee\left(1-\alpha\right)\right)\times\left|\varepsilon_{k,t+h}\right|$,
	\begin{eqnarray*}
		\limsup_{T_{P}\rightarrow\infty}\frac{1}{T_{P}}\left(\sum_{t=T_{R}}^{T-h}\left\Vert L_{\alpha,\theta}^{E}\left(X_{kt},Y_{t+h}\right)\right\Vert _{s}^{2}\right)^{\frac{1}{2}} & \leq & \limsup_{T_{P}\rightarrow\infty}\frac{1}{T_{P}}\left(\sum_{t=T_{R}}^{T-h}\left(\alpha\vee\left(1-\alpha\right)\right)^{2}\times\left(E\left[\left|\varepsilon_{k,t+h}\right|^{s}\right]\right)^{\frac{2}{s}}\right)^{\frac{1}{2}}\\
		& < & \infty
	\end{eqnarray*}
	by the assumption of $\left\Vert \varepsilon_{k,t+h}\right\Vert _{s}<\infty$
	and the second condition of Assumption 4 (equation (12)) in \citet{Hansen_1996}
	holds. The rest proof can be completed by using arguments in proving
	Theorem 1 of \citet{Hansen_1996} and comparison of pairs of \citet{AP_1994}.
	
	For the case of (\ref{secQ}), it can be shown that the pseudometric
	$\rho_{*}^{E}\left(\theta,\theta^{\prime}\right)$ is bounded for
	any $\left(\theta,\theta^{\prime}\right)$ since $\left|L_{\alpha,\theta}^{Q}\left(X_{kt},Y_{t+h}\right)\right|\leq\max\left(\alpha,1-\alpha\right)$,
	\[
	\left\Vert L_{\alpha,\theta}^{Q}\left(X_{kt},Y_{t+h}\right)-L_{\alpha,\theta^{\prime}}^{Q}\left(X_{kt},Y_{t+h}\right)\right\Vert _{s}\leq2\max\left(\alpha,1-\alpha\right)<2.
	\]
	Again let 
	\begin{eqnarray*}
		A_{k,T_{P}}^{Q}\left(\theta,\theta^{\prime}\right) & = & \frac{1}{\sqrt{T_{P}}}\sum_{t=T_{R}}^{T-h}\left(L_{\alpha,\theta}^{Q}\left(X_{kt},Y_{t+h}\right)-L_{\alpha,\theta^{\prime}}^{Q}\left(X_{kt},Y_{t+h}\right)\right),\\
		B_{k,T_{P}}^{Q}\left(\theta,\theta^{\prime}\right) & = & \frac{1}{\sqrt{T_{P}}}\sum_{t=T_{R}}^{T-h}\left(E\left[L_{\alpha,\theta}^{Q}\left(X_{kt},Y_{t+h}\right)\right]-E\left[L_{\alpha,\theta^{\prime}}^{Q}\left(X_{kt},Y_{t+h}\right)\right]\right).
	\end{eqnarray*}
	Then 
	\begin{eqnarray*}
		\left\Vert \sup_{\rho_{*}^{E}\left(\theta,\theta^{\prime}\right)<\delta}\left|v_{k,T_{P}}^{Q}\left(\theta\right)-v_{k,T_{P}}^{Q}\left(\theta^{\prime}\right)\right|\right\Vert _{r} & \leq & \left\Vert \sup_{\rho_{*}^{E}\left(\theta,\theta^{\prime}\right)<\delta}\left|A_{k,T_{P}}^{Q}\left(\theta,\theta^{\prime}\right)\right|\right\Vert _{r}+\left\Vert \sup_{\rho_{*}^{E}\left(\theta,\theta^{\prime}\right)<\delta}\left|B_{k,T_{P}}^{Q}\left(\theta,\theta^{\prime}\right)\right|\right\Vert _{r}
	\end{eqnarray*}
	For the second term of the above inequality, by using a similar argument
	used in previous proof, it can be shown that 
	\begin{eqnarray*}
		\left\Vert \sup_{\rho_{*}^{E}\left(\theta,\theta^{\prime}\right)<\delta}\left|B_{k,T_{P}}^{Q}\left(\theta,\theta^{\prime}\right)\right|\right\Vert _{r} & \leq & \frac{1}{\sqrt{T_{P}}}\sum_{t=T_{R}}^{T-h}\sup_{\rho_{*}^{E}\left(\theta,\theta^{\prime}\right)<\delta}C^{Q}\left|\theta-\theta^{\prime}\right|^{\frac{1}{q}}.
	\end{eqnarray*}
	Here $C^{Q}=\left(\max f_{X_{kt}}\left(x\right)\right)^{\frac{1}{q}}\vee\left(\max f_{Y_{t+h}}\left(y\right)\right)^{\frac{1}{q}}$
	and the constant $q$ satisfies that $1/s+1/q=1/r$ and $s$, $q\in\left[1,\infty\right]$.
	For the first term, 
	\begin{eqnarray*}
		\left\Vert \sup_{\rho_{*}^{E}\left(\theta,\theta^{\prime}\right)<\delta}\left|A_{k,T_{P}}^{Q}\left(\theta,\theta^{\prime}\right)\right|\right\Vert _{r} & \leq & \frac{1}{\sqrt{T_{P}}}\sum_{t=T_{R}}^{T-h}\left\Vert \sup_{\rho_{*}^{E}\left(\theta,\theta^{\prime}\right)<\delta}\left|L_{\alpha,\theta}^{Q}\left(X_{kt},Y_{t+h}\right)-L_{\alpha,\theta^{\prime}}^{Q}\left(X_{kt},Y_{t+h}\right)\right|\right\Vert _{r}.
	\end{eqnarray*}
	It can be shown that 
	\begin{eqnarray*}
		\left|L_{\alpha,\theta}^{Q}\left(X_{kt},Y_{t+h}\right)-L_{\alpha,\theta^{\prime}}^{Q}\left(X_{kt},Y_{t+h}\right)\right| & \leq & \left|1\left\{ Y_{t+h}-X_{kt}<0\right\} -\alpha\right|\left(\left|1\left\{ \theta^{\prime}<X_{kt}\leq\theta\right\} \right|\right.\\
		&  & \left.+\left|1\left\{ \theta^{\prime}<Y_{t+h}\leq\theta\right\} \right|\right)\\
		& \leq & \left|1\left\{ \theta^{\prime}<X_{kt}\leq\theta\right\} \right|+\left|1\left\{ \theta^{\prime}<Y_{t+h}\leq\theta\right\} \right|.
	\end{eqnarray*}
	Thus 
	\begin{eqnarray*}
		\sup_{\rho_{*}^{E}\left(\theta,\theta^{\prime}\right)<\delta}\left|L_{\alpha,\theta}^{Q}\left(X_{kt},Y_{t+h}\right)-L_{\alpha,\theta^{\prime}}^{Q}\left(X_{kt},Y_{t+h}\right)\right| & \leq & \sup_{\rho_{*}^{E}\left(\theta,\theta^{\prime}\right)<\delta}\left|1\left\{ \theta^{\prime}<X_{kt}\leq\theta\right\} \right|\\
		&  & +\sup_{\rho_{*}^{E}\left(\theta,\theta^{\prime}\right)<\delta}\left|1\left\{ \theta^{\prime}<Y_{t+h}\leq\theta\right\} \right|.
	\end{eqnarray*}
	Then 
	\begin{eqnarray*}
		\left\Vert \sup_{\rho_{*}^{E}\left(\theta,\theta^{\prime}\right)<\delta}\left|L_{\alpha,\theta}^{Q}\left(X_{kt},Y_{t+h}\right)-L_{\alpha,\theta^{\prime}}^{Q}\left(X_{kt},Y_{t+h}\right)\right|\right\Vert _{r} & \leq & \left\Vert \sup_{\rho_{*}^{E}\left(\theta,\theta^{\prime}\right)<\delta}1\left\{ \theta^{\prime}<X_{kt}\leq\theta\right\} \right\Vert _{r}\\
		& + & \left\Vert \sup_{\rho_{*}^{E}\left(\theta,\theta^{\prime}\right)<\delta}1\left\{ \theta^{\prime}<Y_{t+h}\leq\theta\right\} \right\Vert _{r}.
	\end{eqnarray*}
	Like in previous proof, with Assumptions 1 and 3, we can use similar
	arguments used in proving Lemma 1 of \citet{LMW_2005} to show that
	there exists constant $C_{3}$ and $C_{4}$ such that 
	\begin{eqnarray*}
		\left\Vert \sup_{\rho_{*}^{E}\left(\theta,\theta^{\prime}\right)<\delta}1\left\{ \theta^{\prime}<X_{kt}\leq\theta\right\} \right\Vert _{r} & \leq & C_{3}\left|\theta-\theta^{\prime}\right|^{\frac{1}{r}},\\
		\left\Vert \sup_{\rho_{*}^{E}\left(\theta,\theta^{\prime}\right)<\delta}1\left\{ \theta^{\prime}<Y_{t+h}\leq\theta\right\} \right\Vert _{r} & \leq & C_{4}\left|\theta-\theta^{\prime}\right|^{\frac{1}{r}},
	\end{eqnarray*}
	where $\theta$ and $\theta^{\prime}$ satisfy $\rho_{*}^{E}\left(\theta,\theta^{\prime}\right)<\delta$.
	If we take $\left|\theta-\theta^{\prime}\right|$ very small (say
	$\left|\theta-\theta^{\prime}\right|\ll1$), we may conclude that
	\begin{eqnarray*}
		\left\Vert \sup_{\rho_{*}^{E}\left(\theta,\theta^{\prime}\right)<\delta}\left|L_{\alpha,\theta}^{Q}\left(X_{kt},Y_{t+h}\right)-L_{\alpha,\theta^{\prime}}^{Q}\left(X_{kt},Y_{t+h}\right)\right|\right\Vert _{r} & \leq & C_{5}\sup_{\rho_{*}^{E}\left(\theta,\theta^{\prime}\right)<\delta}\left|\theta-\theta^{\prime}\right|^{\frac{1}{r}}\\
		& \leq & C_{5}\sup_{\rho_{*}^{E}\left(\theta,\theta^{\prime}\right)<\delta}\left|\theta-\theta^{\prime}\right|^{\frac{1}{q}},
	\end{eqnarray*}
	where $C_{5}=C_{3}\vee C_{4}$. The second inequality is due to $1/q\leq1/r$
	by $1/s+1/q=1/r$ and $s$, $q\in\left[1,\infty\right]$. Therefore
	\[
	\left\Vert \sup_{\rho_{*}^{E}\left(\theta,\theta^{\prime}\right)<\delta}\left|A_{k,T_{P}}^{Q}\left(\theta,\theta^{\prime}\right)\right|\right\Vert _{q}\leq\frac{1}{\sqrt{T_{P}}}\sum_{t=T_{R}}^{T-h}C_{5}\sup_{\rho_{*}^{E}\left(\theta,\theta^{\prime}\right)<\delta}\left|\theta-\theta^{\prime}\right|^{\frac{1}{q}}
	\]
	Combining the above results, we have 
	\[
	\left\Vert \sup_{\rho_{*}^{E}\left(\theta,\theta^{\prime}\right)<\delta}\left|v_{k,T_{P}}^{Q}\left(\theta\right)-v_{k,T_{P}}^{Q}\left(\theta^{\prime}\right)\right|\right\Vert _{r}\leq\frac{1}{\sqrt{T_{P}}}\sum_{t=T_{R}}^{T-h}C_{6}\left|\theta-\theta^{\prime}\right|^{q},
	\]
	where $C_{6}=C^{Q}\vee C_{5}$. Again, we can choose $l=l\left(T_{P}\right)$
	depending on $T_{P}$ such that $\sqrt{T_{P}}2^{-l\left(T_{P}\right)/q}\rightarrow0$
	as $T_{P}\rightarrow\infty$. Then the right hand side of the above
	inequality will becomes arbitrage small as $T_{P}\rightarrow\infty$.
	With a suitable choice for $Q$, we may set the corresponding $\delta=2^{-l\left(T_{P}\right)/q}$.
	Finally note that the condition of mixing coefficients in Assumption
	4 in \citet{Hansen_1996} is implied by Assumption 1. In addition,
	since $0\leq L_{\alpha,\theta}^{Q}\left(X_{kt},Y_{t+h}\right)\leq\alpha\vee\left(1-\alpha\right)$,
	\begin{eqnarray*}
		\limsup_{T_{P}\rightarrow\infty}\frac{1}{T_{P}}\left(\sum_{t=T_{R}}^{T-h}\left\Vert L_{\alpha,\theta}^{E}\left(X_{kt},Y_{t+h}\right)\right\Vert _{s}^{2}\right)^{\frac{1}{2}} & \leq & \limsup_{T_{P}\rightarrow\infty}\frac{1}{T_{P}}\left(\sum_{t=T_{R}}^{T-h}\left(\alpha\vee\left(1-\alpha\right)\right)^{2}\right)^{\frac{1}{2}}\\
		& < & \infty
	\end{eqnarray*}
	and the second condition of Assumption 4 (equation (12)) in \citet{Hansen_1996}
	holds. The rest proof can be completed by using arguments in proving
	Theorem 1 of \citet{Hansen_1996} and comparison of pairs of \citet{AP_1994}.\end{proof}
With Lemma 1 and 2, we can have the following result.
\begin{lemma}
	
	Assume Assumptions 1-3 hold. Then for $i\in\{E,Q\}$, $\theta_{1},\theta_{2}\in\Theta\subseteq\mathbb{R}$
	and $k,l=1,\ldots,K$, $k\neq l$, with the following pseudometric
	\[
	\rho_{d}^{i}\left(\theta_{1},\theta_{2}\right)=\left\Vert L_{\alpha,\theta_{1}}^{i}\left(X_{kt},Y_{t+h}\right)-L_{\alpha,\theta_{1}}^{i}\left(X_{lt},Y_{t+h}\right)-\left[L_{\alpha,\theta_{2}}^{i}\left(X_{kt},Y_{t+h}\right)-L_{\alpha,\theta_{2}}^{i}\left(X_{lt},Y_{t+h}\right)\right]\right\Vert _{s}
	\]
	we have 
	\[
	v_{k,T_{P}}^{i}\left(\theta\right)-v_{l,T_{P}}^{i}\left(\theta\right)\Rightarrow\tilde{g}_{kl}^{i}\left(\theta\right),
	\]
	where $\tilde{g}_{kl}^{i}\left(\theta\right)$ is a mean zero Gaussian
	process with covariance 
	\[
	var_{kl}^{i}\left(\theta_{1},\theta_{2}\right)=\lim_{T_{P}\rightarrow\infty}E\left[\left(v_{k,T_{P}}^{i}\left(\theta_{1}\right)-v_{l,T_{P}}^{i}\left(\theta_{1}\right)\right)\left(v_{k,T_{P}}^{i}\left(\theta_{2}\right)-v_{l,T_{P}}^{i}\left(\theta_{2}\right)\right)\right].
	\]
	In addition, except at zero, the sample paths of $\tilde{g}_{kl}^{i}\left(\theta\right)$
	are uniformly continuous with respect to the pseudometric $\rho_{d}^{i}\left(\theta_{1},\theta_{2}\right)$
	on $\Theta$ with probability one.
\end{lemma}
\begin{proof}[Proof of Lemma 3]
	The proof is similar as the one in proving Lemma 4 of \citet{LMW_2005}.
	We need to verify the following three conditions (Theorem 10.2 of
	\citet{Pollard_1990}): 
		\begin{condition} Total boundedness of pseudometric spaces $\left(\Theta,\rho_{d}^{i}\right)$,
		$i\in\left\{ E,Q\right\} $.\end{condition} 
		\begin{condition} Stochastic equicontinuity of $\left\{ v_{k,T_{P}}^{i}\left(\theta\right)-v_{l,T_{P}}^{i}\left(\theta\right):T_{P}\text{\ensuremath{\ge}}1,i\in\left\{ E,Q\right\} \right\} $ \end{condition}
		\begin{condition}
		Finite dimensional (fidi) convergence.
		\end{condition}  
	It can be shown that Conditions 1 and 2 are satisfied by using
	Lemma 1. For Condition 3, we need to show that \[\left(v_{k,T_{P}}^{i}\left(\theta_{1}\right)-v_{l,T_{P}}^{i}\left(\theta_{1}\right),v_{k,T_{P}}^{i}\left(\theta_{2}\right)-v_{l,T_{P}}^{i}\left(\theta_{2}\right),\ldots,v_{k,T_{P}}^{i}\left(\theta_{J}\right)-v_{l,T_{P}}^{i}\left(\theta_{J}\right)\right)\]converge in distribution to $\left(\tilde{d}_{kl}^{i}\left(\theta_{1}\right),\tilde{d}_{kl}^{i}\left(\theta_{2}\right),\ldots,\tilde{d}_{kl}^{i}\left(\theta_{J}\right)\right)$
	for all $\theta_{j}\in\Theta$ and $J\geq1$. For the case of $i=E$,
	this can be first established by using convergence results of sum
	of strong-mixing stationary sequences, such as Corollary 5.1 of \citet{HH_1980}.
	Let $\Delta_{kt}^{E}\left(\theta_{j}\right)=L_{\alpha,\theta_{j}}^{E}\left(X_{kt},Y_{t+h}\right)-E\left[L_{\alpha,\theta_{j}}^{E}\left(X_{kt},Y_{t+h}\right)\right]$,
	$t=T_{R},\ldots,T-h$ and $j=1,\ldots,J$. Then $v_{k,T_{P}}^{i}\left(\theta_{j}\right)=T_{P}^{-1/2}\sum_{t=T_{R}}^{T-h}\Delta_{kt}^{E}\left(\theta_{j}\right)$.
	By Assumption 1, it can be seen that $E\left[\Delta_{kt}^{E}\left(\theta_{1}\right)\right]=0$
	and the mixing coefficients $\alpha\left(n\right)$ satisfy $\sum_{n=1}^{\infty}\left[\alpha\left(n\right)\right]^{\delta/\left(2+\delta\right)}\leq\sum_{n=1}^{\infty}\left[\alpha\left(n\right)\right]^{A}<\infty$.
	Also $E\left[\left|\Delta_{kt}^{E}\left(\theta_{1}\right)\right|^{2+\delta}\right]<2^{2+\delta}\left\Vert \varepsilon_{k,t+h}\right\Vert _{2+\delta}^{2+\delta}\leq\left\Vert \varepsilon_{k,t+h}\right\Vert _{s}^{2+\delta}<\infty$
	by the Lyapunov's inequality and Assumption 2. Thus the conditions
	in Corollary 5.1 of \citet{HH_1980} are satisfied. For $v_{l,T_{P}}^{i}(\theta_{j})$,
	the same conditions also hold. Then by using the Cramer-Wold theorem,
	the result of fidi can be constructed. For the case of $i=Q$, note
	that $\Delta_{kt}^{Q}\left(\theta_{j}\right)=L_{\alpha,\theta_{j}}^{Q}\left(X_{kt},Y_{t+h}\right)-E\left[L_{\alpha,\theta_{j}}^{Q}\left(X_{kt},Y_{t+h}\right)\right]\leq\max\left(\alpha,1-\alpha\right)<\infty$
	is bounded. Thus the results of fidi for this case can be established
	by using similar arguments for the case of $i=E$. \end{proof}
 
\begin{proof}[Proof of Theorem 1] Under the null, if $S_{T_{P},\alpha}^{i}=0$, then
	at least there exists a pair $\left(k,l\right)$ such that $\sup_{\theta\in\Theta}D_{kl,\alpha}^{i}\left(\theta\right)=0$.
	This implies that for the pair $\left(k,l\right)$, $D_{kl,\alpha}^{i}\left(\theta\right)\leq0$
	for all $\theta\in\Theta$ and $D_{kl,\alpha}^{i}\left(\theta\right)=0$
	for some $\theta\in\mathcal{A}_{kl}^{i}$, where $\mathcal{A}_{kl}^{i}=\left\{ \theta\in\Theta,D_{kl,\alpha}^{i}\left(\theta\right)=0\right\} .$
	We need to show that $\sup_{\theta\in\Theta}\sqrt{T_{P}}\hat{D}_{kl,\alpha}^{i}\left(\theta\right)\Rightarrow\sup_{\theta\in\mathcal{A}_{kl}^{i}}\tilde{g}_{kl}^{i}\left(\theta\right)$.
	For $\hat{D}_{kl,\alpha}^{i}\left(\theta\right)$, we can have 
	\begin{eqnarray*}
		\sqrt{T_{P}}\hat{D}_{kl,\alpha}^{i}\left(\theta\right) & = & B_{1,kl}^{i}\left(\theta\right)+B_{2,kl}^{i}\left(\theta\right),\\
		B_{1,kl}^{i}\left(\theta\right) & = & v_{k,T_{P}}^{i}\left(\theta\right)-v_{l,T_{P}}^{i}\left(\theta\right),\\
		B_{2,kl}^{i}\left(\theta\right) & = & \sqrt{T_{P}}\left(E\left[L_{\theta,\alpha}^{i}\left(X_{lt},Y_{t+h}\right)\right]-E\left[L_{\theta,\alpha}^{i}\left(X_{kt},Y_{t+h}\right)\right]\right).
	\end{eqnarray*}
	If Assumptions 1-3 hold, by using Lemma 3 and the continuous mapping
	theorem, it can be shown that $\sup_{\theta\in\mathcal{A}_{kl}^{i}}B_{1,kl}^{i}\left(\theta\right)\Rightarrow\sup_{\theta\in\mathcal{A}_{kl}^{i}}\tilde{g}_{kl}^{i}\left(\theta\right)$.
	By definition of $\mathcal{A}_{kl}^{i}$, $\sup_{\theta\in\mathcal{A}_{kl}^{i}}\sqrt{T_{P}}\hat{D}_{kl,\alpha}^{i}\left(\theta\right)=\sup_{\theta\in\mathcal{A}_{kl}^{i}}\sqrt{T_{P}}B_{1,kl}^{i}\left(\theta\right)$
	and thus $\sup_{\theta\in\mathcal{A}_{kl}^{i}}\sqrt{T_{P}}\hat{D}_{kl,\alpha}^{i}\left(\theta\right)\Rightarrow\sup_{\theta\in\mathcal{A}_{kl}^{i}}\tilde{g}_{kl}^{i}\left(\theta\right)$.
	Now we verify that $\sup_{\theta\in\Theta}\sqrt{T_{P}}\hat{D}_{kl,\alpha}^{i}\left(\theta\right)\Rightarrow\sup_{\theta\in\mathcal{A}_{kl}^{i}}\sqrt{T_{P}}\hat{D}_{kl,\alpha}^{i}\left(\theta\right)$.
	To see this, note that 
	\[
	\sup_{\theta\in\Theta}\sqrt{T_{P}}\hat{D}_{kl,\alpha}^{i}\left(\theta\right)=\sup_{\theta\in\Theta}\left[B_{1,kl}^{i}\left(\theta\right)+B_{2,kl}^{i}\left(\theta\right)\right].
	\]
	If $\mathcal{A}_{kl}^{i}$ is non-empty and the supremum occurs when
	$\theta\in\mathcal{A}_{kl}^{i}\subseteq\Theta$, it is trivial to
	see that 
	\[
	\sup_{\theta\in\Theta}\sqrt{T_{P}}\hat{D}_{kl,\alpha}^{i}\left(\theta\right)=\sup_{\theta\in\mathcal{A}_{kl}^{i}}\sqrt{T_{P}}\hat{D}_{kl,\alpha}^{i}\left(\theta\right)=\sup_{\theta\in\mathcal{A}_{kl}^{i}}B_{1,kl}^{i}\left(\theta\right)\Rightarrow\sup_{\theta\in\mathcal{A}_{kl}^{i}}\tilde{g}_{kl}^{i}\left(\theta\right).
	\]
	If $\mathcal{A}_{kl}^{i}$ is non-empty but the supremum occurs when
	$\theta\in\Theta/\mathcal{A}_{kl}^{i}$, $E\left[L_{\theta,\alpha}^{i}\left(X_{lt},Y_{t+h}\right)\right]-E\left[L_{\theta,\alpha}^{i}\left(X_{kt},Y_{t+h}\right)\right]\neq0$
	and the term $B_{2,kl}^{i}\left(\theta\right)$ will diverge as $T_{P}\rightarrow\infty$
	and $\sup_{\theta\in\Theta}\sqrt{T_{P}}\hat{D}_{kl,\alpha}^{i}\left(\theta\right)$
	will also diverge. By continuous mapping theorem, in this case the
	asymptotic distribution of the test statistic $\hat{S}_{T_{P},\alpha}^{i}$
	will not be affected. Now if $S_{T_{P},\alpha}^{i}<0$, $\mathcal{A}_{kl}^{i}$
	is empty. It implies that for some pairs $\left(k,l\right)$, $D_{kl,\alpha}^{i}\left(\theta\right)<0$
	for all $\theta\in\Theta$ and $B_{2,kl}^{i}\left(\theta\right)\rightarrow-\infty$
	as $T_{P}\rightarrow\infty$. Then $\sup_{\theta\in\Theta}\sqrt{T_{P}}\hat{D}_{kl,\alpha}^{i}\left(\theta\right)\rightarrow-\infty$.
\end{proof}

\begin{proof}[Proof of Theorem 2]
	
	To prove the first part of the theorem, we can use Theorem 2 of \citet{PR_1994}.
	To see this, note that $E\left[\left|\hat{d}_{t,kl}^{i*}\left(\theta\right)\right|^{2+\varrho}\right]<\infty$
	for some $\varrho>0$ holds by Assumptions 2. The condition for mixing
	coefficients holds by Assumption 1. Furthermore, $var\left(\hat{d}_{t,kl}^{i*}\left(\theta\right)\right)+\sum_{m=1}^{\infty}m\left|Cov\left(\hat{d}_{t,kl}^{i*}\left(\theta\right),\hat{d}_{t+m,kl}^{i*}\left(\theta\right)\right)\right|<\infty$
	for all $\theta\in\Theta$. Thus by using Theorem 2 of \citet{PR_1994},
	\begin{eqnarray*}
		\sup_{\omega\in\mathbb{R}}\left|P\left(\sqrt{T_{P}}\left(\hat{D}_{kl,\alpha}^{i*}\left(\theta\right)-\hat{D}_{kl,\alpha}^{i}\left(\theta\right)\right)\leq\omega|W_{T_{R}},\ldots,W_{T-h}\right)\right.\\
		\left.-P\left(\sqrt{T_{P}}\left(\hat{D}_{kl,\alpha}^{i}\left(\theta\right)-D_{kl,\alpha}^{i}\left(\theta\right)\right)\leq\omega\right)\right| & \overset{p.}{\rightarrow} & 0
	\end{eqnarray*}
	for all $\theta\in\Theta$. Then by using continuous mapping theorem,
	it follows that 
	\begin{eqnarray*}
		\sup_{\omega\in\mathbb{R}}\left|P\left(\sqrt{T_{P}}\max_{k\neq l,k,l=1,\ldots,K}\sup_{\theta\in\Theta}\left(\hat{D}_{kl,\alpha}^{i*}\left(\theta\right)-\hat{D}_{kl,\alpha}^{i}\left(\theta\right)\right)\leq\omega|W_{T_{R}},\ldots,W_{T-h}\right)\right.\\
		\left.-P\left(\sqrt{T_{P}}\max_{k\neq l,k,l=1,\ldots,K}\sup_{\theta\in\Theta}\left(\hat{D}_{kl,\alpha}^{i}\left(\theta\right)-D_{kl,\alpha}^{i}\left(\theta\right)\right)\leq\omega\right)\right| & \overset{p.}{\rightarrow} & 0.
	\end{eqnarray*}
	For the second part of Theorem 2, let the asymptotic distribution of the test $\hat{S}_{T_{P},\alpha}^{i}$
	be
	\[ H^{i}\left(\omega\right)=P\left(\max_{\left(k,l\right)\in\mathcal{K}}\sup_{\theta\in\mathcal{A}_{kl}^{i}}\tilde{g}_{kl}^{i}\left(\theta\right)\leq\omega\right)
	\]
	for $i\in\left\{ E,Q\right\} $ and $\omega\in\mathbb{R}$. Since
	the Gaussian process $\tilde{g}_{kl}^{i}\left(\theta\right)$ has
	nonsingular covariance function and is finite, the distribution is
	absolutely continuous in $\omega\in\mathbb{R}$. We would like to
	show that the bootstrap distribution $\hat{H}_{M}^{i}\left(\omega\right)\stackrel{p.}{\rightarrow}H^{i}\left(\omega\right)$
	for all $\omega\in\mathbb{R}$ if (\ref{implicit_constraint}) holds.
	Let $H_{T_{P}}^{i}\left(\omega\right)=P\left(\hat{S}_{T_{P},\alpha}^{i}\leq\omega\right)$
	for $i\in\left\{ E,Q\right\} $. When (\ref{implicit_constraint})
	holds, it implies that $D_{kl,\alpha}^{i}=0$ for $k\neq l$, $k,l=1,\ldots,K$.
	Thus in the special situation, we have 
	\[
	\sup_{\omega\in\mathbb{R}}\left|P\left(\hat{S}_{c,T_{p},\alpha}\leq\omega|W_{T_{R}},\ldots,W_{T-h}\right)-P\left(\hat{S}_{T_{p},\alpha}\leq\omega\right)\right|\overset{p.}{\rightarrow}0
	\]
	as $T_{P}\rightarrow\infty$. Also $\hat{H}_{M}^{i}\left(\omega\right)\overset{p.}{\rightarrow}P\left(\hat{S}_{c,T_{p},\alpha}\leq\omega|W_{T_{R}},\ldots,W_{T-h}\right)$
	as $M\rightarrow\infty$. Therefore 
	\[
	\hat{H}_{M}^{i}\left(\omega\right)\overset{p.}{\rightarrow}P\left(\hat{S}_{T_{p},\alpha}\leq\omega\right)=H_{T_{P}}^{i}\left(\omega\right).
	\]
	for all $\omega\in\mathbb{R}$ as $M\rightarrow\infty$. Finally by
	Theorem 1, $H_{T_{P}}^{i}\left(\omega\right)\overset{p.}{\rightarrow}H^{i}\left(\omega\right)$
	as $T_{P}\rightarrow\infty$. Thus $\hat{H}_{M}^{i}\left(\omega\right)\overset{p.}{\rightarrow}H^{i}\left(\omega\right)$
	as $T_{P}$ and $M\rightarrow\infty$ and it follows that $\hat{h}_{M}^{i}\left(1-\gamma\right)\stackrel{p.}{\rightarrow}h^{i}\left(1-\gamma\right)$.
	Also
	\begin{eqnarray*}
		P\left(\hat{S}_{T_{p},\alpha}\geq\hat{h}_{M}^{i}\left(1-\gamma\right)\right) & = & P\left(\hat{S}_{T_{p},\alpha}\geq h^{i}\left(1-\gamma\right)+o_{p}\left(1\right)\right)\\
		& \rightarrow & P\left(\max_{\left(k,l\right)\in\mathcal{K}}\sup_{\theta\in\mathcal{A}_{kl}^{i}}\tilde{g}_{kl}^{i}\left(\theta\right)\geq h^{i}\left(1-\gamma\right)\right)\\
		& = & \gamma
	\end{eqnarray*}
	as $T_{P}$ and $M\rightarrow\infty$. Finally, if $S_{\alpha}^{i}>0$,
	$\hat{S}_{T_{p},\alpha}\rightarrow\infty$ as $T_{P}\rightarrow\infty$.
	By $\hat{h}_{M}^{i}\left(1-\gamma\right)=O_{p}\left(1\right)$ as
	$M\rightarrow\infty$, $P\left(\hat{S}_{T_{p},\alpha}\geq\hat{h}_{M}^{i}\left(1-\gamma\right)\right)\rightarrow1$
	as $T_{P}$ and $M\rightarrow\infty$.\end{proof}

\subsection{Implementing the stationary bootstrap of \citet{PR_1994}}
Let $W_{t}=\left(X_{1t},X_{2t},Y_{t+h}\right)$. By Assumption 1,
$W_{t}$ is a strictly stationary time series. To ease notations,
with loss of generality, here we will set $t=1,\ldots,T_{P}$ rather
than $t=T_{R},\ldots,T-h$ used in the main context. Let 
\[
B_{t,b}=\left(W_{t},W_{t+1},\ldots,W_{t+b-1}\right)
\]
be a block of $b$ observations from period $t$ to $t+b-1$. Let
$p\in\left[0,1\right]$ be a constant. Let $L_{1},L_{2},\ldots,$
be a sequence of i.i.d. random variables drawn from the geometric
distribution with density function $\left(1-p\right)^{m-1}p$ for
$m=1,2,\ldots$. Let $I_{1},I_{2},\ldots,$ be a sequence of i.i.d.
random variables drawn from the discrete uniform distribution on $\left\{ 1,\ldots,T_{R}\right\} $.
Note that here we require $L_{1},L_{2},\ldots,$ $I_{1},I_{2},\ldots,$
and $W_{t}$, $t=1,\ldots,T_{P}$ should be mutually independent.
Let $W_{1}^{*},W_{2}^{*},\ldots,W_{T_{P}}^{*}$ be a pseudo time series
generated by the stationary bootstrap of \citet{PR_1994}.
The procedures for implementing the stationary bootstrap are as follows. 
\begin{step} Sample a sequence of blocks with random lengths $B_{I_{1},L_{1}},B_{I_{2},L_{2}},\ldots.$ 
\end{step}
\begin{step} Combine the observations in $B_{I_{1},L_{1}},B_{I_{2},L_{2}},\ldots$
together as the pseudo time series $W_{1}^{*},W_{2}^{*},\ldots,W_{T_{P}}^{*}$
. So in the pseudo time series, the first $L_{1}$ observations are
$W_{I_{1}},W_{I_{1+1}},\ldots,W_{I_{1}+L_{1}-1}$, and the subsequent
$L_{2}$ observations (from the $(L_{1}+1)$th observation to the
$(L_{1}+L_{2})$th) are $W_{I_{2}},W_{I_{2+1}},\ldots,W_{I_{2}+L_{2}-1}$
and so on.
\end{step}
\begin{step} If length of the pseudo time series is greater than $T_{P}$, we
eliminate the extra observations to make length of the pseudo time series equal to $T_{P}$.\end{step}
\begin{step} Use the pseudo time series $W_{1}^{*},W_{2}^{*},\ldots,W_{T_{P}}^{*}$ to calculate the test statistic.\end{step}
\begin{step}
Repeat steps 1 to 4 independently $M$ times.
\end{step}
Note that if in a certain block, say $B_{I_{3},L_{3}}$, we have $I_{3}=T_{P}$
and $L_{3}=3$, then we will set $B_{I_{3},L_{3}}=\left(W_{T_{P}},W_{1},W_{2}\right)$.
That is, if in a certain block the last observation $W_{T_{P}}$ is
used, we will have the first observation $W_{1}$ to follow it. 

In the procedures, both the starting point and length of each block
are randomly determined (by $I_{1},I_{2},\ldots$ and $L_{1},L_{2},\ldots$).
The expected length of each block is $1/p$. For the choice of parameter
$p$, \citet{PR_1994} suggest that $p=p_{T_{P}}=\hat{C}_{T_{P}}T_{P}^{-1/3}$,
where $\hat{C}_{T_{P}}$ depends on the spectral density and might
be estimated consistently. Finally, our simulations are conducted
with \texttt{R} and the function we use to implement the stationary bootstrap is \texttt{tsboot}
in package \texttt{boot}.

\subsection{The size-power curves for the simulations}

To compare powers of a test statistic under different alternatives, it is ideal that the test statistic has a correct size, however, this is sometimes not easily achievable. For fairly demonstrating properties of power of the test statistic, we thus need to take the size effect into account. One of the statistical tools for this purpose is the size-power curve \citep{DM_1998}. 

A size-power curve is generated as follows. Let $\hat{p}_{0}$
and $\hat{p}_{1}$ denote the empirical p-values under the least favorable configuration and an alternative. We first calculate the empirical $\gamma-$quantile of $\hat{p}_{0}$: $\hat{q}_{\hat{p}_{0}}(\gamma):=\inf\left\{ x:\#\left\{ \hat{p}_{0}\leq x\right\} /N\geq\gamma\right\} $, where $N$ (here equals to 1000) is the number of simulations. In a simulation study, we say that the test statistic has a good size if $\hat{q}_{\hat{p}_{0}}(\gamma)$ is very similar to $\gamma$ for every $\gamma$. We then calculate the corresponding adjusted empirical power $\#\left\{ \hat{p}_{1}\leq\hat{q}_{\hat{p}_{0}}(\gamma)\right\} /N$.
The size-power curve is a set of points $\left(\gamma,\#\left\{ \hat{p}_{1}\leq\hat{q}_{\hat{p}_{0}}(\gamma)\right\} /N\right)$.
Ideally, in the least favorable configuration, the size-power curve
should be a 45 degree line. For two alternatives, say
$H_{1}$ and $H_{1}^{\prime}$, if $H_{1}$ deviates the null more than $H_{1}^{\prime}$ does, the test statistic should have more power under $H_{1}$ and the size-power curve for $H_{1}$ should lie above the size-power curve for $H_{1}^{\prime}$. For any alternative deviating from the null, ideally its size-power curve should lie above the 45 degree line. On contrary, if the hypothesis is deep in the null, its size-power curve should lie below the 45 degree line.

In Figures \ref{figure4} to \ref{figure6}, we plot size-power curves for models E1, E2 and E3 under different settings and lengths of generated forecasts (left: $T_{P}=100$, middle: $T_{P}=300$ and right: $T_{P}=1000$). 
In each plot, the x-axis is the empirical size and the y-axis is the corresponding adjusted empirical power. For model E1, Figure \ref{figure4} show
that the size-power curves for the two better competing forecasts $\mu_{t+1|t}+e^{Z}(\alpha)$
and $\mu_{t+1|t}+e^{Z}(\alpha)+\varsigma(\alpha)Z_{2t}$ consistently lie above the 45 degree line over different empirical sizes. As the length of generated forecast $T_{P}$ increases,
the size-power curves also shrink toward to the upper-left corner of the
plot, which suggests that power of the test statistic increases with $T_{P}$ after adjusted for the size effect. For the three worse competing forecasts, their size-power curves consistently lie below the 45 degree line. 

For model E2, as can be seen from Figure \ref{figure5}, in all settings, the size-power curves all lie above the 45 degree line. It also can be seen that the size-power curves for low $\beta_{2}$ (0.1) and low correlation between $W_{1t}$ and $W_{2t}$ (0.3) obviously lie below those for the other settings, which suggests that the proposed test statistic has a lower power under the two situations. As $T_{P}$ increases, power of the proposed test statistic for all settings becomes obviously better. For model E3, as can
be seen from Figure \ref{figure6}, all the size-power curves lie above the 45 degree line and shrink toward to the upper-left corner of the plot as $T_{p}$ increases, which suggest that power of the proposed test statistic gets improved as $T_{P}$ increases.

The size-power curve plots for models Q1 and Q2 are shown in Figures \ref{figure7} and \ref{figure8}. For model Q1, the size-power curves for the two better competing forecasts $\mu_{t+1|t}+\varPhi^{-1}\left(\alpha\right)$
and $\mu_{t+1|t}+\varPhi^{-1}\left(\alpha\right)+Z_{2t}$ consistently lie above the 45 degree line over
different empirical sizes. As the forecast length $T_{P}$ increases,
the size-power curves also shrink toward to the upper-left corner of the
plot, suggesting that power of the proposed test statistic increases with $T_{P}$ after adjusted for the size effect. For the three worse competing forecasts, their size-power curves consistently lie below the 45 degree line. 

For model Q2, the size-power curves for settings (2) to (5) are all lie above the 45 degree line and shrink toward to the upper-left corner of the plot, suggesting that power of the proposed test statistic increases with $T_{P}$. It also can be seen that the size-power curve for setting (2) obviously lies below those for the other three settings, which suggests that the proposed test statistic has a lower power for this case.

\subsection{The consistent loss function associated with the logistic regression estimation}
A interesting case of $\phi(x)$ of the consistent loss function for the $\alpha-$expectile
forecast in (3) is  
\begin{equation}
\phi\left(x\right)=\phi_{1}\left(x\right):=x\log x+\left(1-x\right)\log\left(1-x\right).\label{logistic}
\end{equation}
for $x\in\left[0,1\right]$. It is easy to see that $\lim_{x\rightarrow0}\phi_{1}\left(x\right)=\lim_{x\rightarrow1}\phi_{1}\left(x\right)=0$
and $\phi_{1}^{\prime\prime}\left(x\right)>0$ for $x\in\left[0,1\right]$. Let $L_{\alpha}^{E,1}\left(x,y\right)$ denote the consistent loss function associated with $\phi_{1}\left(x\right)$. 
Assume $Y\in\left\{ 0,1\right\} $. It can be shown that when $\alpha=1/2$,
the consistent loss function $L_{1/2}^{E,1}\left(x,Y\right)$ is proportional
to $-\log x$ if $Y=1$ and to $-\log\left(1-x\right)$ if $Y=0$.
To see this, note that by using the result in p.511 of \citet{EGJK_2016}, we
can have 
\[
L_{1/2}^{E,1}\left(x,Y\right)=\begin{cases}
\frac{1}{2}\left(\phi_{1}\left(1\right)+x\phi_{1}^{\prime}\left(x\right)-\phi_{1}\left(x\right)-\phi_{1}^{\prime}\left(x\right)\right) & \text{ if }Y=1,\\
\frac{1}{2}\left(x\phi_{1}^{\prime}\left(x\right)-\phi_{1}\left(x\right)\right) & \text{ if }Y=0.
\end{cases}
\]
If we let
\begin{eqnarray*}
	\phi_{1}\left(1\right)+x\phi_{1}^{\prime}\left(x\right)-\phi_{1}\left(x\right)-\phi_{1}^{\prime}\left(x\right) & = & -\log\left(x\right),\\
	x\phi_{1}^{\prime}\left(x\right)-\phi_{1}\left(x\right) & = & -\log\left(1-x\right),
\end{eqnarray*}
it yields $\phi_{1}^{\prime}\left(x\right)=\log\left(x/\left(1-x\right)\right)$ if $\lim_{x\rightarrow 1}\phi_{1}\left(x\right)=0$. It can be verified that $\phi_{1}\left(x\right)=x\log x+\left(1-x\right)\log\left(1-x\right)$.
The expectation of $L_{1/2}^{E,1}\left(x,Y\right)$ is a convex function
of $x$ and is related to the negative log likelihood in the logistic
regression estimation. Minimizing the expectation of $L_{1/2}^{E,1}\left(x,Y\right)$
yields the success probability (expectation of $Y$).

\subsection{Some mathematical derivations for Section 4.1.1}
The subsection provides some mathematic derivations for results used in Section 4.1.1. Suppose the data generating process for $Y_{t+1}$ is (\ref{sim_mse_expb}). The benchmark forecast $X_{1t}=c_{1}+b_{1}W_{1t}$ and the competing forecast $X_{2t}=c_{2}+b_{2}W_{2t}$. It can be shown that 
\begin{eqnarray*}
	E\left[\left(Y_{t+1}-X_{1t}\right)^{2}\right] & = & E\left[Y_{t+1}^{2}\right]+c_{1}^{2}+\left(b_{1}^{2}-2b_{1}\beta_{1}\right)E\left[W_{1t}^{2}\right]-2c_{1}\gamma,\\
	E\left[\left(Y_{t+1}-X_{2t}\right)^{2}\right] & = & E\left[Y_{t+1}^{2}\right]+c_{2}^{2}+\left(b_{2}^{2}-2b_{2}\beta_{2}\right)E\left[W_{2t}^{2}\right]-2c_{2}\gamma.
\end{eqnarray*}
Thus $E\left[\left(Y_{t+1}-X_{1t}\right)^{2}\right]=E\left[\left(Y_{t+1}-X_{2t}\right)^{2}\right]$ implies that  
\begin{equation}
c_{1}^{2}+\left(b_{1}^{2}-2b_{1}\beta_{1}\right)\sigma_{W_{1}}^{2}-2c_{1}\gamma=c_{2}^{2}+\left(b_{2}^{2}-2b_{2}\beta_{2}\right)\sigma_{W_{2}}^{2}-2c_{2}\gamma.
\label{eq_mse}
\end{equation}
It is not difficult to see that if we set $c_{1}=c_{2}=2\gamma$, $b_{1}=2\beta_{1}$
and $b_{2}=2\beta_{2}$, equality of (\ref{eq_mse}) will hold. 

Now consider
the exponential Bregman loss function 
\[
\frac{1}{a^{2}}\left[\exp\left(ay\right)-\exp\left(ax\right)\right]-\frac{1}{a}\exp\left(ax\right)\left(y-x\right).
\]
The difference between expectations of the exponential Bregman loss
function for $X_{1t}$ and $X_{2t}$ is 
\[
\frac{1}{a^{2}}E\left[\exp\left(aX_{2t}\right)-\exp\left(aX_{1t}\right)\right]-\frac{1}{a}\left(E\left[\exp\left(aX_{1t}\right)\left(Y-X_{1t}\right)\right]-E\left[\exp\left(aX_{2t}\right)\left(Y-X_{2t}\right)\right]\right),
\]
where 
\begin{eqnarray*}
	E\left[\exp\left(aX_{2t}\right)\right] & = & \exp\left(ac_{1}+\frac{a^{2}b_{1}^{2}\sigma_{W_{1}}^{2}}{2}\right),\\
	E\left[\exp\left(aX_{1t}\right)\right] & = & \exp\left(ac_{2}+\frac{a^{2}b_{2}^{2}\sigma_{W_{2}}^{2}}{2}\right),\\
	E\left[\exp\left(aX_{1t}\right)Y\right] & = & \exp\left(ac_{1}+\frac{a^{2}b_{1}^{2}\sigma_{W_{1}}^{2}}{2}\right)\left(\gamma+a\beta_{1}b_{1}\sigma_{W_{1}}^{2}\right),\\
	E\left[\exp\left(aX_{2t}\right)Y\right] & = & \exp\left(ac_{2}+\frac{a^{2}b_{2}^{2}\sigma_{W_{2}}^{2}}{2}\right)\left(\gamma+a\beta_{2}b_{2}\sigma_{W_{2}}^{2}\right),\\
	E\left[\exp\left(aX_{1t}\right)X_{1t}\right] & = & \exp\left(ac_{1}+\frac{a^{2}\beta_{1}^{2}\sigma_{W_{1}}^{2}}{2}\right)\left(c_{1}+ab_{1}^{2}\sigma_{W_{1}}^{2}\right),\\
	E\left[\exp\left(aX_{2t}\right)X_{2t}\right] & = & \exp\left(ac_{2}+\frac{a^{2}b_{2}^{2}\sigma_{W_{2}}^{2}}{2}\right)\left(c_{2}+ab_{2}^{2}\sigma_{W_{2}}^{2}\right).
\end{eqnarray*}
Now consider the extremal consistent loss function for the $\alpha-$expectile,
\[
L_{\theta,\alpha}^{E}\left(x,y\right)=\left|1\left\{ y<x\right\} -\alpha\right|\left[\left(y-\theta\right)_{+}-\left(x-\theta\right)_{+}-1\left\{ \theta<x\right\} \left(y-x\right)\right].
\]
Here we fix $\alpha=0.5$ for the conditional expectation forecast.
Then
\begin{small} 
\begin{eqnarray*}
	E\left[L_{\theta,0.5}^{E}\left(X_{1t},Y_{t+1}\right)\right]-E\left[L_{\theta,0.5}^{E}\left(X_{2t},Y_{t+1}\right)\right] & = & 0.5\left(E\left[1\left\{ \theta<X_{2t}\right\} \left(Y_{t+1}-\theta\right)\right]-E\left[1\left\{ \theta<X_{1t}\right\} \left(Y_{t+1}-\theta\right)\right]\right),
\end{eqnarray*}
\end{small}
where 
\begin{eqnarray*}
	E\left[1\left\{ \theta<X_{2t}\right\} \left(Y_{t+1}-\theta\right)\right] & = & \left(\gamma-\theta\right)\left(1-\varPhi\left(\frac{\theta-c_{2}}{b_{2}\sigma_{W_{2}}}\right)\right)+\beta_{2}\frac{1}{\sqrt{2\pi}\sigma_{W_{2}}}\int_{\frac{\theta-c_{2}}{b_{2}}}^{\infty}w\exp\left(-\frac{w^{2}}{2\sigma_{W_{2}}^{2}}\right)dw,\\
	E\left[1\left\{ \theta<X_{1t}\right\} \left(Y_{t+1}-\theta\right)\right] & = & \left(\gamma-\theta\right)\left(1-\varPhi\left(\frac{\theta-c_{1}}{b_{1}\sigma_{W_{1}}}\right)\right)+\beta_{1}\frac{1}{\sqrt{2\pi}\sigma_{W_{1}}}\int_{\frac{\theta-c_{1}}{b_{1}}}^{\infty}w\exp\left(-\frac{w^{2}}{2\sigma_{W_{1}}^{2}}\right)dw,
\end{eqnarray*}
and $\varPhi(.)$ is the cumulative distribution function of a standard normal random variable.  
\clearpage

\begin{figure}[ht]
	\begin{center}
		\mbox{
			\subfigure{\includegraphics[height=5.5cm,width=5.5cm]{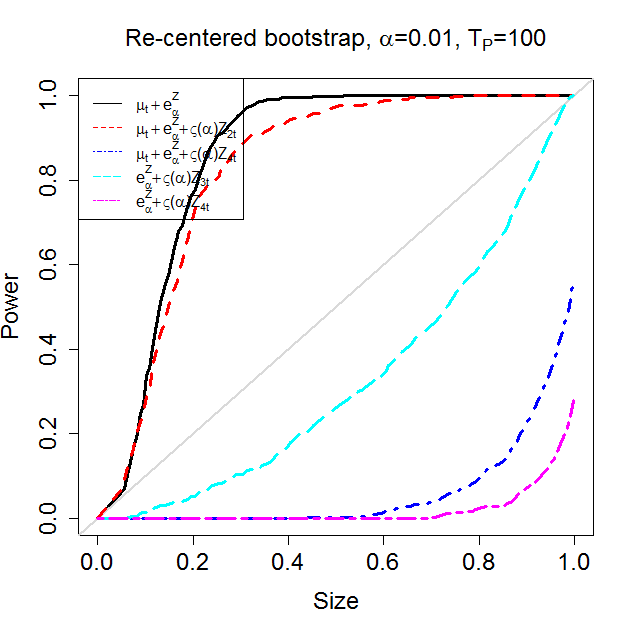}}
			\subfigure{\includegraphics[height=5.5cm,width=5.5cm]{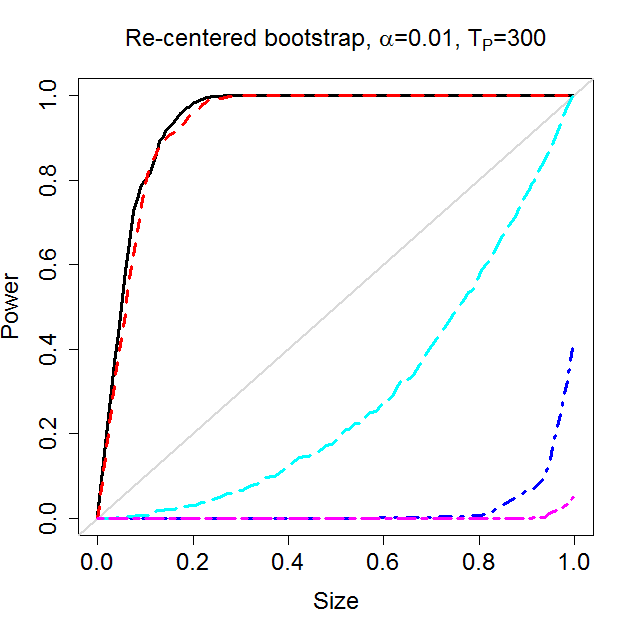}}
			\subfigure{\includegraphics[height=5.5cm,width=5.5cm]{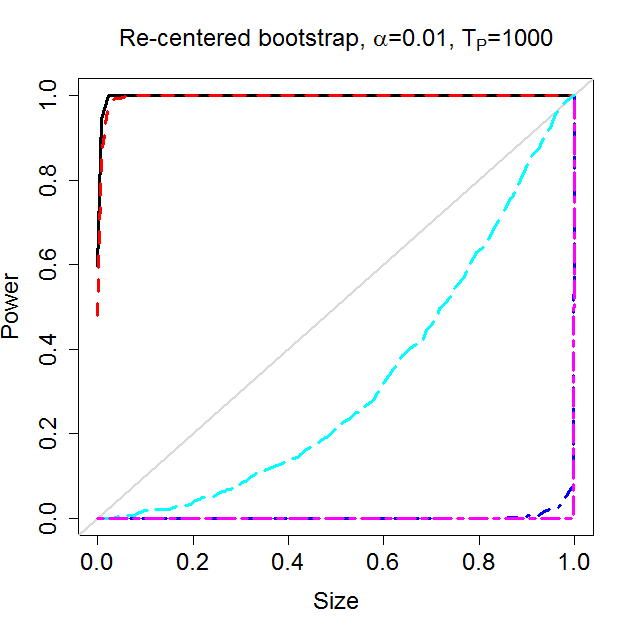}}
		} 
		\mbox{
			\subfigure{\includegraphics[height=5.5cm,width=5.5cm]{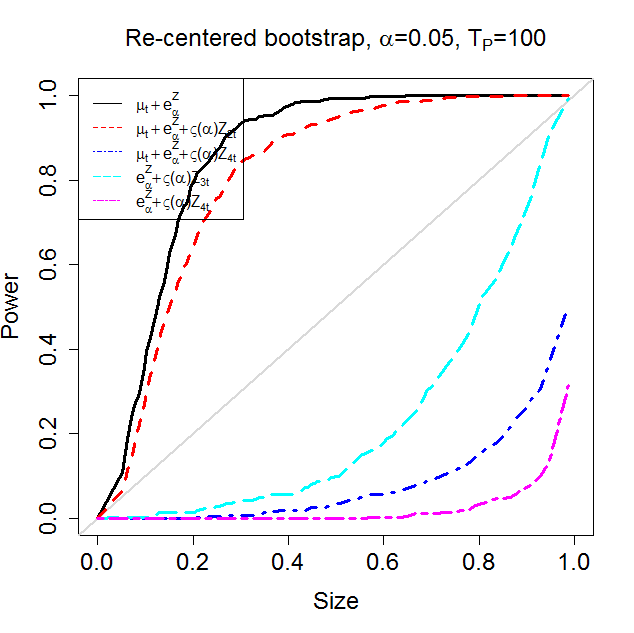}}
			\subfigure{\includegraphics[height=5.5cm,width=5.5cm]{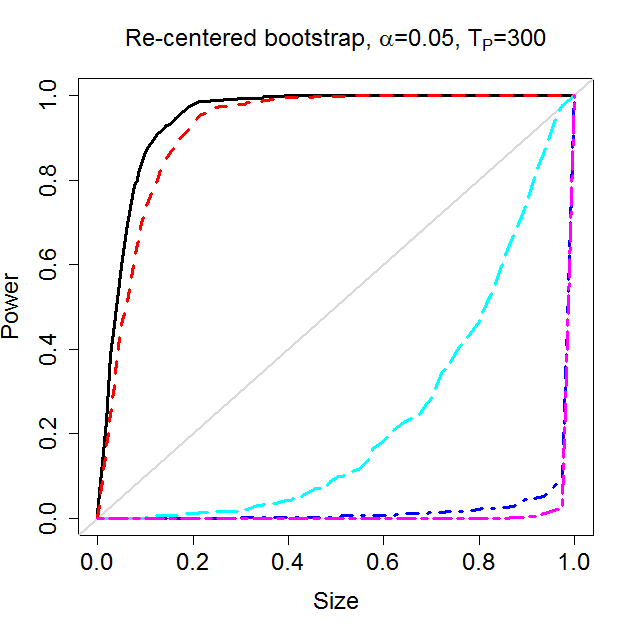}}
			\subfigure{\includegraphics[height=5.5cm,width=5.5cm]{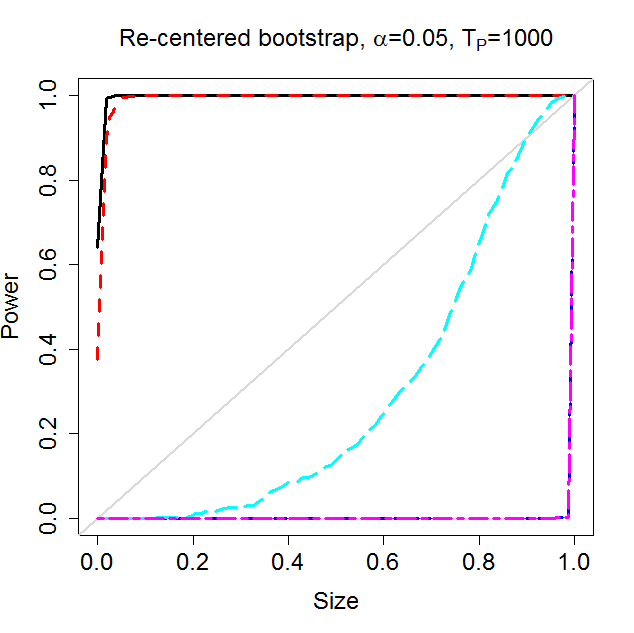}}
		} 
		\mbox{
			\subfigure{\includegraphics[height=5.5cm,width=5.5cm]{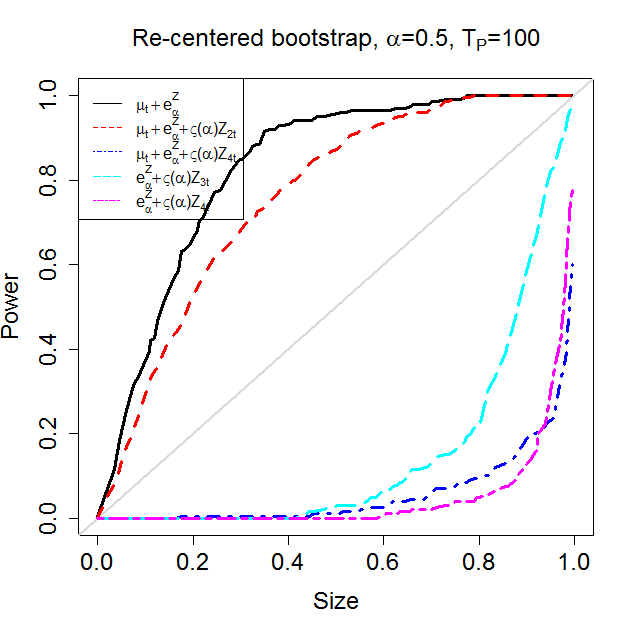}}
			\subfigure{\includegraphics[height=5.5cm,width=5.5cm]{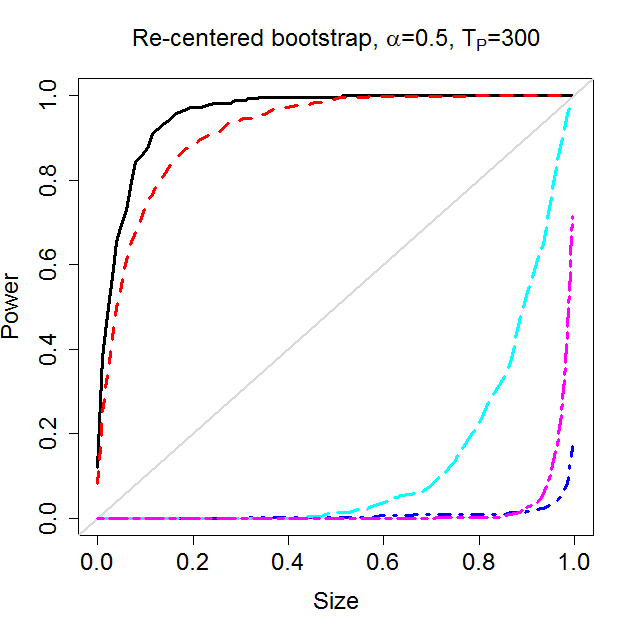}}
			\subfigure{\includegraphics[height=5.5cm,width=5.5cm]{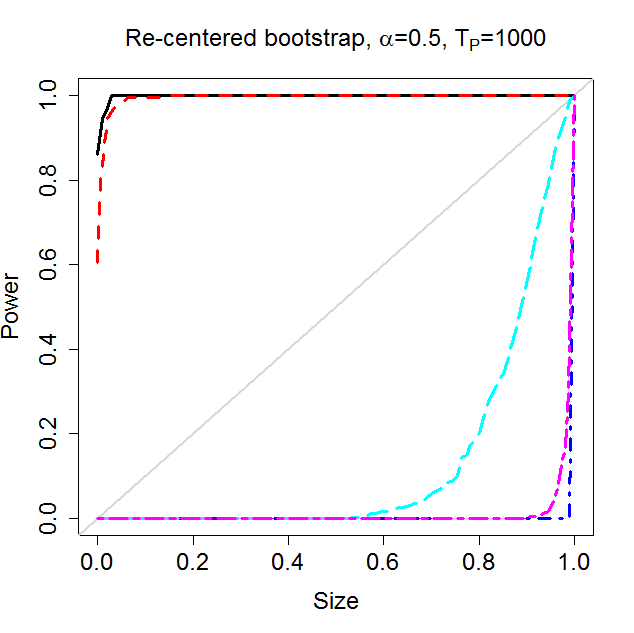}}
		} 
	\end{center}
	\caption{The figure shows the size-power curve (Davidson and MacKinnon, 1998) for simulation of model E1 under different settings. Upper panel: $\alpha=0.01$; middle panel: $\alpha=0.05$ and bottom panel: $\alpha=0.5$. Left: $T_{P}=100$; middle: $T_{P}=300$ and right: $T_{P}=1000$. In each plot, the x-axis is the empirical size and the y-axis is the corresponding adjusted empirical power.}
	\label{figure4}
\end{figure}

\begin{figure}[ht]
	\begin{center}
		\mbox{
			\subfigure{\includegraphics[height=5.5cm,width=5.5cm]{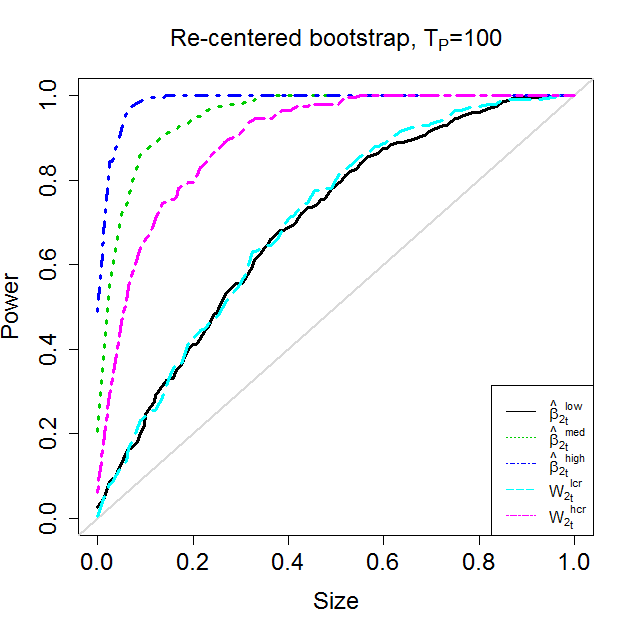}}
			\subfigure{\includegraphics[height=5.5cm,width=5.5cm]{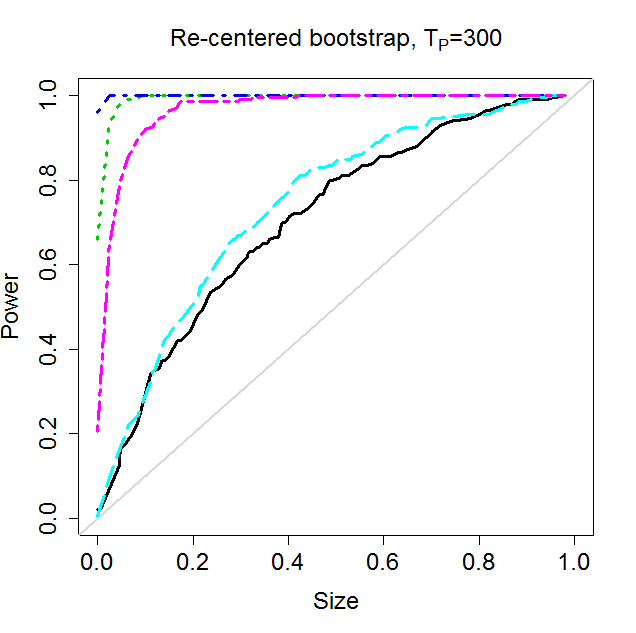}}
			\subfigure{\includegraphics[height=5.5cm,width=5.5cm]{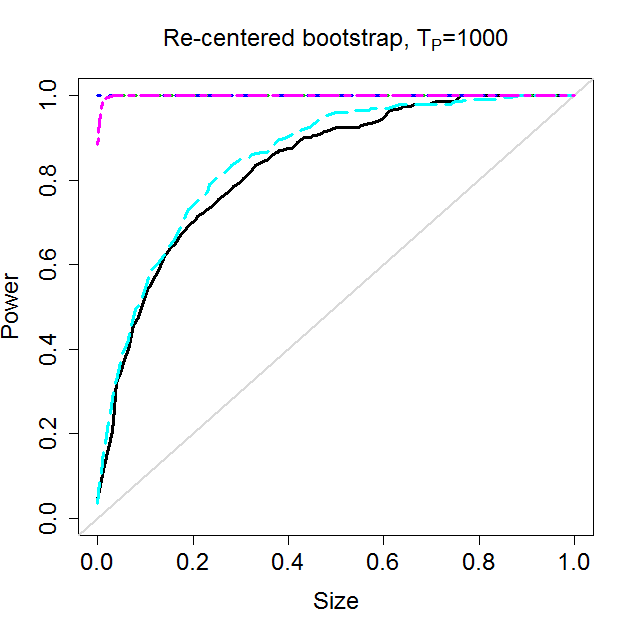}}
		} 
	\end{center}
	\caption{The figure shows the size-power curve (Davidson and MacKinnon, 1998) for simulation of model E2 under different settings. Left: $T_{P}=100$; middle: $T_{P}=300$ and right: $T_{P}=1000$. In each plot, the x-axis is the empirical size and the y-axis is the corresponding adjusted empirical power.}
	\label{figure5}
\end{figure}

\begin{figure}[ht]
	\begin{center}
		\mbox{
			\subfigure{\includegraphics[height=5.5cm,width=5.5cm]{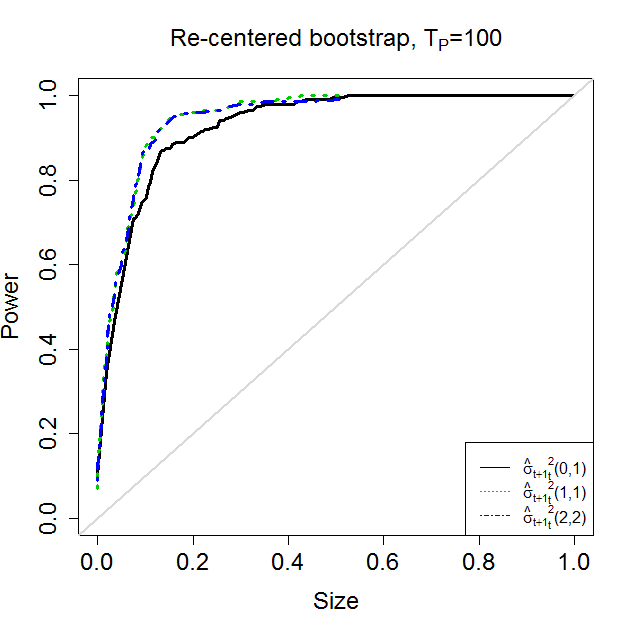}}
			\subfigure{\includegraphics[height=5.5cm,width=5.5cm]{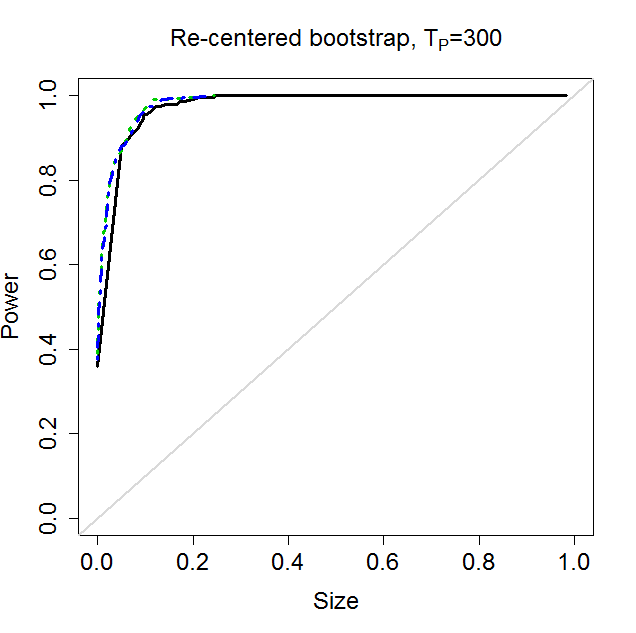}}
			\subfigure{\includegraphics[height=5.5cm,width=5.5cm]{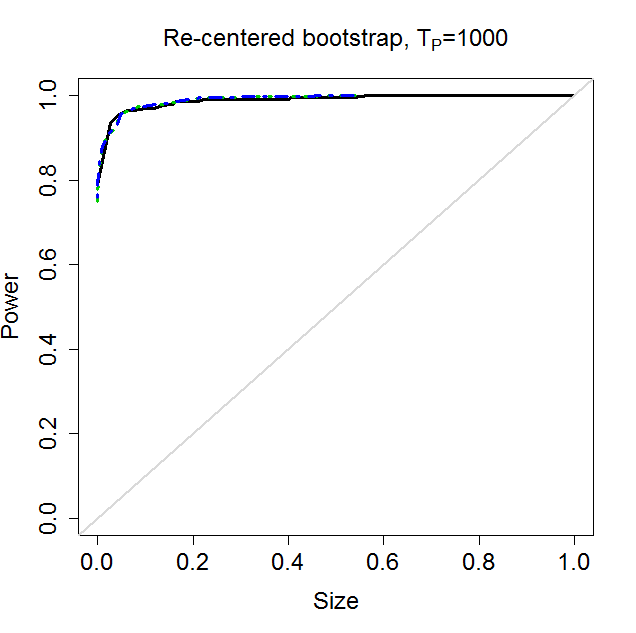}}
		} 
	\end{center}
	\caption{The figure shows the size-power curve (Davidson and MacKinnon, 1998) for simulation of model E3 under different settings. Left: $T_{P}=100$; middle: $T_{P}=300$ and right: $T_{P}=1000$. In each plot, the x-axis is the empirical size and the y-axis is the corresponding adjusted empirical power.}
	\label{figure6}
\end{figure}

\begin{figure}[ht]
	\begin{center}
		\mbox{
			\subfigure{\includegraphics[height=5.5cm,width=5.5cm]{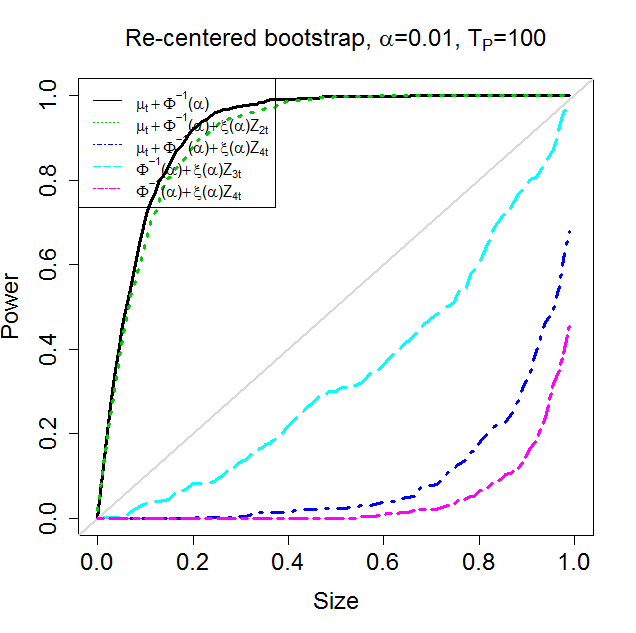}}
			\subfigure{\includegraphics[height=5.5cm,width=5.5cm]{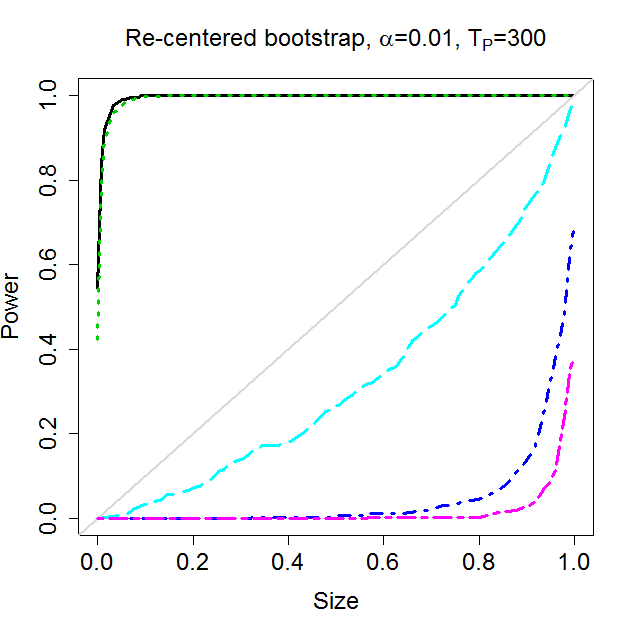}}
			\subfigure{\includegraphics[height=5.5cm,width=5.5cm]{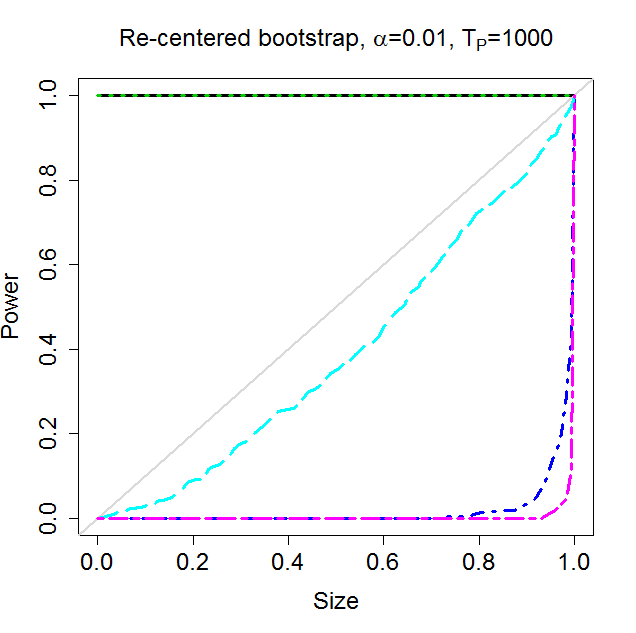}}
		} 
		\mbox{
			\subfigure{\includegraphics[height=5.5cm,width=5.5cm]{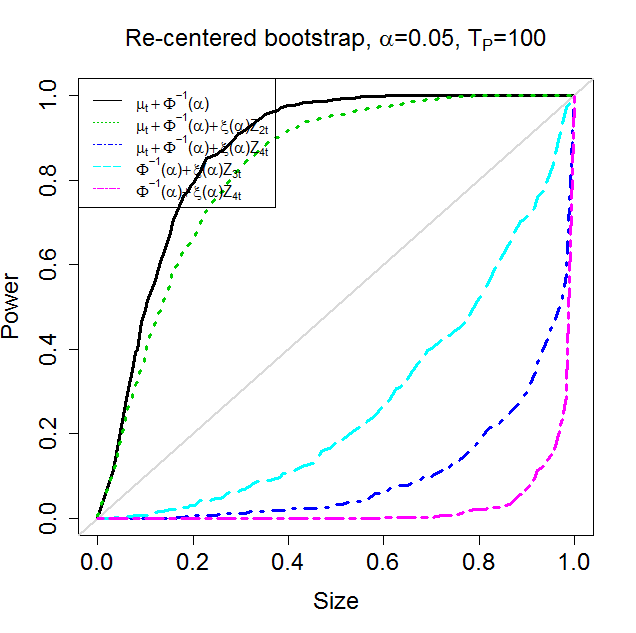}}
			\subfigure{\includegraphics[height=5.5cm,width=5.5cm]{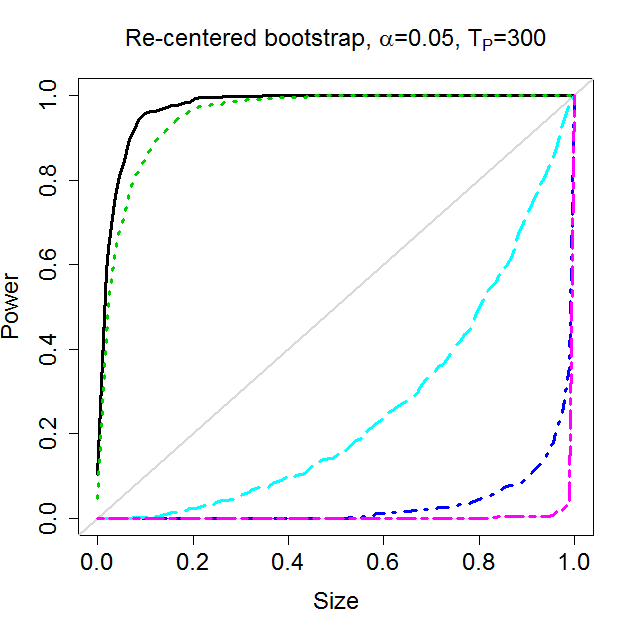}}
			\subfigure{\includegraphics[height=5.5cm,width=5.5cm]{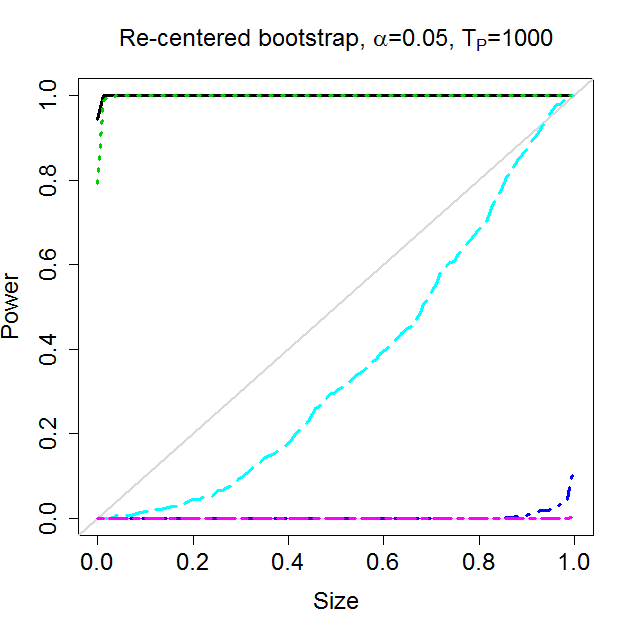}}
		} 
		\mbox{
			\subfigure{\includegraphics[height=5.5cm,width=5.5cm]{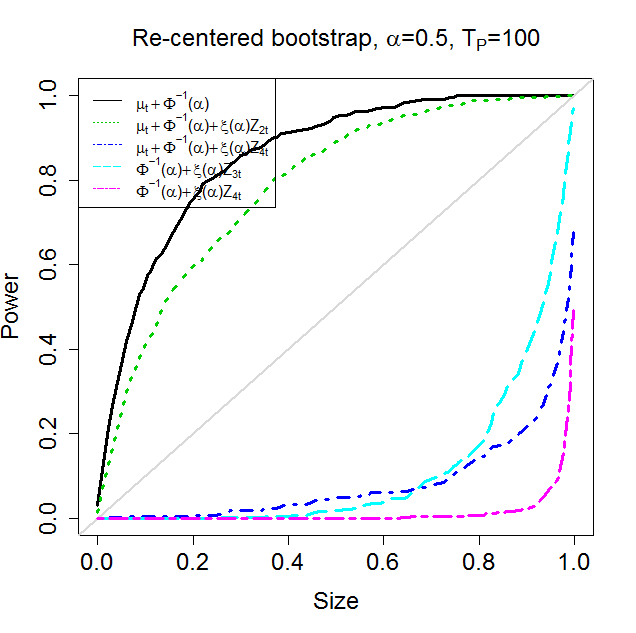}}
			\subfigure{\includegraphics[height=5.5cm,width=5.5cm]{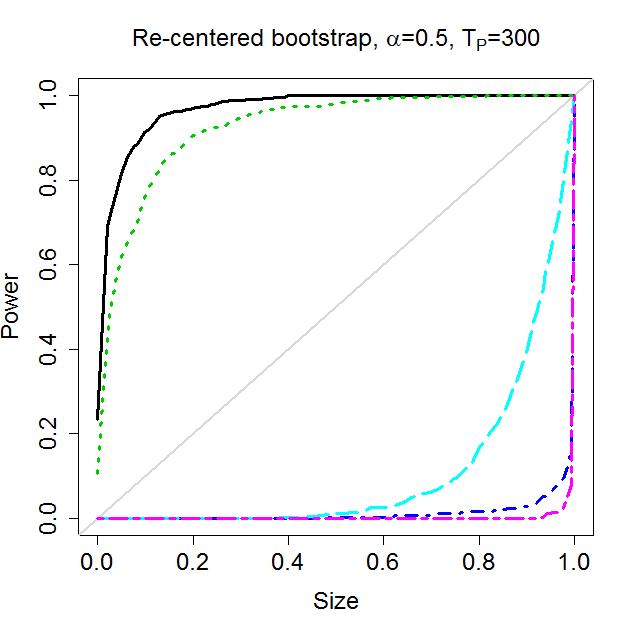}}
			\subfigure{\includegraphics[height=5.5cm,width=5.5cm]{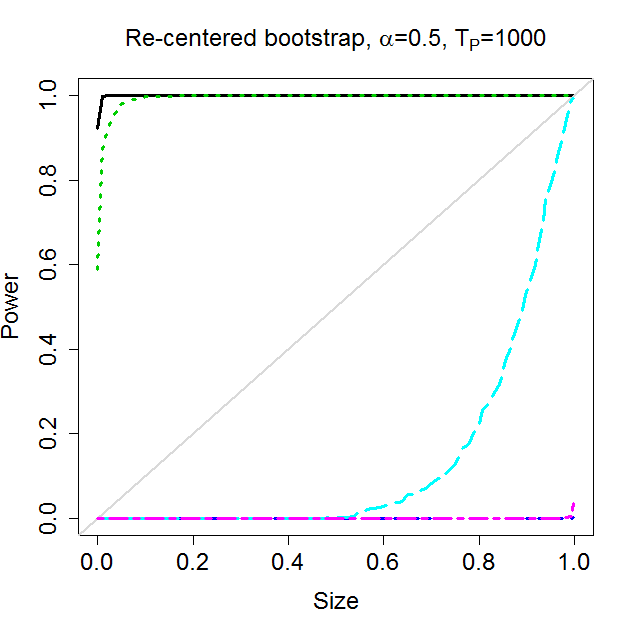}}
		}
	\end{center}
	\caption{The figure shows the size-power curve (Davidson and MacKinnon, 1998) for simulation of model Q1 under different settings. Upper panel: $\alpha=0.01$; middle panel: $\alpha=0.05$ and bottom panel: $\alpha=0.5$. Left: $T_{P}=100$; middle: $T_{P}=300$ and right: $T_{P}=1000$. In each plot, the x-axis is the empirical size and the y-axis is the corresponding adjusted empirical power.}
	\label{figure7}
\end{figure}

\begin{figure}[ht]
	\begin{center}
		\mbox{
			\subfigure{\includegraphics[height=5.5cm,width=5.5cm]{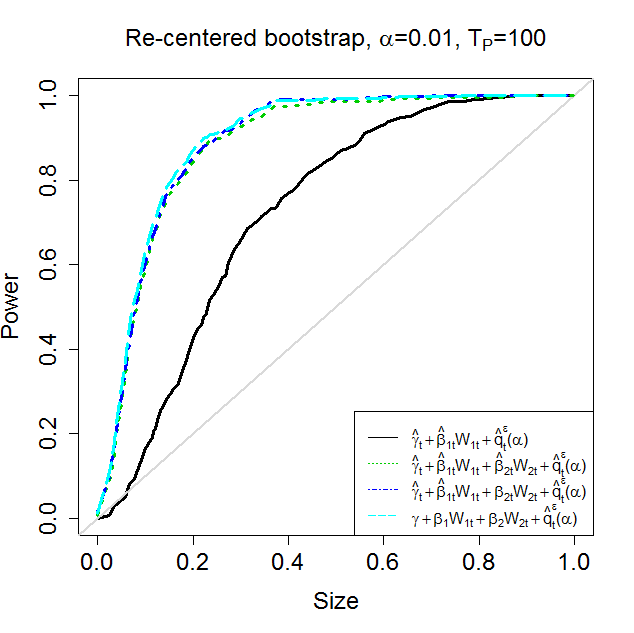}}
			\subfigure{\includegraphics[height=5.5cm,width=5.5cm]{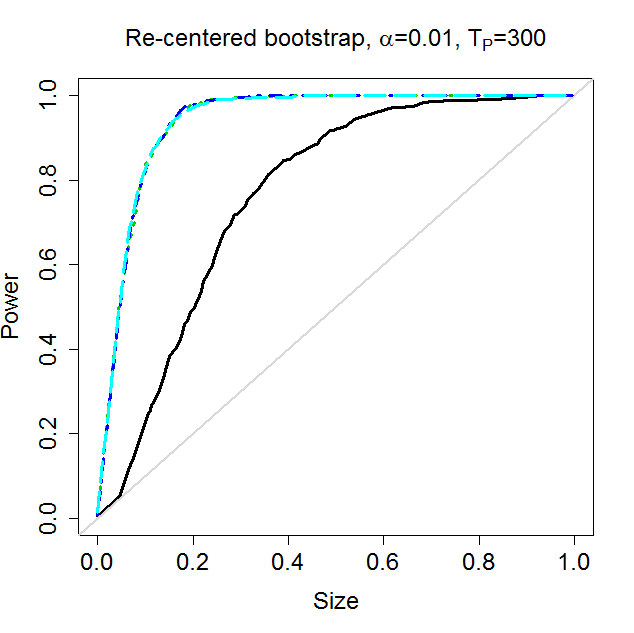}}
			\subfigure{\includegraphics[height=5.5cm,width=5.5cm]{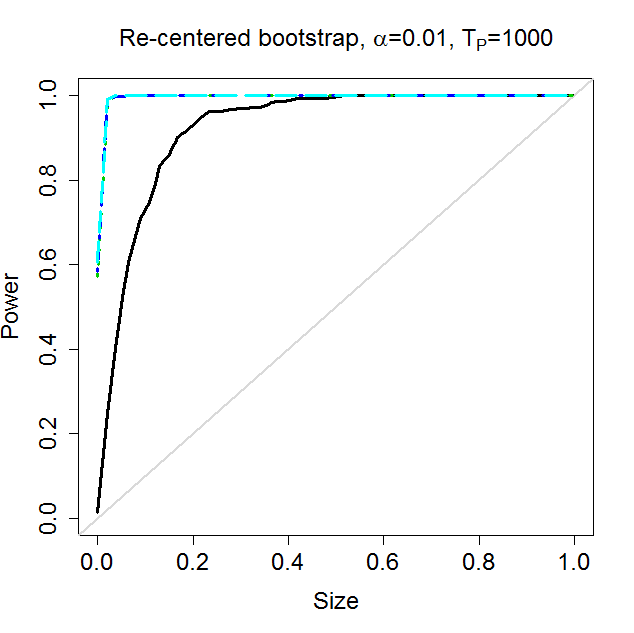}}
		} 
		
		\mbox{
			\subfigure{\includegraphics[height=5.5cm,width=5.5cm]{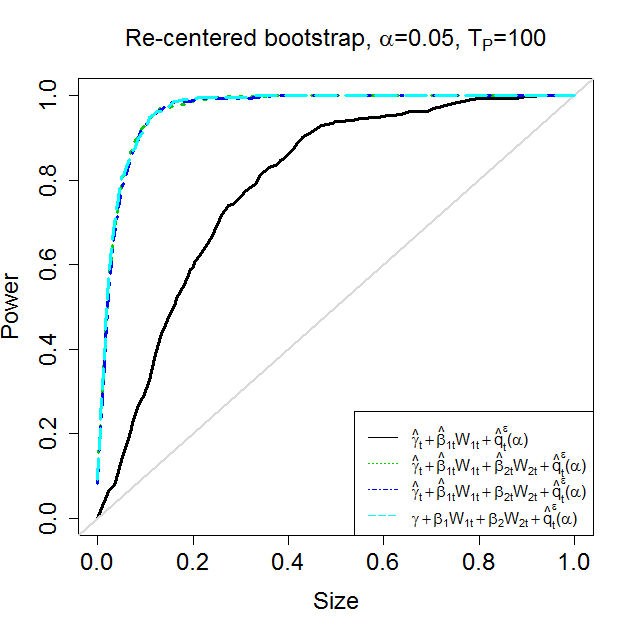}}
			\subfigure{\includegraphics[height=5.5cm,width=5.5cm]{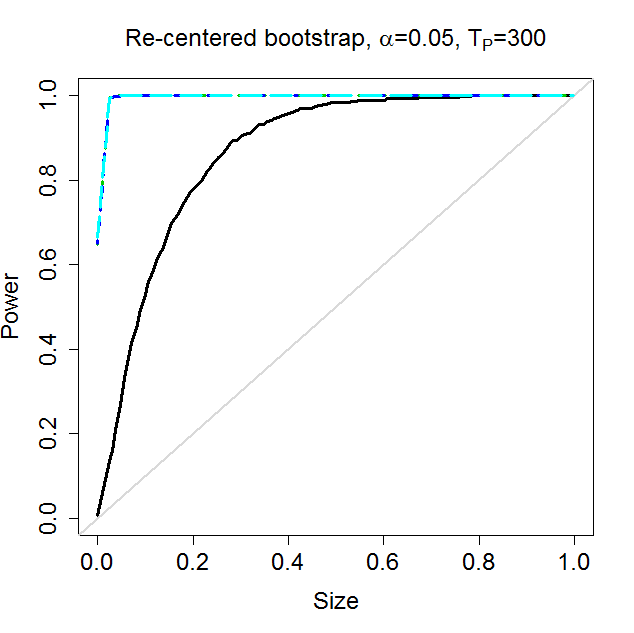}}
			\subfigure{\includegraphics[height=5.5cm,width=5.5cm]{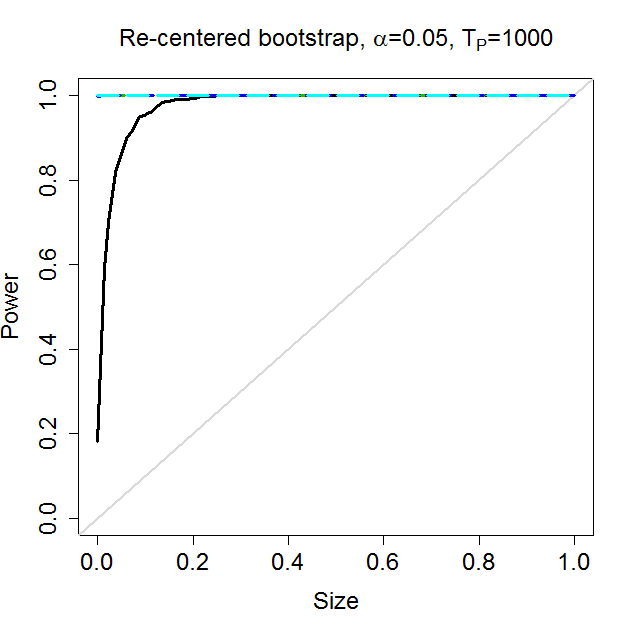}}
		} 
		
		\mbox{
			\subfigure{\includegraphics[height=5.5cm,width=5.5cm]{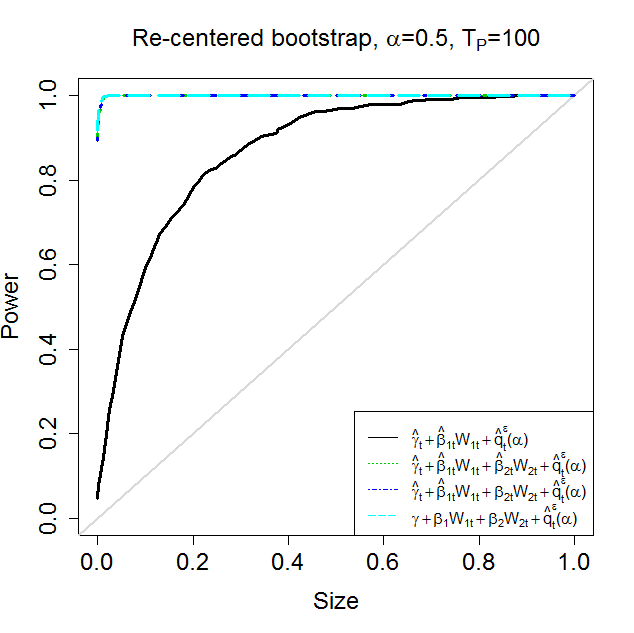}}
			\subfigure{\includegraphics[height=5.5cm,width=5.5cm]{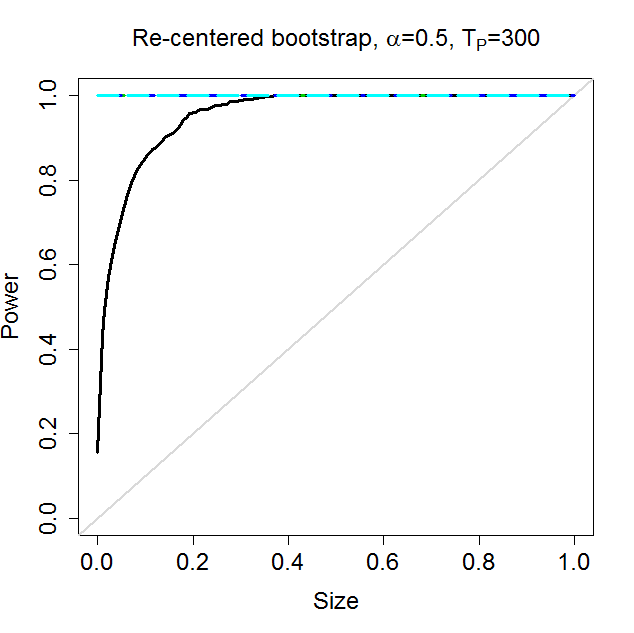}}
			\subfigure{\includegraphics[height=5.5cm,width=5.5cm]{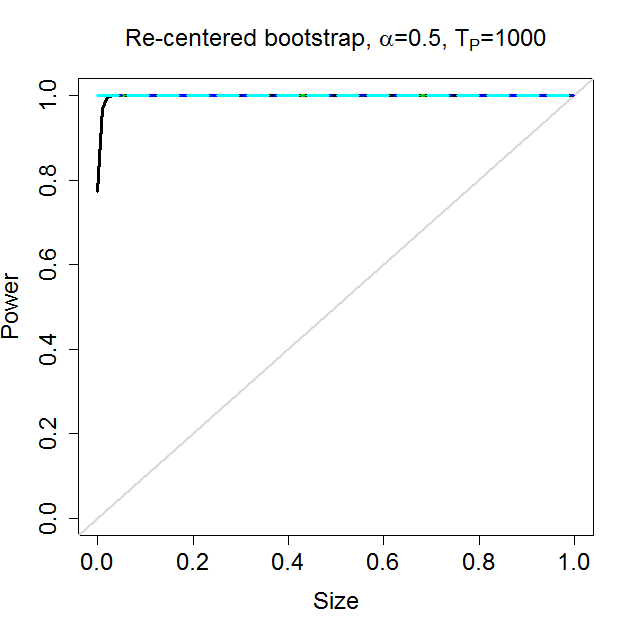}}
		}
		
	\end{center}
	\caption{The figure shows the size-power curve (Davidson and MacKinnon, 1998) for simulation of model Q2 under different settings. Upper panel: $\alpha=0.01$; middle panel: $\alpha=0.05$ and bottom panel: $\alpha=0.5$. Left: $T_{P}=100$; middle: $T_{P}=300$ and right: $T_{P}=1000$. In each plot, the x-axis is the empirical size and the y-axis is the corresponding adjusted empirical power.}
	\label{figure8}
\end{figure}
\end{document}